\documentclass{article}

\usepackage{amsmath,amssymb,amsthm}
\usepackage{fullpage}

\usepackage{xcolor,xspace}
\usepackage{graphicx}
\usepackage{caption}
\usepackage{booktabs}
\usepackage{microtype}
\usepackage{enumerate}
\usepackage{algorithm}
\usepackage[
    indLines=true,
    noEnd=true,
    rightComments=true,
    italicComments=true,
]{algpseudocodex}
\algrenewcommand\algorithmicrequire{\textbf{Input:}}
\algrenewcommand\algorithmicensure{\textbf{Output:}}

\usepackage{thmtools}
\usepackage{thm-restate}

\declaretheorem[name=Theorem]{theorem}
\declaretheorem[name=Lemma,sibling=theorem]{lemma}
\declaretheorem[name=Corollary,sibling=theorem]{corollary}
\declaretheorem[name=Observation,sibling=theorem]{observation}

\declaretheorem[name=Definition,sibling=theorem,style=definition]{definition}

\usepackage[
    colorlinks=true,
    linkcolor=blue!62!black,
    citecolor=green!48!black,
    urlcolor=blue!70!black,
    linktoc=page
]{hyperref}
\usepackage[nameinlink,noabbrev]{cleveref}

\crefname{theorem}{Theorem}{Theorems}
\Crefname{theorem}{Theorem}{Theorems}
\crefname{lemma}{Lemma}{Lemmas}
\Crefname{lemma}{Lemma}{Lemmas}
\crefname{corollary}{Corollary}{Corollaries}
\Crefname{corollary}{Corollary}{Corollaries}
\crefname{observation}{Observation}{Observations}
\Crefname{observation}{Observation}{Observations}
\crefname{proposition}{Proposition}{Propositions}
\Crefname{proposition}{Proposition}{Propositions}
\crefname{claim}{Claim}{Claims}
\Crefname{claim}{Claim}{Claims}
\crefname{conjecture}{Conjecture}{Conjectures}
\Crefname{conjecture}{Conjecture}{Conjectures}
\crefname{assumption}{Assumption}{Assumptions}
\Crefname{assumption}{Assumption}{Assumptions}
\crefname{definition}{Definition}{Definitions}
\Crefname{definition}{Definition}{Definitions}
\crefname{remark}{Remark}{Remarks}
\Crefname{remark}{Remark}{Remarks}
\crefname{algorithm}{Algorithm}{Algorithms}
\Crefname{algorithm}{Algorithm}{Algorithms}
\crefname{section}{Section}{Sections}
\Crefname{section}{Section}{Sections}
\crefname{appendix}{Appendix}{Appendices}
\Crefname{appendix}{Appendix}{Appendices}
\crefname{equation}{Equation}{Equations}
\Crefname{equation}{Equation}{Equations}

\tikzset{
  algpxIndentLine/.style={draw=black!100,very thin}
}

\title{Faster Minimum k-Cut II: Improved Randomized and Deterministic Algorithms for Weighted Graphs}

\author{
    Trevor Vaughn\thanks{Carnegie Mellon University. email: tnvaughn@cmu.edu}
}

\begin{document}

\maketitle

\begin{abstract}
The Minimum $k$-Cut problem asks for a minimum-weight set of edges whose
removal leaves an undirected weighted graph with at least $k$ connected
components. We consider only $k\ge3$. The fastest previous algorithm for
weighted graphs ran in
$n^{k-2}(m+n)(\log n)^{O(k^2)}$ randomized time~\cite{LV26}; on dense
graphs this is $\widetilde O_k(n^k)$, and for $k=3$ it gives an
$\widetilde O(nm)$-time algorithm.

We improve the dense-graph exponent from $k$ to $k-1$ and obtain the same
exponent deterministically. After aggregating parallel edges in $O(m)$
time, our randomized algorithm runs with high probability in
\begin{equation*}
k^{O(k^2)}n^{k-1}\log^2n
\end{equation*}
time, while our deterministic algorithm runs in
\begin{equation*}
k^{O(k^2)}n^{k-1}\log^{O(1)}n
\end{equation*}
time. In particular, weighted Minimum $3$-Cut can be solved in
$O(n^2\log^2 n)$ randomized time and in
$\widetilde O(n^2)$ deterministic time.

The algorithms have two main components.  First, we give a faster
algorithm for weighted Minimum $3$-Cut.  After handling optima with a very
small side and optima with two light sides, the remaining optimum has a
unique structured side.  Tree packing reduces its completion to a batched
collection of $2$-respecting cut problems.  Second, we reduce Minimum
$k$-Cut to Minimum $3$-Cut by enumerating a bounded family of light-cut
candidates and recursively completing either side of each candidate.  If
the enumeration produces too many cuts, then we can instead produce an optimum $k$-cut directly.  We derandomize
the $3$-cut algorithm using a deterministic near-minimum-cut skeleton,
and derandomize the reduction using a specialized $4$-cut algorithm using the skeleton,
the constructive light-cut bounds, and the deterministic spectral sparsifier of \cite{BSS12}.
\end{abstract}

\newpage

\tableofcontents

\newpage

\section{Introduction}

The Minimum $k$-Cut problem asks for a minimum-capacity set of edges whose
removal separates an undirected graph into at least $k$ connected
components. Throughout, edge capacities are positive integers with
polynomially bounded total capacity. We write $\lambda_k(G)$ for its
optimum value. The case
$k=2$ is global minimum cut, for which randomized and deterministic
near-linear-time algorithms are known~\cite{Kar00,HLRW24}.  For every fixed
$k$, Minimum $k$-Cut is polynomial-time solvable, but historically the
growth of the exponent with $k$ has been large.

Goldschmidt and Hochbaum~\cite{GH94} gave the first polynomial-time
algorithm for fixed $k$, a deterministic algorithm running in
$n^{O(k^2)}$ time.  Karger and Stein~\cite{KS96} subsequently obtained a
randomized contraction algorithm with running time
$\widetilde O(n^{2k-2})$.  On the deterministic side,
Thorup~\cite{Tho08} developed a tree-packing algorithm running in
$\widetilde O(mn^{2k-2})$ time.  Chekuri, Quanrud, and Xu~\cite{CQX19}
later used the dual of an LP relaxation to find a tree that
$(2k-3)$-respects an optimum cut, improving the deterministic running
time to $\widetilde O(mn^{2k-3})$.

More recent work substantially improved the dependence on $k$ in the
exponent.  Gupta, Lee, and Li~\cite{GLL18,GLL19} gave an
$n^{(1.981+o(1))k}$-time randomized algorithm for general weighted graphs
and, for polynomially bounded integral capacities, a deterministic
$k^{O(k)}n^{(2\omega/3+o(1))k}$-time algorithm.  Gupta, Harris, Lee, and
Li~\cite{GHLL22} then proved nearly tight bounds on the number of minimum
and near-minimum $k$-cuts and obtained a randomized
$n^k(\log n)^{O(k^2)}$-time algorithm.

Previous work gave a fine-grained reduction from Max-Weight
$(k-1)$-Clique to weighted Minimum $k$-Cut~\cite{GLL18,GHLL22,HL22}.
Under the then-standard Max-Weight Clique conjecture, this gave an
$n^{k-1-o(1)}$ conditional lower bound and made $n^{k-1+o(1)}$ the
natural target. Very recently, however, Alman and Vassilevska
Williams~\cite{AVW26} gave truly faster algorithms for Min- and
Max-Weight $s$-Clique with polynomially bounded integer weights,
refuting this conjecture. Thus the previous conditional lower bound no
longer applies.

In our previous work~\cite{LV26}, we combined the structural bounds on
light cuts with tree packing to obtain a randomized weighted Minimum
$k$-Cut algorithm running in
$n^{k-2}(m+n)(\log n)^{O(k^2)}$ time. On sufficiently sparse graphs this
already has exponent $k-1$, but it becomes $\widetilde O_k(n^k)$ when
$m=\Theta(n^2)$. For $k=3$, it gives an
$\widetilde O(nm)$-time algorithm, leaving a gap between the cubic
worst-case bound and the natural quadratic target. That work also gave a
deterministic $k^{O(k^2)}n^{k+O(1)}$-time algorithm. The present paper
improves the dense randomized exponent from $k$ to $k-1$ and brings the
deterministic exponent to the same value.

\paragraph{Our results.}
Our first result is a faster algorithm for weighted Minimum
$3$-Cut.  It runs with high probability in
$O(n^2 \log^2 n)$ time.  We also give a deterministic algorithm running
in $\widetilde O(n^2)$ time.

We then reduce Minimum $k$-Cut to Minimum $3$-Cut.  Combined with the
new base case, the reduction gives the following bounds.

\begin{theorem}[Informal]
\label{thm:informal_main}
For every fixed $k\ge3$, let $G$ be an $n$-vertex, $m$-edge
weighted multigraph. After $O(m)$-time aggregation of parallel edges,
Minimum $k$-Cut can be solved with high probability in
$k^{O(k^2)}n^{k-1}\log^2 n$ time, or deterministically in
$k^{O(k^2)}n^{k-1}\log^{O(1)}n$ time.
\end{theorem}

\paragraph{A faster algorithm for $3$-cut.}
Let $(A,B,C)$ be an optimum $3$-cut.  Since the boundaries of its three
sides sum to $2\lambda_3$, at least one side has boundary at most
$2\lambda_3/3$.  We call such sides light.  We first handle two easier
configurations.  If an optimum has a side of boundary below
$\lambda_3/2$, that cut is the unique bipartition of this size (called strict-small); a
global minimum cut identifies it, after which one minimum-cut computation
completes the optimum.  If an optimum has two light sides, then either the
corresponding rooted cuts are represented in a laminar family and can be
combined directly, or two relevant cuts cross and immediately certify an
optimum $3$-cut.

To obtain the relevant light cuts, we compute an approximately optimal fractional tree
packing and sample $O(\log n)$ spanning trees from its distribution.  A
fixed light optimum side is $2$-respected by a sampled tree with constant
probability.  For each sampled tree, we enumerate all cuts obtained by
deleting at most two tree edges and retain only its $O(n)$ lightest
candidates.  A capacity matrix in DFS order evaluates each candidate in
constant time after $O(m+n^2)$ preprocessing for the tree.  We then
deduplicate the retained candidates and keep the longest laminar prefix of
the globally lightest cuts.  If a relevant optimum side is lost during
either truncation, the retained cuts themselves contain a crossing
certificate for an optimum $3$-cut.

It remains to complete a unique medium side $A$.  Performing an
independent global minimum-cut computation in $G[V\setminus A]$ for every
candidate $A$ would cost $\widetilde O(nm)$.  Instead, a sample
from the tree packing gives a controlled representation of the three
optimum sides by the tree edges crossing them.  We classify the possible
crossing patterns (how many tree edges are crossed by each side).
The patterns $(1,3,4)$ and $(1,4,3)$, two of the hard cases, can be avoided with constant probability.
We reduce the remaining patterns to finding a minimum cut that
$2$-respects a residual tree or to direct enumeration.
In the most difficult such pattern, $(2,3,3)$, we add a repair edge joining two
components of $T\setminus A$.  We show that a constant fraction of the
eligible edge capacity has both endpoints in the same residual optimum
side.  We deterministically construct a constant-size net containing a
valid repair edge.  After adding the repair edge, the residual optimum $2$-respects the
repaired tree.

The final optimization uses the $2$-respecting cut algorithm
of~\cite{GMW20b}.  Its graph-dependent preprocessing can be performed once
for a completion tree, while the remaining optimization is repeated for
all $O(n)$ candidate sides.  This separation is what permits the
$\widetilde O(m+n^2)$ batched running time.

\paragraph{Reducing $k$-cut to $3$-cut.}
Let $(A_1,\ldots,A_k)$ be an optimum $k$-cut.  Because the sum of its
side boundaries is $2\lambda_k$, some side $A_i$ satisfies
$c_G(A_i)\le2\lambda_k/k$.  We call such a cut light.  The reduction
constructs a bounded family of candidate cuts (oriented away from a root) and, for each
candidate $A$, recursively computes a minimum $(k-1)$-cut in both
$G[A]$ and $G[V\setminus A]$.  If the rooted orientation of $A_i$
appears in the family, one of these two recursive calls recovers the other
$k-1$ optimum parts.

For the randomized reduction, we sample spanning trees from the dual
packing distribution.  A light optimum side is $(k-1)$-respected by one
sample with probability $\Omega(1/k)$.  For each sampled tree we
enumerate the $(k-1)$-respecting cuts, but retain only the
$C_kn+1$ lightest, where $C_k=k^{O(k)}$ is obtained from the
constructive light-cut bound.  We then retain the same number of globally
lightest distinct cuts.

This truncation does not require knowing $\lambda_k$.  If the relevant
optimum side is discarded, the retained family is full and every retained
cut is no heavier.  If $B$ is the largest retained boundary value, then
$kB/2\le\lambda_k$.  The constructive bounds
turn the overflowing family into a $k$-cut of value at most $kB/2$,
which must therefore be optimum.  These arguments algorithmize the
near-minimum-cut bounds of~\cite{GHLL22} and avoid the perturbation argument
used in our previous work.

Only $k^{O(k)}n$ candidates are recursively completed.  The parameter
drops from $k$ to $k-1$ at every recursive level, so expanding the
recurrence down to the new $3$-cut base case gives
$k^{O(k^2)}n^{k-3}(m+n^2)\log^2 n$ time.  All randomness is confined to
sampling from explicitly computed tree-packing distributions.

\paragraph{Derandomization.}
We derandomize the $3$-cut algorithm using a deterministic
near-minimum-cut skeleton based on~\cite{HLRW24}.  After the strict-small
and two-light-side cases have been removed, all cuts needed by the
completion algorithm lie in a fixed range above the global
minimum cut.  The skeleton preserves the required upper bounds for these
cuts and a lower bound for every cut.  A packing of only
$\log^{O(1)}n$ edge-disjoint spanning trees in the skeleton
deterministically covers every relevant light optimum side and contains a
tree with the crossing pattern required by the completion algorithm.
Processing this entire packing replaces both random tree samples and gives
the $\widetilde O(n^2)$ deterministic bound.

The reduction requires a different construction.  For $k=4$, we reuse
both graphs produced by the deterministic near-minimum-cut construction.
If the global minimum cut of the current quotient has value at least one
third of the optimum $4$-cut, then a packing of only
$\log^{O(1)}n$ trees in the final skeleton contains a tree that
$3$-respects a light optimum side.  Enumerating the
$3$-respecting cuts of these trees costs $\widetilde O(n^3)$.

Otherwise, unless the optimum already refines the current minimum cut, the
intermediate graph exposes a polylogarithmic-size set of contractions, one
of which joins two vertices in the same optimum part.  Branching over these
contractions therefore preserves an optimum in one child.  Each safe
contraction destroys a distinct cut of value below
$\lambda_4/3$, and the strict-small-cut bound of~\cite{GHLL22}
limits such a path to three contractions.  This gives a 
deterministic $O(m)+\widetilde O(n^3)$-time $4$-cut algorithm.

For $k \ge 5$ we use a deterministic spectral sparsifier and
inspect the full support of its dual tree packing. For every relevant light cut, some support tree $(k-2)$-respects that cut. Although the support has $O(n)$ trees, enumerating the $(k-2)$-respecting cuts of
each tree takes $n^{k-2}$ time, giving $n^{k-1}$ total enumeration time.
Together with the same constructive overflow certificate and recursive
completion, this yields the deterministic bound in
\cref{thm:informal_main}.

\paragraph{Organization.}
\cref{sec:faster_three_cut} gives the faster randomized
algorithm for weighted Minimum $3$-Cut.
\cref{sec:reduction_kcut} develops the constructive light-cut
certificates and the randomized reduction from Minimum $k$-Cut to Minimum
$3$-Cut.  \cref{sec:derandomizing_three_cut} derandomizes the
$3$-cut algorithm using the near-minimum-cut skeleton.
\cref{sec:derandomizing_reduction_overview} derandomizes the
reduction, treating $k=4$ separately and using deterministic spectral
sparsification for $k\ge5$.  The constant-adjusted version of the skeleton construction is
proved in \cref{app:skeleton_construction}. We apply the weighted algorithm to give faster algorithms on simple graphs for $k \in \{3,4\}$ in \cref{app:simple_graph_consequences}.

\section{Preliminaries}

\paragraph{Conventions.}
All graphs are finite, undirected, capacitated multigraphs unless stated
otherwise.  Parallel edges are allowed.  Self-loops may arise under
contraction but are discarded, since they cross no cut.  For a graph
$G=(V,E,c)$, write $n:=|V|$ and $m:=|E|$.  Edge capacities are positive
integers; for $F\subseteq E$, write $c(F):=\sum_{e\in F}c(e)$.
Unless stated otherwise, the total capacity $C:=c(E)$ is bounded by
$n^{O(1)}$. Auxiliary reweighted graphs may
have exact rational capacities of polynomial bit complexity.

We use $\widetilde O(\cdot)$ to suppress factors polylogarithmic in $n$.
An event holds with high probability if it holds with probability
$1-n^{-\Omega(1)}$; the constant in the exponent can be increased by
changing the constants in the relevant sampling bounds.

\paragraph{Graph and cut notation.}
For a function $z$ on a finite set and a subset $X$ of its domain, write
$z(X):=\sum_{x\in X}z(x)$.

For $S\subseteq V$, define its edge boundary and boundary value by
$\delta_G(S):=\{uv\in E:|\{u,v\}\cap S|=1\}$ and
$c_G(S):=c(\delta_G(S))$, respectively. We identify $S$ with the cut
$(S,V\setminus S)$ when no ambiguity can arise.

For a partition $\mathcal P=(P_1,\ldots,P_t)$ of $V$, let
$\delta_G(\mathcal P)$ be the set of edges whose endpoints lie in
different parts, and let $c_G(\mathcal P):=c(\delta_G(\mathcal P))$.
For pairwise disjoint vertex sets $A,B\subseteq V$, write
$E_G(A,B):=\{uv\in E:u\in A,\ v\in B\}$.  Also write
$E_G(A):=\{uv\in E:u,v\in A\}$ and
$E_G(v,A):=E_G(\{v\},A)$.

Let $G[S]$ denote the subgraph induced by $S$, and let
$\operatorname{cc}(G)$ denote the number of connected components of
$G$.  For a nonempty $S\subseteq V$, the quotient $G/S$ is obtained
by contracting $S$ to one vertex, deleting the resulting self-loops, and
retaining parallel edges.  More generally, a quotient may contract several
pairwise disjoint vertex sets.  Every vertex of the quotient represents
the corresponding subset of original vertices, so cuts and partitions of
the quotient lift canonically to $G$.

\paragraph{Sets, cuts, and partitions.}
A family $\mathcal L$ of vertex subsets is laminar if, for every
$A,B\in\mathcal L$, either $A\cap B=\emptyset$, $A\subseteq B$, or
$B\subseteq A$.  Two sets cross if none of these relations holds.

The two members of a cut $(S,V\setminus S)$ are its sides.  Its shore is
the smaller side, with an arbitrary fixed tie-breaking rule when the two
sides have equal cardinality.

When a root $r\in V$ is fixed, the rooted orientation of the cut
$(S,V\setminus S)$ is the unique side not containing $r$.  We refer to
this side as the rooted cut.  Whenever cuts from
different rooted trees are compared, all trees use the same root.

A partition $\mathcal P$ refines a partition $\mathcal Q$ if every part
of $\mathcal P$ is contained in a part of $\mathcal Q$; equivalently,
$\mathcal Q$ is a coarsening of $\mathcal P$.

\paragraph{Minimum $k$-cut.}
A $k$-cut is a partition $\mathcal P$ of $V$ into exactly $k$
nonempty parts; its cut set is $\delta_G(\mathcal P)$.  Equivalently, the
minimum $k$-cut value is the minimum capacity of an edge set whose removal
leaves at least $k$ connected components. The minimum value of a $k$-cut of $G$ is denoted by
$\lambda_k(G)$, or simply by $\lambda_k$ when the graph is clear.  We
use the conventions $\lambda_1(G):=0$ for every nonempty graph and
$\lambda_j(G):=+\infty$ when $|V(G)|<j$.

The parts of a $k$-cut are called its sides.  If
$\mathcal P=(P_1,\ldots,P_k)$, then
$c_G(P_1)+\cdots+c_G(P_k)=2c_G(\mathcal P)$, since every edge crossing
the partition is counted in the boundaries of exactly two sides.

\paragraph{Respecting a tree.}
An edge set $F$ $p$-respects a tree $T$ if
$|F\cap E(T)|\le p$.  A cut or partition $p$-respects $T$ when its
boundary does. When $T$ is rooted, deleting the edges
$\delta_T(S)$ represents the rooted orientation of $S$ as the symmetric
difference of the rooted subtrees below those edges.

\begin{definition}[Cut-value regimes]
\label{def:cut_sizes}
Fix an integer $k\ge3$ and a value $U>0$.  A cut $S\subsetneq V$ is
$U$-light if $c_G(S)\le2U/k$.  It is $U$-small if
$c_G(S)\le U/(k-1)$, and it is $U$-medium if
$U/(k-1)<c_G(S)\le2U/k$.  The strict versions replace the corresponding
upper inequality by a strict inequality; in particular, a
$U$-strict-medium cut satisfies
$U/(k-1)<c_G(S)<2U/k$.

When $U=\lambda_k(G)$, we omit $U$ and simply say light, small, medium,
strict-light, strict-small, or strict-medium.
\end{definition}

\section{A Faster Algorithm for Weighted Minimum \texorpdfstring{$3$}{3}-Cut}
\label{sec:faster_three_cut}

Throughout this section, let $G=(V,E,c)$ be a connected undirected
capacitated multigraph with $n:=|V|\ge3$ and $m:=|E|$, and write
$\lambda:=\lambda_3(G)$.  Capacities are positive.
We use the $O(m\log^2 n)$-time randomized global-minimum-cut
algorithm of Gawrychowski, Mozes, and Weimann~\cite{GMW20a}.

The algorithm constructs a collection of explicit $3$-cuts and returns the
lightest one.  Randomness is used by the global-minimum-cut algorithm in the
strict-small preprocessing, to construct a cut skeleton, and to sample
trees from an explicitly computed tree packing distribution.

For the randomized $3$-cut algorithm, the bound forced
by failure of the strict-small branch lets us compute the required packing
quickly on a sampled skeleton.

We also use the range-counting formulation of the deterministic fixed-tree
$2$-respecting minimum-cut algorithm of Gawrychowski, Mozes, and
Weimann~\cite{GMW20b}. The centroid searches in the $2$-respecting cut optimizer require a bounded-degree tree, although the
sampled spanning trees may have arbitrary degree.  For a rooted tree $T$,
construct its binary expansion $\widehat T$ by replacing the
children of each high-degree vertex $v$ by a rooted binary tree.  The
additional vertices and edges are called auxiliary.  Let
$\pi_T\colon V(\widehat T)\to V(T)$ map every auxiliary vertex in the
gadget of $v$ to $v$ and fix every original vertex.  Contracting all
auxiliary edges recovers $T$.  The edges corresponding to $E(T)$ are
called distinguished edges.

Auxiliary vertices are used only by tree data structures and never represent
vertices of $G$.  In particular, for $X\subseteq V(T)$, let
$\widehat X:=\pi_T^{-1}(X)$.  No auxiliary edge crosses $\widehat X$,
and the distinguished edges crossing $\widehat X$ correspond bijectively
to $\delta_T(X)$.  The expansion has $O(n)$ vertices and maximum degree
at most three. We record a modified version with improved work for our specific case below.

\begin{lemma}[Oracle form of the $2$-respecting optimizer]
\label{lem:oracle_two_respecting_optimizer}
Let $H$ be a weighted graph on a vertex set $W$, and let $R$ be a
rooted spanning tree on $W$; the edges of $R$ need not belong
to $H$.  Let $\widehat R$ be an $O(n)$-vertex constant-degree
expansion of $R$, where $n:=|W|$, and let
$\pi_R\colon V(\widehat R)\to W$ be its contraction map.  The edges
corresponding to $E(R)$ are distinguished, while all other edges are
auxiliary.  For $f\in E(\widehat R)$, let
\[
    R_f^\circ:=W\cap V(\widehat R_f),
\]
where $\widehat R_f$ is the rooted subtree below $f$.

Suppose that every query of the form
$c_H(R_f^\circ)$, $c(E_H(R_f^\circ,R_g^\circ))$, or
$c(E_H(R_g^\circ,W\setminus R_f^\circ))$, whenever the displayed sets
are disjoint, can be answered in $O(Q)$ time.  After
$O(n\log n)$ preprocessing, a minimum cut of $H$ that crosses at most
two distinguished edges of $\widehat R$, equivalently at most two edges
of $R$, can be found in $O(n(Q+1)\log n)$ time and
$O(n\log n)$ space.
\end{lemma}

\begin{proof}
Note that we do not consider auxillary edges as candidate cut edges.
Build a heavy-light decomposition of the original tree $R$ and a centroid
decomposition of the constant-degree expansion $\widehat R$.

Fix a distinguished edge $e\in E(R)$.  For
$f\in E(\widehat R)$ independent of $e$, say that $e$ is
\emph{cross-interested} in $f$ if
\[
    c_H(R_e)<2c\bigl(E_H(R_e,R_f^\circ)\bigr).
\]
If $f$ is a descendant of $e$, say that $e$ is
\emph{down-interested} in $f$ if
\[
    c_H(R_e)
    <
    2c\bigl(E_H(R_f^\circ,W\setminus R_e)\bigr).
\]
As in~\cite{GMW20b}, $e$ is defined to be down-interested in every
ancestor edge $f$. Every test uses a
constant number of the assumed oracle queries.

The argument of~\cite{GMW20b} to establish the existence of the paths of interested edges applies unchanged to the laminar
family $\{R_f^\circ:f\in E(\widehat R)\}$.  Indeed, two disjoint expanded subtrees cannot each
contain more than half of the relevant boundary capacity, while interest
is preserved when moving from a subtree to an ancestor subtree.  Hence the
cross-interested and down-interested edges form root-to-node paths
$\widehat C_e$ and $\widehat D_e$, ending at vertices
$\widehat c_e$ and $\widehat d_e$.

At a centroid $x$ of $\widehat R$, test every edge $f$ incident to
$x$ for membership in the relevant interested path.  When $f$ is
auxiliary, this is a test on the original vertices
$R_f^\circ$: use the displayed cross-interest or down-interest inequality
exactly as for a distinguished edge.  No capacity is assigned to $f$.  The outcomes determine
the component of $\widehat R-x$ containing the path endpoint.  In particular, if the edge from $x$ toward the root is not on the path, then
the endpoint lies in the root-side component; otherwise, if a child edge
is on the path, the endpoint lies in that child component; and if no child
edge is on the path, the endpoint is $x$.  Since $\widehat R$ has
constant degree, recursing through its centroid decomposition finds each
endpoint using $O(\log n)$ interest tests.  This argument remains valid
when either the centroid or an incident edge is auxiliary.

For every distinguished edge $f$, its expanded subtree contains exactly
the original vertices of $R_f$, and hence $R_f^\circ=R_f$.  Therefore
$e$ is interested in $f$ in $R$ if and only if it is interested in
the distinguished copy of $f$ in $\widehat R$.  It follows that
\[
    C_e=\widehat C_e\cap E(R),
    \qquad
    D_e=\widehat D_e\cap E(R).
\]
If an expanded endpoint is auxiliary, it lies in the gadget of the original
vertex at which the corresponding original interested path terminates.
Thus the endpoints in $R$ are
$c_e=\pi_R(\widehat c_e)$ and
$d_e=\pi_R(\widehat d_e)$. Computing all
critical endpoints takes $O(nQ\log n)$ time.

For each edge $e$, the interested paths determined by $c_e$ and $d_e$
meet $O(\log n)$ heavy paths of $R$.  Emit a tag $(P(e),P)$ for every
heavy path $P$ in which $e$ is interested.  The total number of records
is $L=O(n\log n)$.

Assign every heavy path an integer identifier in
$\{1,\ldots,n\}$, and radix-sort the tags by the unordered pair of their
path identifiers and one orientation bit.  A pair of heavy paths is
mutually interested exactly when its group contains both orientations.
This identifies all mutually interested pairs in $O(L+n)$ time.

For every incidence belonging to a mutually interested pair, emit a tag
consisting of the pair identifier, its side of the pair, and the position
of the edge on its heavy path.  A second radix sort groups the edges
belonging to each pair and orders every resulting list along its heavy
path, again in $O(L+n)$ time.  The entire interest-list construction
therefore costs $O(n\log n)$.

For two distinct mutually interested paths, search the corresponding Monge
matrices by SMAWK exactly as in~\cite{GMW20b}.  The sum of their row and
column counts is $O(L)$, so this step uses $O(n\log n)$ oracle queries
and takes $O(nQ\log n)$ time.  The same-heavy-path partial-Monge instances
require $O(n\alpha(n))$ further oracle queries and are covered by the same
bound.  All remaining tree operations take $O(n\log n)$ time, proving the
lemma.
\end{proof}

\subsection{Strict-Small Cuts}
\label{subsec:strict_small}

\begin{lemma}[Crossing certificate for $k=3,4$]
\label{lem:crossing_certificate_k34}
Let $k\in\{3,4\}$ and let $A,S\subsetneq V$ be crossing cuts satisfying
$c_G(A), c_G(S)\le2U/k$.  One can construct from $A$ and $S$ a
$k$-cut of value at most $U$.
\end{lemma}

\begin{proof}
Let $X:=A\cap S$, $Y:=A\setminus S$, $Z:=S\setminus A$, and
$W:=V\setminus(A\cup S)$.  All four sets are nonempty.  If $k=4$, this
partition has value at most
$c_G(A)+ c_G(S)\le U$.

Suppose $k=3$, and set $a:=c(E_G(X,Y))$ and $f:=c(E_G(Z,W))$.  These two
edge classes are disjoint subsets of $\delta_G(S)$, so
$\min\{a,f\}\le U/3$.  Splitting $A$ by $S$ costs
$c_G(A)+a$, whereas splitting $V\setminus A$ by $S$ costs
$c_G(A)+f$.  The cheaper partition has value at most $U$.
\end{proof}

\begin{lemma}[Uniqueness of a strict-small cut bipartition]
\label{lem:unique_strict_small_cut}
There is at most one strict-small cut bipartition.  Equivalently, if
$S,T\subsetneq V$ satisfy
$c_G(S), c_G(T)<\frac{\lambda}{2}$,
then
$\{S,V\setminus S\} = \{T,V\setminus T\}$.
\end{lemma}
\begin{proof}
Suppose the two cut bipartitions are distinct.  Their common refinement has at least three nonempty parts.  Coarsen it arbitrarily to
three parts.  Every edge crossing the resulting partition crosses at least
one of the two original cuts, so its value is at most
$c_G(S)+c_G(T)<\lambda$, contradicting the definition of
$\lambda=\lambda_3(G)$.
\end{proof}

Let $\mu:=\lambda_2(G)$, and compute a global minimum cut
$(S,V\setminus S)$ of value $\mu$.  If $|S|\ge2$, compute a
global minimum cut of $G[S]$; if $|V\setminus S|\ge2$, compute a global
minimum cut of $G[V\setminus S]$.  Combining each induced cut with the
border $(S,V\setminus S)$ gives at most two candidate $3$-cuts.  Let
$\mathcal P_{\mathrm{small}}$ be the lighter one, with value $+\infty$
if neither side can be split.  The two induced graphs have disjoint edge
sets, so these computations take $O(m\log^2 n)$ total time using
the algorithm of~\cite{GMW20a}.

\begin{lemma}[Strict-small-side completion]
\label{lem:small_side_completion}
If an optimum $3$-cut has a strict-small side, then
$c_G(\mathcal P_{\mathrm{small}})=\lambda$.
\end{lemma}

\begin{proof}
Let $A$ be a strict-small optimum side.  The global minimum-cut value is at
most $c_G(A)<\lambda/2$, so every global minimum-cut bipartition is
strict-small.  By \cref{lem:unique_strict_small_cut}, it is exactly
$(A,V\setminus A)$ up to complementation.

The optimum partition splits one side of this bipartition.  The minimum
induced cut on that side is no more expensive than the corresponding
optimum split, so one of the two constructed $3$-cuts has value at most
$\lambda$.  By definition of $\lambda$, equality holds.
\end{proof}

The candidate $\mathcal P_{\mathrm{small}}$ is always a valid $3$-cut.

\begin{lemma}[Structure after failure of strict-small completion]
\label{lem:strict_small_failure_structure}
Suppose $c_G(\mathcal P_{\mathrm{small}})>\lambda$ and
$\mu<\lambda/2$, and let $(R,V\setminus R)$ be the global minimum cut used
in constructing $\mathcal P_{\mathrm{small}}$.  For every optimum
$3$-cut, after relabeling its parts as $(A,B,C)$ and possibly replacing
$R$ by its complement, we have $R\subsetneq A$.  Put
$D:=A\setminus R$, and define
$x:=c(E_G(R,D))$, $p:=c(E_G(R,B))$,
$q:=c(E_G(R,C))$, $r:=c(E_G(D,B))$,
$t:=c(E_G(D,C))$, and $s:=c(E_G(B,C))$.
Then $x>\max\{r,t\}$, $x\ge s$, and
$r+t\ge p+q$.  Consequently,
$x>\mu/3$, $c_G(D)<3\mu$, and $c_G(A)<2\mu$.
\end{lemma}

\begin{proof}
Let $x_R$ and $x_{V\setminus R}$ be the capacities of optimum-partition
edges with both endpoints in the indicated sides of the global minimum cut.
If both sides meet at least two optimum parts, restricting the optimum
partition to either side gives two candidates of values at most
$\mu+x_R$ and $\mu+x_{V\setminus R}$.  Since
$x_R+x_{V\setminus R}\le\lambda$, the better candidate has value at most
$\mu+\lambda/2<\lambda$, a contradiction.  Thus one side of the global
minimum cut lies in one optimum part; orient the cut so that this side is
$R\subseteq A$.  Equality $R=A$ would make $R$ a strict-small optimum
side, and then \cref{lem:small_side_completion} would give
$c_G(\mathcal P_{\mathrm{small}})=\lambda$.  Hence $R\subsetneq A$.

We have $\mu=x+p+q$ and $\lambda=p+q+r+t+s$.  Splitting
$V\setminus R$ as $B$ and $D\cup C$ gives a strict-small candidate of
value at most $\mu+r+s=\lambda+x-t$.  Since
$c_G(\mathcal P_{\mathrm{small}})>\lambda$, this implies $x>t$.
The symmetric split gives $x>r$.  Optimality against the partition
$(R,D,B\cup C)$ gives $x\ge s$, while global minimality of $R$
against the cut $D$ gives $r+t\ge p+q$.

Now $r+t<2x$ and $r+t\ge\mu-x$, so $x>\mu/3$.  Moreover,
$c_G(D)=x+r+t<3x\le3\mu$, while
$c_G(A)=p+q+r+t<p+q+2x=\mu+x\le2\mu$.
\end{proof}

\begin{corollary}[Comparison of $\lambda_3$ and $\lambda_2$]
\label{cor:lambda_three_vs_lambda_two}
If $c_G(\mathcal P_{\mathrm{small}})>\lambda$, then
$\lambda<3\mu$.
\end{corollary}

\begin{proof}
If $\mu\ge\lambda/2$, then $\lambda\le2\mu<3\mu$.  Otherwise apply
\cref{lem:strict_small_failure_structure}.  Its inequalities give
$\lambda=p+q+r+t+s<p+q+3x=\mu+2x\le3\mu$.
\end{proof}

\subsection{Respecting Cuts and Capacity Matrices}
\label{sec:enumeration_respecting_cuts}

Let $T$ be a tree on $V(G)$, rooted at a fixed vertex $r$; its edges
need not belong to $G$.  We only consider rooted cuts
$S\subseteq V(G)\setminus\{r\}$.  For a non-root vertex $v$, let $e_v$
be the tree edge joining $v$ to its parent, and let $T_v$ be the rooted
subtree below $v$.  We also use a dummy symbol $\bot$, with
$T_\bot=\emptyset$.

\begin{lemma}[Normal form for rooted $p$-respecting cuts]
\label{lem:normal_form_p_respecting}
Let $p\ge1$.  A rooted cut $S\subseteq V\setminus\{r\}$ satisfies
$|\delta_T(S)|\le p$ if and only if it can be written as
$S=T_{a_1}\triangle T_{a_2}\triangle\cdots\triangle T_{a_q}$
for some $q\le p$ and some non-root vertices $a_1,\ldots,a_q$.  Equivalently, allowing dummy copies of $\bot$, every rooted $p$-respecting cut can be written as
$S=T_{a_1}\triangle T_{a_2}\triangle\cdots\triangle T_{a_p}, \qquad a_i\in V\cup\{\bot\}$.
Moreover, the representation is unique if we require $\{e_{a_1},\ldots,e_{a_q}\}=\delta_T(S)$ and list these tree edges in a fixed canonical order.
\end{lemma}
\begin{proof}
Let $B:=\delta_T(S)$.  Since $S$ is rooted, $r\notin S$.  For any vertex $x$, the membership of $x$ in $S$ changes exactly when the path from $r$ to $x$ crosses an edge of $B$.  Thus $x\in S$ if and only if the number of edges of $B$ on the $r$-to-$x$ path is odd.  This is precisely the membership rule for
$\triangle_{e_a\in B} T_a$.
Therefore $S=\triangle_{e_a\in B} T_a$, and $|B|\le p$.

Conversely, each fundamental cut $T_a$ crosses exactly the tree edge $e_a$. The symmetric difference of $q$ such cuts can therefore cross only the corresponding $q$ tree edges.

The uniqueness statement follows because, for a rooted cut $S$, the set of tree edges where membership changes is exactly $\delta_T(S)$.  Hence the canonical representation is determined by $S$.
\end{proof}

We use the following for the queries in \cref{lem:oracle_two_respecting_optimizer}.

\begin{lemma}[Capacity matrix]
\label{lem:capacity_matrix}
For a rooted tree $T$ on $V(G)$, let $\widehat T$ be its binary
expansion.  Order the original vertices $V$ by a DFS traversal of
$\widehat T$, ignoring auxiliary vertices, and define
$W_T[i,j]:=c(E_G(v_i,v_j))$, with parallel capacities summed.  The matrix
$W_T$ and its two-dimensional prefix sums can be constructed in
$O(m+n^2)$ time and $O(n^2)$ space.  After this preprocessing, if
$X,Y\subseteq V$ are each represented as a disjoint union of $O(1)$
intervals in this order, then $c(E_G(X,Y))$ can be computed in $O(1)$
time.  The same construction and query bound hold for any prescribed order
of $V$.

Every rooted subtree of $\widehat T$ intersects $V$ in one interval.
More generally, if $\widehat X\subseteq V(\widehat T)$ satisfies
$|\delta_{\widehat T}(\widehat X)|\le q$ for a fixed constant $q$, then
both $\widehat X\cap V$ and $V\setminus\widehat X$ can be represented as
disjoint unions of at most $q+1$ intervals in $O_q(1)$ time.  In
particular, the same conclusion holds for $X\subseteq V$ whenever
$|\delta_T(X)|\le q$.
\end{lemma}

\begin{proof}
Constructing $\widehat T$ and its DFS order takes $O(n)$ time.
Initialize $W_T$, scan $E(G)$, and add every capacity to the two
symmetric entries corresponding to its endpoints.  This takes
$O(m+n^2)$ time, and constructing the prefix-sum table takes another
$O(n^2)$ time.  The interaction of two constant-size interval unions is
the sum of a constant number of rectangle sums.

A rooted subtree of $\widehat T$ is consecutive in DFS order, even when
its root edge is auxiliary.  Orient the cut of $\widehat T$ so that its
root is outside $\widehat X$.  The parity argument of
\cref{lem:normal_form_p_respecting}, applied to $\widehat T$, expresses
$\widehat X$ as the symmetric difference of the rooted subtrees below
the edges of $\delta_{\widehat T}(\widehat X)$.  The traces of these
subtrees on $V$ form a laminar family of intervals.  Sorting their
$O(q)$ endpoints and considering parity gives at most $q+1$
intervals for either side.

Finally, if $X\subseteq V$, then
$|\delta_{\widehat T}(\pi_T^{-1}(X))|=|\delta_T(X)|$, so the preceding
claim applies to its lift.
\end{proof}

\begin{lemma}[Enumeration of respecting cuts]
\label{lem:respecting_enumeration}
Fix a constant $p\ge1$, a rooted tree $T$ on $V(G)$, and an integer
$M\ge1$.  After constructing the capacity matrix of
\cref{lem:capacity_matrix}, one can enumerate every nonempty rooted
cut $S\subseteq V\setminus\{r\}$ satisfying $|\delta_T(S)|\le p$, compute
its capacity, and retain the $M$ lightest such cuts in
$p^{O(1)}n^p$ time and $O(n^2+pM)$ space.  Each retained cut is stored
by its set of crossing tree edges.  In particular, for $p=2$
the total time, including the capacity matrix, is $O(m+n^2)$.
\end{lemma}

\begin{proof}
Enumerate the subsets $B\subseteq E(T)$ with $1\le|B|\le p$ in canonical
order.  \cref{lem:normal_form_p_respecting} identifies the corresponding
rooted side as the symmetric difference of the subtrees below the edges of
$B$, and it shows that different sets $B$ give different cuts.  By
\cref{lem:capacity_matrix}, this side and its complement are unions
of $O(p)$ intervals, so its capacity is the sum of $p^{O(1)}$ rectangle
queries.  This gives $p^{O(1)}n^p$ total time.

To keep only the $M$ lightest cuts without an additional logarithmic
factor, store candidates in a buffer of size $2M$.  Whenever the buffer
fills, use linear-time selection to discard all but its $M$ lightest
members, breaking ties by the canonical boundary-edge order.  Another pruning happens only after $M$ new candidates have been inserted and takes $O(M)$, hence the total
selection time is linear in the number of enumerated cuts.
\end{proof}

\subsection{Sampling Trees from a Random Skeleton}
\label{subsec:sampling_trees}

We first record the sampled skeleton used to reduce the cost of the 
packing.  For a weighted graph $F$ with
$\mu_F:=\lambda_2(F)$, let $\overline F$ be obtained by replacing every
capacity $c(e)$ by $\overline c(e):=\min\{c(e),3\mu_F\}$.

\begin{lemma}[Random cut skeleton]
\label{lem:random_cut_skeleton}
Let $F$ be a connected $n$-vertex, $m$-edge weighted multigraph and let
$0<\varepsilon<1/2$.  In $O(m\varepsilon^{-2}\log^2 n)$ randomized time,
one can construct an unweighted multigraph $H$ on $V(F)$ and a scale $W>0$
such that, with probability $1-n^{-\Omega(1)}$, every cut $S$ satisfies
$(1-\varepsilon)c_{\overline F}(S)\le Wc_H(S)
\le(1+\varepsilon)c_{\overline F}(S)$.
Moreover, $H$ has $O(m\varepsilon^{-2}\log n)$ edges counted with
multiplicity, these edges can be stored in $O(m)$ space,
and $\mu_F/W=O(\varepsilon^{-2}\log n)$.  Every edge copy of $H$
corresponds to an edge of $F$, so every spanning tree of $H$ is 
a spanning tree of the underlying graph of $F$.
\end{lemma}

\begin{proof}
This is the standard weighted sampling skeleton of
Karger~\cite{Kar00}, with the weighted sampling implementation of
Gawrychowski, Mozes, and Weimann~\cite{GMW20a}.  Conceptually replace an
edge $e$ by $\overline c(e)$ unit-capacity copies and sample each copy with
probability
$q:=\min\{1,C_0\varepsilon^{-2}\log n/\mu_F\}$, where $C_0$ is a
sufficiently large constant; set $W:=1/q$.  Capping does not change the
global minimum-cut value: a cut containing a capped edge has capped value
at least $3\mu_F$, while every other cut is unchanged.  Karger's
cut-counting bound and the Chernoff bound therefore give the simultaneous
cut approximation.

Since $\overline c(E)\le3m\mu_F$, the sampled graph has
$O(m\varepsilon^{-2}\log n)$ copies with high probability.  The weighted
sampling implementation of~\cite{GMW20a}
constructs it in $O(m\varepsilon^{-2}\log^2 n)$ time.
\end{proof}

\begin{lemma}[Fast tree packing on the skeleton]
\label{lem:fast_skeleton_packing}
For every $0<\eta\le1$, one can compute a feasible fractional
spanning-tree packing $(y_T)$ in $H$ whose value
$\tau:=\sum_Ty_T$ satisfies
$\tau\ge(1-\eta)\tau^*(H)$, where $\tau^*(H)$ is the maximum
fractional spanning-tree-packing value of $H$.
The packing has support
$O(\varepsilon^{-2}\eta^{-2}\log^2 n)$ and can be constructed in
$O(m\varepsilon^{-2}\eta^{-2}\log^2 n)$ randomized time with probability $1 - e^{-\Omega(m)}$, including
the construction of $H$.

Moreover, if
$c_G(\mathcal P_{\mathrm{small}})>\lambda_3(G)$, then
$\tau\ge(1-\eta)\lambda_3(H)/4$.
\end{lemma}

\begin{proof}
Write $\lambda_H:=\lambda_3(H)$ and $\mu_H:=\lambda_2(H)$.

We apply the greedy tree-packing algorithm of~\cite{GMW20a}.  Writing $m':=|E(H)|$, set
$B:=\lceil3\eta^{-2}\ln m'\rceil$, each iteration adds weight $1/B$ to
the chosen tree and load $1/B$ to each of its edges, we stop once any edge reaches load $1$.  The standard MWU
analysis gives a feasible packing of value
$\tau\ge(1-\eta)\tau^*(H)$.  If $L$ is the number of iterations, then
$L=B\tau$.  Since every feasible tree packing has value at most
$c':=\lambda_2(H)$, we have $L\le Bc'$.  The skeleton guarantee gives
$c'\le(1+\varepsilon)\mu/W=O(\varepsilon^{-2}\log n)$, while
$\log m'=O(\log n)$.  Consequently
$L=O(\varepsilon^{-2}\eta^{-2}\log^2 n)$, which also bounds the support
size.

The parallel copies sampled from each original edge are stored by
multiplicity.  Always choosing a least-loaded copy keeps the loads of the
copies in one bundle within $1/B$ of one another, so its lightest copy is
determined by the bundle's total number of previous uses.  A minimum
spanning tree of $H$ can therefore be computed over the at most $m$
original edge bundles.  Using the randomized linear-time implementation
employed in~\cite{GMW20a}, each iteration takes $O(m)$ time with probability $1 - e^{-\Omega(m)}$, giving total
time $O(m\varepsilon^{-2}\eta^{-2}\log^2 n)$.

Now assume $c_G(\mathcal P_{\mathrm{small}})>\lambda_3(G)$ and write $\rho:=(1+\varepsilon)/(1-\varepsilon)$. By \cref{cor:lambda_three_vs_lambda_two}, $\lambda<3\mu$.
The skeleton guarantees consequently imply
$\mu_H\ge(1-\varepsilon)\mu/W$ and
$\lambda_H\le(1+\varepsilon)\lambda/W$, and hence
\[
  \mu_H>\frac{1-\varepsilon}{3(1+\varepsilon)}\lambda_H
       =\frac{\lambda_H}{3\rho}>\frac{\lambda_H}{4},
\]
where the last inequality uses $\varepsilon<1/7$.

By the Tutte--Nash-Williams formula,
\[
  \tau^*(H)=
  \min_{\mathcal P}
  \frac{c_H(\mathcal P)}{|\mathcal P|-1},
\]
where the minimum ranges over all nontrivial vertex partitions
$\mathcal P$.

For a bipartition, the corresponding ratio is at least
$\mu_H>\lambda_H/4$.  Now let
$\mathcal P=(P_1,\ldots,P_r)$ have $r\ge3$ parts.  Choose two parts
uniformly at random and retain them as two parts of a $3$-partition,
merging the other $r-2$ parts into the third part.  An edge of
$\mathcal P$ remains cut unless both of its incident parts belong to the
merged collection.  Its probability of remaining cut is
\[
  1-\frac{\binom{r-2}{2}}{\binom r2}
  =\frac{4r-6}{r(r-1)}.
\]
Some resulting $3$-cut therefore has value at most
$(4r-6)c_H(\mathcal P)/(r(r-1))$.  By the definition of $\lambda_H$,
\[
  \frac{c_H(\mathcal P)}{r-1}
  \ge\frac{r}{4r-6}\lambda_H
  >\frac{\lambda_H}{4}.
\]
Thus every partition in the Tutte--Nash-Williams formula has ratio
strictly larger than $\lambda_H/4$, and consequently
$\tau^*(H)>\lambda_H/4$.
\end{proof}

Fix $\varepsilon=\eta=1/100$, construct $H$, and compute the packing of
\cref{lem:fast_skeleton_packing}, aborting if it takes too long.  Write
$\lambda_H:=\lambda_3(H)$ and
$\rho:=(1+\varepsilon)/(1-\varepsilon)$.  Sample a tree with probability
$y_T/\tau$, and regard it as a spanning tree of $G$.

\begin{lemma}[A light side is $2$-respected with constant probability]
\label{lem:light_side_two_respected}
Suppose $c_G(\mathcal P_{\mathrm{small}})>\lambda$, and let
$S\subsetneq V$ satisfy $c_G(S)\le2\lambda/3$.  Then
\[
 \Pr[|\delta_T(S)|\le2]
 \ge\frac{9(1-\eta)-8\rho}{6(1-\eta)}.
\]
For $\varepsilon=\eta=1/100$, this probability is greater than $1/8$.
\end{lemma}

\begin{proof}
\cref{cor:lambda_three_vs_lambda_two} gives
$c_G(S)<2\mu$, so capacity capping does not affect $S$.  The skeleton
guarantees give $c_H(S)\le(1+\varepsilon)c_G(S)/W$ and
$\lambda_H\ge(1-\varepsilon)\lambda/W$, and hence
$c_H(S)\le(2\rho/3)\lambda_H$.

Let $D_T:=|\delta_T(S)|$.  Feasibility of the packing and
\cref{lem:fast_skeleton_packing} give
\[
 \tau\mathbb E[D_T]\le c_H(S)
 \qquad\text{and}\qquad
 \tau\ge\frac{1-\eta}{4}\lambda_H.
\]
Consequently,
$\mathbb E[D_T]\le8\rho/(3(1-\eta))$.
If $p:=\Pr[D_T\le2]$, then every spanning tree crosses $S$ at least once,
while $D_T\ge3$ whenever $D_T>2$.  Thus
$\mathbb E[D_T]\ge p+3(1-p)=3-2p$, and therefore
\[
 p\ge\frac{3-\mathbb E[D_T]}2
 \ge\frac{9(1-\eta)-8\rho}{6(1-\eta)}.
\]
For $\varepsilon=\eta=1/100$, the last expression equals
$7409/58806>1/8$.
\end{proof}

For a sampled tree $T$, define
$d_A:=|\delta_T(A)|$, $d_B:=|\delta_T(B)|$,
$d_C:=|\delta_T(C)|$, and
$h:=|T\cap\delta_G(A,B,C)|$.  Thus $d_A+d_B+d_C=2h$.
We call $T$ \emph{good} for $(A,B,C)$ if $h\le4$ and
$(d_A,d_B,d_C)\notin\{(1,3,4),(1,4,3)\}$.

We next explain the intuition behind this definition in the unique-light-side case.

Let $(A,B,C)$ be an optimum $3$-cut in which $A$ is the unique light side and is medium.
The objective is not purely few optimum-crossing
edges.  What matters after fixing $A$ is how complicated the residual
$B$--$C$ cut remains, and the $(1,3,4)$ or $(1,4,3)$ patterns are the most problematic.
This is captured by charging an $A$--$B$ or
$A$--$C$ tree edge three units and a $B$--$C$ tree edge four units,
which gives the statistic
\[
    d_A+2d_B+2d_C=4h-d_A.
\]
Edges crossing the residual $B$--$C$ cut are charged more heavily
because they survive after $A$ is removed.

\begin{lemma}[Constant probability of a good tree]
\label{lem:constant_probability_good_tree}
Suppose $c_G(\mathcal P_{\mathrm{small}})>\lambda$, and let
$(A,B,C)$ be an optimum $3$-cut satisfying
$\lambda/2\le c_G(A)\le2\lambda/3$ and
$c_G(B),c_G(C)>2\lambda/3$.  Then
\[
 \Pr[T\text{ is good}]
 \ge\frac{15(1-\eta)-14\rho}{9(1-\eta)}.
\]
For $\varepsilon=\eta=1/100$, this probability is greater than $1/16$.
\end{lemma}

\begin{proof}
For every tree $T$,
\[
 \mathbf 1[T\text{ is good}]
 \ge\frac53-\frac{d_A+2d_B+2d_C}{9}.
\]
Indeed, if $h\ge5$, then
$d_A+2d_B+2d_C=4h-d_A\ge3h\ge15$.  If $h\le4$ but $T$ is not good,
the degree triple is one of the two excluded triples and the weighted sum
is exactly $15$.  Finally, if $T$ is good, then $h\ge2$, so the
right-hand side is at most one.

Every optimum side has boundary at most $\lambda<3\mu$, so capacity
capping does not affect $A$, $B$, or $C$.  Since
$c_G(A)+2c_G(B)+2c_G(C)=4\lambda-c_G(A)\le7\lambda/2$, the skeleton
guarantees give
\[
 c_H(A)+2c_H(B)+2c_H(C)
 \le\frac{7\rho}{2}\lambda_H.
\]
Feasibility of the packing therefore gives
\[
 \tau\mathbb E[d_A+2d_B+2d_C]
 \le\frac{7\rho}{2}\lambda_H.
\]
Together with
$\tau\ge(1-\eta)\lambda_H/4$, this yields
\[
 \mathbb E[d_A+2d_B+2d_C]
 \le\frac{14\rho}{1-\eta}.
\]
Taking expectations in the pointwise inequality gives
\[
 \Pr[T\text{ is good}]
 \ge\frac53-\frac{14\rho}{9(1-\eta)}
 =\frac{15(1-\eta)-14\rho}{9(1-\eta)}.
\]
For $\varepsilon=\eta=1/100$, the last expression equals
$5615/88209>1/16$.
\end{proof}

Choose $s_{\mathrm{trees}}=\Theta(\log n)$ independent samples from the
packing distribution, with a sufficiently large absolute constant in the
$\Theta(\cdot)$, and denote the resulting multiset by $\mathcal T$.
By \cref{lem:light_side_two_respected} and
\cref{lem:constant_probability_good_tree}, with probability
$1-n^{-\Omega(1)}$, the following hold simultaneously for any fixed
optimum $3$-cut:

\begin{enumerate}
    \item every light optimum side is $2$-respected by some tree in
    $\mathcal T$;

    \item if the optimum has a unique medium side, some tree in
    $\mathcal T$ is good for that optimum.
\end{enumerate}

\subsection{Selecting the Relevant Light Cuts}
\label{subsec:reporting_light_cuts}

\begin{lemma}[Linear-incidence laminarity test]
\label{lem:laminarity_test}
Let $\mathcal F$ be a family of distinct explicitly represented nonempty
proper subsets of an $n$-element ground set.  In
$O\left(n+\sum_{S\in\mathcal F}|S|\right)$
time, one can either find two crossing members of $\mathcal F$, or construct
the inclusion forest of $\mathcal F$.
\end{lemma}

\begin{proof}
Add the whole ground set as a sentinel root.  Sort the members of
$\mathcal F$ by nonincreasing cardinality.  Maintain, for every vertex $v$,
a pointer $\operatorname{owner}(v)$ to the smallest processed set containing
$v$, initially the sentinel.

Until a crossing pair is found, the processed sets are laminar, and
$\operatorname{owner}(v)$ is the inclusion-minimal processed set
containing $v$.

When processing $S$, inspect the owners of all vertices of $S$.  If they are
all the same set $P$, then $S\subseteq P$; make $P$ the parent of $S$ and set
$\operatorname{owner}(v):=S \qquad(v\in S)$.
No previously processed proper subset of $S$ exists, because sets are
processed in nonincreasing order of size.

Otherwise, choose vertices $x,y\in S$ with distinct owners $P,Q$.  If one of
$P,Q$ is the sentinel and the other is not, then $S$ intersects the
nonsentinel owner and also has a vertex outside it.  Since that owner was
processed earlier and has size at least $|S|$, it also has a vertex outside
$S$, so the two sets cross.  If both $P,Q$ are nonsentinel, then they are
disjoint minimal processed containers.  The set $S$ meets both, and neither
can be contained in $S$ because it was processed earlier and is distinct
from $S$.  Hence $S$ crosses at least one of them.

Every set incidence is examined and updated only once, which gives the claimed running time.
\end{proof}

Fix a root $r\in V$.  For every $T\in\mathcal T$, use
\cref{lem:respecting_enumeration} with $p=2$ to retain
the $2n+1$ lightest nonempty rooted cuts that $2$-respect $T$, or all of
them if fewer exist.  Materialize these $O(n\log n)$ cuts, radix-sort their
incidence vectors to remove duplicates, and retain the globally lightest
$2n+1$ distinct cuts, again retaining all cuts if fewer exist.  Denote this
family by $\mathcal F$, ordered by nondecreasing boundary value with a
fixed deterministic tie-breaking rule.

Laminarity is monotone under taking subsets.  Binary search with
\cref{lem:laminarity_test} therefore finds the longest laminar prefix
$\mathcal L$ of $\mathcal F$ in $O(n^2\log n)$ time.  If
$\mathcal L\ne\mathcal F$, apply the lemma once more to the prefix ending
at the first omitted cut and obtain a crossing pair $P,Q$.  Let $B$ be
the boundary of that last cut.  Both $P$ and $Q$ have boundary at most
$B$, so \cref{lem:crossing_certificate_k34}, with $k=3$ and
$U=3B/2$, constructs a $3$-cut $\mathcal P_{\mathrm{cross}}$ of value
at most $3B/2$.  If $\mathcal F$ is laminar, set
$c_G(\mathcal P_{\mathrm{cross}}):=+\infty$.

\begin{lemma}[Coverage by the truncated laminar prefix]
\label{lem:sampled_light_list_complete}
Let $S$ be a side of an optimum $3$-cut satisfying
$c_G(S)\le2\lambda/3$.  With probability
$1-n^{-\Omega(1)}$, either the rooted orientation of
$(S,V\setminus S)$ belongs to $\mathcal L$, or
$c_G(\mathcal P_{\mathrm{cross}})=\lambda$.
\end{lemma}

\begin{proof}
By \cref{lem:light_side_two_respected}, some sampled tree
$2$-respects the cut with high probability.  Condition on this event and
let $S_r$ be its rooted orientation.  If $S_r\notin\mathcal F$, then the
local or global truncation retained $2n+1$ distinct cuts of boundary at
most $c_G(S)$.  Hence $|\mathcal F|=2n+1$, and every member of
$\mathcal F$, including the first nonlaminar prefix, has boundary at most
$c_G(S)$.  If instead $S_r\in\mathcal F\setminus\mathcal L$, the
last cut in the first nonlaminar prefix precedes $S_r$ and again has
boundary $B\le c_G(S)$.

In either case, the family is nonlaminar because a rooted laminar family
has fewer than $2n$ members in the first case and by definition in the
second.  Thus $\mathcal P_{\mathrm{cross}}$ exists and has value at most
$3B/2\le3 c_G(S)/2\le\lambda$.  Since no $3$-cut has value below
$\lambda$, it is optimum.  The only remaining possibility is
$S_r\in\mathcal L$.
\end{proof}

\subsection{Eliminating Two-Light-Side Optima}
\label{subsec:two_light_sides}

Attach every original vertex to the smallest
member of $\mathcal L$ containing it, or to the sentinel root if no such
member exists.  A DFS of this inclusion forest gives a vertex order in which
every set $S\in\mathcal L$ is one interval.  Build the capacity
matrix and its prefix sums in this order.

For two members $S,T\in\mathcal L$, laminarity implies that they are either
disjoint or nested.

\begin{lemma}[Value of the partition of two rooted cuts]
\label{lem:two_cut_value}
Let $S,T\in\mathcal L$ be distinct.

\begin{enumerate}
    \item If $S\cap T=\emptyset$, the partition
    $\bigl(S,T,V\setminus(S\cup T)\bigr)$
    has value
    $c_G(S)+ c_G(T)-c(E_G(S,T))$.

    \item If $S\subsetneq T$, the partition
    $\bigl(S,T\setminus S,V\setminus T\bigr)$
    has value
    $c_G(S)+ c_G(T)-c(E_G(S,V\setminus T))$.
\end{enumerate}

Each value can be computed in $O(1)$ time from the prefix table.
\end{lemma}

\begin{proof}
In the disjoint case, the edge class $E_G(S,T)$ is counted twice in
$c_G(S)+ c_G(T)$, while the other two crossing edge classes are
counted once.  In the nested case, the edge class $E_G(S,V\setminus T)$ is
counted twice and the other two crossing classes are counted once.
Subtracting the duplicated class gives the stated formulas.
\end{proof}

Inspect all $O(n^2)$ unordered pairs in $\mathcal L$ and retain the
minimum partition.  The nested case is necessary because the cuts are
oriented away from a fixed root: if the root belongs to one light optimum
side, the rooted orientation of that side is its complement and is nested
with the rooted orientation of another light side.

\begin{lemma}[Completeness of two-light-side elimination]
\label{lem:two_light_side_complete}
If some optimum $3$-cut has at least two sides of boundary at most
$2\lambda/3$, then, with probability $1-n^{-\Omega(1)}$, either
$\mathcal P_{\mathrm{cross}}$ or the pair scan is an optimum $3$-cut.
\end{lemma}

\begin{proof}
Choose two light optimum sides.  By
\cref{lem:sampled_light_list_complete}, either
$\mathcal P_{\mathrm{cross}}$ is optimum or both rooted cut orientations
belong to $\mathcal L$ with high probability.  In the latter case, if the root
lies in the third side, the two rooted sides are disjoint.  If the root lies
in one of the chosen sides, its rooted side is the complement of that side
and contains the rooted side of the other chosen side.  In either case, the
three parts of the corresponding pair are exactly the three optimum parts,
so \cref{lem:two_cut_value} evaluates the optimum value.
\end{proof}

Suppose henceforth that the strict-small branch and the two-light-side branch
do not recover the optimum.  Then we may fix an optimum partition
$(A,B,C)$ for which $A$ is the unique light side.  Thus
$\lambda/2\le c_G(A)\le2\lambda/3$, while
$c_G(B), c_G(C)>2\lambda/3$.
We call $A$ the \emph{medium side}.  Since the enumeration stores only the
rooted orientation of its cut bipartition, the algorithm processes both
sides of every set in $\mathcal L$.  Consequently, the actual side $A$
is among the processed candidates whenever its rooted cut belongs to
$\mathcal L$.

\subsection{Tree Components and Good-Tree Patterns}
\label{subsec:good_tree_patterns}

Fix a tree $T$ and a vertex set $A$.  Write
$d_A:=|\delta_T(A)|$,
and let
$\kappa_A:=\operatorname{cc}(T[A]), \qquad \rho_A:=\operatorname{cc}(T \setminus A)$.

\begin{lemma}[Component identity in a tree]
\label{lem:tree_component_identity}
For every nontrivial $A\subsetneq V$,
$\kappa_A+\rho_A=d_A+1$.
\end{lemma}

\begin{proof}
Contract every connected component of $T[A]$ and every connected component
of $T \setminus A$.  The resulting graph is a tree.  It has
$\kappa_A+\rho_A$ vertices and $d_A$ edges, so
$d_A=\kappa_A+\rho_A-1$.
\end{proof}

For the optimum partition $(A,B,C)$ and a fixed tree $T$, define
$\alpha:=|T\cap E_G(A,B)|, \qquad \beta:=|T\cap E_G(A,C)|, \qquad \gamma:=|T\cap E_G(B,C)|$.
Then
$d_A=\alpha+\beta, \qquad d_B=\alpha+\gamma, \qquad d_C=\beta+\gamma$,
and
$h=\alpha+\beta+\gamma$.
After deleting $A$, the residual $B$--$C$ cut crosses exactly
$\gamma=h-d_A$
edges of the forest $T \setminus A$.

\begin{lemma}[Coverage of all good-tree patterns]
\label{lem:good_pattern_dichotomy}
Let $T$ be good for $(A,B,C)$.  Then:

\begin{enumerate}
    \item if $d_A=1$, the forest $T \setminus A$ is a tree and the residual optimum
    crosses at most two of its edges;

    \item if $d_A=2$ and $T[A]$ is disconnected, the forest $T \setminus A$ is a
    tree and the residual optimum crosses at most two of its edges;

    \item if $d_A=2$ and $T[A]$ is connected, the forest $T \setminus A$ has two
    components and the residual optimum crosses at most two surviving tree
    edges;

    \item if $d_A\ge3$, then $d_A\le4$, the forest $T \setminus A$ has at most four
    components, and the residual optimum crosses at most one surviving tree
    edge.
\end{enumerate}
\end{lemma}

\begin{proof}
Since $T$ is good, $h\le4$.  If $d_A=1$, the excluded degree patterns imply
$h\le3$, and therefore
$\gamma=h-d_A\le2$.
\cref{lem:tree_component_identity} forces
$\kappa_A=\rho_A=1$.

If $d_A=2$, then $\gamma=h-2\le2$.  The component identity shows that
$T \setminus A$ is connected when $T[A]$ is disconnected, and has two components when
$T[A]$ is connected.

Finally, if $d_A\ge3$, then $d_A\le h\le4$ and
$\gamma=h-d_A\le1$.
Moreover,
$\rho_A=d_A-\kappa_A+1\le d_A\le4$.
\end{proof}

\noindent We will later show that the $d_A \ge 3$ case is easy, thus we focus on $d_A \le 2$.

\subsection{Virtual Links for Connected Two-Respecting Sides}
\label{subsec:virtual_links}

Suppose $d_A=2$ and $T[A]$ is connected.  By
\cref{lem:tree_component_identity}, the forest $T \setminus A$ has exactly two
components, say $K_1$ and $K_2$.  Let $u_i\in K_i$ be the outside endpoint
of the unique tree edge between $A$ and $K_i$.  Define the tree
$R_A^{\mathrm{virt}} := (T \setminus A)+u_1u_2$,
where the added link has capacity zero.  It need not be an edge of $G$.

\begin{lemma}[Virtual-link coverage]
\label{lem:virtual_link_coverage}
Let $T$ be good for an optimum $(A,B,C)$, suppose $d_A=2$, and suppose
$T[A]$ is connected.

\begin{enumerate}
    \item If $h\le3$, the residual $B$--$C$ cut crosses at most two edges of
    $R_A^{\mathrm{virt}}$.

    \item If $h=4$ and
    $(d_A,d_B,d_C)\in\{(2,2,4),(2,4,2)\}$,
    the residual $B$--$C$ cut crosses exactly its two surviving tree edges
    and does not cross the virtual link.

    \item The only connected $d_A=2$ good pattern not covered by the
    virtual link is
    $(d_A,d_B,d_C)=(2,3,3)$.
\end{enumerate}
\end{lemma}

\begin{proof}
If $h\le3$, then the residual cut crosses
$\gamma=h-2\le1$
surviving tree edge.  The virtual link contributes at most one further
crossing.

Now suppose $h=4$.  Then $\gamma=2$, while
$\alpha+\beta=2$.
For the degree pattern $(2,2,4)$ we have
$\alpha=0, \qquad \beta=2$.
Thus both tree edges leaving $A$ enter $C$, so $u_1,u_2\in C$ and the virtual
link is not crossed.  The case $(2,4,2)$ is symmetric.  The remaining
possibility is $\alpha=\beta=1$, which gives the pattern $(2,3,3)$ and places
$u_1,u_2$ on opposite residual sides.
\end{proof}

\subsection{Deterministic Rectangle Nets}
\label{subsec:dense_oracles}

We need to choose edges deterministically from rectangular ranges
in order to obtain repair edges.

\begin{lemma}[Deterministic weighted rectangle net]
\label{lem:rectangle_net}
Let $P$ be a nonnegative weighted point set in $[N]\times[N]$, supported on
a union of $O(1)$ rectangles, and let $W:=w(P)>0$.  Suppose rectangle weights
are available in $O(1)$ time.  For every $0<\varepsilon<1$, one can
deterministically construct a set $Z_\varepsilon\subseteq\operatorname{supp}(P)$
of $O(\varepsilon^{-2})$ points such that every axis-aligned rectangle $R$
with $w(P\cap R)\ge\varepsilon W$ contains a point of $Z_\varepsilon$.
The construction uses $O(\varepsilon^{-2}\log N)$ rectangle-weight queries.
The same conclusion holds for a union of at most $q$ rectangles after
replacing $\varepsilon$ by $\varepsilon/q$.
\end{lemma}

\begin{proof}
Set $\gamma:=\varepsilon/4$.  Using prefix weights and binary search,
partition the first coordinate into $s=O(1/\varepsilon)$ consecutive slabs.
Every slab has weight at most $\gamma W$, except that a coordinate of weight
larger than $\gamma W$ forms a singleton slab.  Such heavy coordinates are
also only $O(1/\varepsilon)$ in number.

Set $a:=\varepsilon W/(8s)$.  In every slab $I$, consider the weighted
one-dimensional distribution obtained by projecting $P\cap(I\times[N])$ to
the second coordinate.  At every positive multiple of $a$ in its cumulative
weight, select one support point of $P$ at the coordinate where that multiple
is reached.  A support point is recovered by a second binary search.  The number of selected points, with duplicates removed,
is
$O\left(s+\sum_I\frac{w(P\cap(I\times[N]))}{a}\right) =O(\varepsilon^{-2})$.

Let $R=J\times K$ have weight at least $\varepsilon W$.  At most two
non-singleton slabs meet $J$ without being contained in it, and together
they have weight at most $2\gamma W=\varepsilon W/2$.  A heavy singleton
slab is either contained in $J$ or disjoint from it.  The slabs contained in
$J$ therefore contribute at least $\varepsilon W/2$ to $R$.  Some contained
slab $I$ contributes at least $\varepsilon W/(2s)=4a$ inside $K$.  The
one-dimensional quantile construction places a selected support point in
$I\times K$.  This proves the rectangle claim.  If a union of $q$ rectangles
has weight at least $\varepsilon W$, one constituent rectangle has weight at
least $\varepsilon W/q$, proving the final statement.
\end{proof}

The prefix table also lets us recover an arbitrary positive-weight endpoint
pair in a queried rectangle in $O(\log n)$ time by binary-searching first its
row and then its column.  Hence \cref{lem:rectangle_net} is
constructive for the edge point sets used below.

\subsection{Repair Edges for the \texorpdfstring{$(2,3,3)$}{(2,3,3)} Pattern}
\label{subsec:repair_edges}

We now handle a good tree with degree pattern $(2,3,3)$ in which the
medium side $A$ is connected.  The forest $T \setminus A$ has two components
$K_1,K_2$, and each contains one outside endpoint of a tree edge leaving
$A$.  Define
$F_{A,T} := E_G(K_1,K_2)$.
Equivalently, if
$\delta_T(A)=\{p,s\}$,
then $F_{A,T}$ consists of the graph edges whose fundamental tree paths
contain both $p$ and $s$.

An edge $e\in F_{A,T}$ is called a \emph{repair edge} for the optimum
$(A,B,C)$ if its endpoints lie in the same one of $B,C$.  For such an edge,
$R_{A,e}:=(T \setminus A)+e$
is a spanning tree of $G \setminus A$, and the residual optimum does not cross $e$.
It therefore crosses exactly its two original $B$--$C$ tree edges in
$R_{A,e}$.

We first establish that repair edges form a constant fraction of
$F_{A,T}$ by capacity.  We use the following local-optimality observation.

\begin{observation}[Local optimality of an optimum side]
\label{obs:local_optimality}
Let $(A,B,C)$ be an optimum $3$-cut.  For every nonempty proper
$X\subsetneq A$,
$c(E_G(X,A\setminus X)) \ge c(E_G(X,B))$,
and
$c(E_G(X,A\setminus X)) \ge c(E_G(X,C))$.
Analogous inequalities hold for subsets of $B$ and $C$.
\end{observation}

\begin{proof}
Moving $X$ from $A$ to $B$ changes the partition value by
$c(E_G(X,A\setminus X))-c(E_G(X,B))$.
Optimality implies that this change is nonnegative.  Moving $X$ to $C$ gives
the second inequality.
\end{proof}

Every optimum side has boundary strictly below $5\lambda/6$.  Indeed, this
is immediate for $A$, and
$c_G(B)=2\lambda- c_G(A)- c_G(C)
<2\lambda-\lambda/2-2\lambda/3=5\lambda/6$. The argument for $C$ is symmetric.

\paragraph{Why repair edges are plentiful.}
The reason many such edges exist is that removing $A$ separates one residual optimum side into multiple tree components
that the graph itself must connect
substantially.  If $S_1\subseteq K_1$ and $S_2\subseteq K_2$ are the
two pieces of one optimum side $S$, then
\[
    2c(E_G(S_1,S_2))
    =
    c_G(S_1)+c_G(S_2)-c_G(S).
\]
The whole side $S$ has boundary below $5\lambda/6$, whereas, except
for the unique possible strict-small piece, its two fragments each have
boundary at least $\lambda/2$.  Their boundaries cannot both be large
while their union has small boundary unless
$\Omega(\lambda)$ capacity directly connects them.  If a strict-small fragment is
present, local optimality plays the same role: it prevents most of that
fragment's boundary from going to the other optimum parts and again
forces substantial same-side capacity across $K_1,K_2$.

Finally, every edge between $K_1$ and $K_2$ that is not a repair edge
is a true $B$--$C$ edge and these have low total capacity.

\begin{lemma}[Substantial repair capacity]
\label{lem:substantial_repair_capacity}
Let $(A,B,C)$ be an optimum $3$-cut with $A$ its unique medium side.  Suppose
$T$ has degree pattern $(2,3,3)$ and $T[A]$ is connected.  Let
$\mathcal R_{A,T}\subseteq F_{A,T}$
be the repair edges.  Then
$c(\mathcal R_{A,T})>\frac{\lambda}{18}$.
\end{lemma}

\begin{proof}
Remove the four optimum-crossing tree edges.  The resulting five tree
components are each contained in one of $A,B,C$.  Since $A$ is connected in
$T$, it contributes one component.  The two $B$--$C$ tree edges are either
in different components of $T \setminus A$ or in the same component.

\paragraph{Case 1: the two $B$--$C$ tree edges lie in different components
of $T \setminus A$.}
Both $B$ and $C$ have one tree component in each of $K_1,K_2$.  Among these
four components, at most one is strict-small by
\cref{lem:unique_strict_small_cut}.  Hence one of $B,C$, say $S$,
has components $S_1\subseteq K_1$ and $S_2\subseteq K_2$ satisfying
$c_G(S_1), c_G(S_2)\ge\frac{\lambda}{2}$.
The edges $E_G(S_1,S_2)$ are repair edges, and
$2c(E_G(S_1,S_2))= c_G(S_1)+ c_G(S_2)- c_G(S)
>\lambda-5\lambda/6=\lambda/6$.
Thus
$c(E_G(S_1,S_2))>\frac{\lambda}{12}$.

\paragraph{Case 2: the two $B$--$C$ tree edges lie in one component of
$T \setminus A$.}
Exactly one component of $T\setminus A$ contains both $B$--$C$ tree
edges, while the other contains none.  The latter component is therefore
contained entirely in one of $B,C$; after interchanging $B$ and $C$, call
it $C_2$.  Let $K$ be the other component of $T\setminus A$ and put
\[
    C_1:=C\cap K=C\setminus C_2.
\]
The repair edges in $F_{A,T}$ are exactly the edges joining $C_1$ and $C_2$.
Write
\[
    g_C:=c(E_G(C_1,C_2)).
\]

If neither $C_1$ nor $C_2$ is strict-small, then
\[
    2g_C=c_G(C_1)+c_G(C_2)-c_G(C)
    >\lambda-\frac{5\lambda}{6}
    =\frac{\lambda}{6},
\]
and hence $g_C>\lambda/12$.

Otherwise, let $R\in\{C_1,C_2\}$ be the unique strict-small cut.
By \cref{obs:local_optimality}, applied to moving $R$ from $C$ to
$A$ and to $B$,
\[
    g_C=c(E_G(R,C\setminus R))
    \ge c(E_G(R,A)),
    \qquad
    g_C\ge c(E_G(R,B)).
\]
Consequently,
\[
    c_G(R)
    =
    g_C+c(E_G(R,A))+c(E_G(R,B))
    \le3g_C.
\]
The valid partition
\[
    \bigl(R,B,V\setminus(R\cup B)\bigr)
\]
has value
\[
    c_G(R)+c_G(B)-c(E_G(R,B))\ge\lambda.
\]
Since $c_G(B)<5\lambda/6$, this implies
$c_G(R)>\lambda/6$, and therefore
\[
    g_C>\frac{\lambda}{18}.
\]
This completes all cases.
\end{proof}

\begin{corollary}[Constant repair fraction]
\label{cor:constant_repair_fraction}
Under the assumptions of \cref{lem:substantial_repair_capacity},
$c(\mathcal R_{A,T})>c(F_{A,T})/10$.
\end{corollary}

\begin{proof}
Every edge of $F_{A,T}$ is either a repair edge or a $B$--$C$ edge.  Moreover,
$c(E_G(B,C)) = \lambda- c_G(A) \le \frac{\lambda}{2}$.
Therefore
$c(F_{A,T}) \le c(\mathcal R_{A,T})+\frac{\lambda}{2}$.
Using $c(\mathcal R_{A,T})>\lambda/18$ gives
$\frac{c(\mathcal R_{A,T})}{c(F_{A,T})} > \frac{\lambda/18}{\lambda/18+\lambda/2} = \frac1{10}$.
\end{proof}

\begin{lemma}[Deterministic repair-edge set]
\label{lem:deterministic_repair_set}
There is an absolute constant $b$ such that, for every processed pair $(A,T)$
with $d_A=2$ and $T[A]$ connected, one can construct in $O(\log n)$ time a
set $Z_{A,T}\subseteq F_{A,T}$ of at most $b$ endpoint pairs with the
following property.  For every optimum $(A,B,C)$ realizing the connected
$(2,3,3)$ pattern, $Z_{A,T}$ contains a repair edge.
\end{lemma}

\begin{proof}
Each component $K_i$ of $T \setminus A$ satisfies $|\delta_T(K_i)|=1$ and is therefore
a union of $O(1)$ DFS intervals.  In the $(2,3,3)$ pattern, each of $B$ and
$C$ crosses three edges of $T$.  Hence each of
$B\cap K_1,B\cap K_2,C\cap K_1,C\cap K_2$ is also a union of $O(1)$ DFS
intervals.  The repair-edge point set is
$E_G(B\cap K_1,B\cap K_2)\mathbin{\dot\cup} E_G(C\cap K_1,C\cap K_2)$,
and is therefore supported on a union of at most $q_0$ axis-aligned
rectangles for an absolute constant $q_0$.

Apply \cref{lem:rectangle_net} to the weighted point set
$F_{A,T}=E_G(K_1,K_2)$ with $\varepsilon_0:=1/(10q_0)$.  By
\cref{cor:constant_repair_fraction}, the repair rectangles
have total weight greater than $c(F_{A,T})/10$, so one of them has weight
greater than $\varepsilon_0c(F_{A,T})$.  The net contains an endpoint pair in
that rectangle, which is a repair edge.  Since $\varepsilon_0$ is an absolute
constant, the net has constant size and is constructed with $O(\log n)$
rectangle queries.
\end{proof}

\subsection{An Optimizer for All Residual Trees}
\label{subsec:shared_residual_optimizer}

We now show that all direct, virtual, and repaired reference trees associated
with one sampled base tree can be optimized without rebuilding an
$\Theta(m)$-size data structure for every tag.

A \emph{tag} is a pair $(A,R)$, where $A$ is a processed candidate side with
$d_A\le2$ and $R$ is one of the following trees on $V\setminus A$:

\begin{enumerate}
    \item $R=T \setminus A$, when this graph is connected;

    \item $R=R_A^{\mathrm{virt}}$, when $T[A]$ is connected and $d_A=2$;

    \item $R=R_{A,e}=(T \setminus A)+e$ for an edge $e\in Z_{A,T}$.
\end{enumerate}

In the virtual case, regard the added link as a zero-capacity edge of the
residual graph.  In every case, $R$ is a spanning tree of the residual graph
$H_A:=G \setminus A$.

Fix the binary expansion $\widehat T$ of the sampled base tree and its
contraction map $\pi_T$.  For a candidate $A$, let
$\widehat A:=\pi_T^{-1}(A)$.  Associated with each tag $(A,R)$ is a
constant-degree expansion $\widehat R$: start with
$\widehat T-\widehat A$, and, in the virtual and repaired cases, add the
same link between its original endpoints.  Contracting the auxiliary
edges of $\widehat R$ recovers $R$.  Its maximum degree is at most four,
which is sufficient for the constant-branching centroid search.  The edges
corresponding to $E(R)$ are distinguished; the remaining edges are used
only by the search.

\begin{lemma}[Inherited boundary complexity]
\label{lem:inherited_boundary_complexity}
Let $(A,R)$ be a tag, let $f\in E(\widehat R)$, and let
$\widehat X_f$ be the vertex set of either component of
$\widehat R-f$.  Put $X_f:=\widehat X_f\cap(V\setminus A)$.  Then
\[
    \delta_{\widehat T}(\widehat X_f)
    \subseteq
    \delta_{\widehat T}(\widehat A)
    \cup\bigl(\{f\}\cap E(\widehat T)\bigr).
\]
Consequently,
$|\delta_{\widehat T}(\widehat X_f)|\le d_A+1\le3$, and $X_f$ is a
disjoint union of $O(1)$ intervals in the DFS order used by the capacity
matrix of $T$.
\end{lemma}

\begin{proof}
Let $g\in\delta_{\widehat T}(\widehat X_f)$.  If one endpoint of $g$
belongs to $\widehat A$, then
$g\in\delta_{\widehat T}(\widehat A)$.  Otherwise both endpoints survive
in $\widehat T-\widehat A$, which is a subgraph of $\widehat R$.  If
$g\ne f$, then $g$ remains in $\widehat R-f$ while joining its two
components, a contradiction.  Thus $g=f$, which is possible only when
$f\in E(\widehat T)$.

Since auxiliary edges have both endpoints in the same contraction gadget,
$|\delta_{\widehat T}(\widehat A)|=|\delta_T(A)|=d_A$.  The cardinality
bound follows, and \cref{lem:capacity_matrix} gives the interval
representation of $X_f=\widehat X_f\cap V$.
\end{proof}

Root $R$ and $\widehat R$ consistently.  For
$f\in E(\widehat R)$, let $X_f$ denote the original vertices in the
rooted subtree below $f$, and define
$d_A(f):=c_{H_A}(X_f)$.  If $X_f$ and $X_g$ are disjoint, then
\[
    c_{H_A}(X_f\triangle X_g)
    =
    d_A(f)+d_A(g)-2c(E_G(X_f,X_g)).
\]
If $X_g\subseteq X_f$, then
\[
    c_{H_A}(X_f\setminus X_g)
    =
    d_A(f)+d_A(g)
    -2c(E_G(X_g,(V\setminus A)\setminus X_f)).
\]
By \cref{lem:inherited_boundary_complexity}, every set in these expressions
is a union of $O(1)$ intervals in the one capacity matrix of the base tree.
Hence every displayed quantity is available in $O(1)$ time.

\begin{lemma}[Shared residual-tree optimizer]
\label{lem:shared_residual_tree_optimizer}
Fix one sampled tree $T$ and its capacity matrix.  Let
$\mathcal Q_T$
be any collection of tags derived from $T$.  For every tag $(A,R)$, one can
compute a minimum cut of $G \setminus A$ crossing at most two edges of $R$, together
with its value, in total time
$O\bigl(n|\mathcal Q_T|\log n\bigr)$.
The preprocessing shared by all tags is $O(m+n^2)$ for the base tree.
\end{lemma}

\begin{proof}
Consider one tag.  If the original base-tree root is outside $A$, root
$R$ there; otherwise root it at an arbitrary residual vertex.  Root
$\widehat R$ at the same original vertex.  Build the heavy-light
decomposition and LCA structure of $R$, and build the centroid
decomposition of $\widehat R$.

For every $f\in E(\widehat R)$, test ancestry using the DFS order in $\widehat R$
to determine which of the at most two outside endpoints of
$\delta_{\widehat T}(\widehat A)$ lie below $f$.  Together with $f$
itself when $f\in E(\widehat T)$, these are exactly the edges of
$\delta_{\widehat T}(\widehat X_f)$ identified in
\cref{lem:inherited_boundary_complexity}.  The rooted side
$\widehat X_f$ never contains the base-tree root: either that vertex is
the root of $R$, or it belongs to $\widehat A$.  Its crossing edges
therefore determine its side uniquely, and \cref{lem:capacity_matrix}
converts them into an $O(1)$-interval representation.  All these
representations are constructed in $O(n)$ total time.

The formulas preceding the lemma now provide the $Q=O(1)$ oracle required
by \cref{lem:oracle_two_respecting_optimizer}.  That lemma finds the
minimum residual cut crossing at most two distinguished edges, equivalently
at most two edges of $R$, in $O(n\log n)$ time.

Processing the tags sequentially and summing their running times gives
$O(n|\mathcal Q_T|\log n)$.  The $O(m+n^2)$ capacity-matrix
preprocessing is performed only once for the sampled base tree.
\end{proof}

The optimizer returns the minimum residual value
$\operatorname{res}(A,R):=\min\{ c_{G \setminus A}(X):
\emptyset\ne X\subsetneq V\setminus A,\ |\delta_R(X)|\le2\}$.
The resulting $3$-cut candidate has value
$c_G(A)+\operatorname{res}(A,R)$.
This value is computed from the capacity matrix.  We do not
scan the original edge list separately for every tag; only the final returned
partition is verified by one $O(m)$ scan.

\subsection{The Cases \texorpdfstring{$d_A\ge3$}{dA at least 3}}
\label{subsec:large_tree_degree}

Suppose $3\le d_A\le4$ and let
$F:=T \setminus A$.
By \cref{lem:tree_component_identity}, $F$ has at most four
components.  For a good tree, the residual optimum crosses at most one edge
of $F$.  This case can therefore be enumerated directly.

\begin{algorithm}[H]
\caption{$\textsc{EnumerateForestCompletion}(G,T,A)$}
\label{alg:enumerate_forest_completion}
\begin{algorithmic}[1]
\Require A graph $G$, a sampled tree $T$, and a candidate side $A$ with
$3\le |\delta_T(A)|\le4$.
\Ensure The best $3$-cut obtained from a cut of $T \setminus A$ crossing at most one
surviving tree edge.

\State Compute the connected components $K_1,\ldots,K_q$ of $T \setminus A$.
\State Initialize the best candidate to $+\infty$.
\For{each nontrivial bipartition of $\{K_1,\ldots,K_q\}$, up to complementation}
    \State Evaluate the corresponding residual cut and update the best candidate.
\EndFor
\For{each edge $f\in E(T \setminus A)$}
    \State Remove $f$, thereby splitting one component of $T \setminus A$ into two pieces.
    \State Let $\mathcal C_f$ be the resulting family of at most $q+1\le5$ pieces.
    \For{each nontrivial bipartition of $\mathcal C_f$, up to complementation}
        \State Evaluate the corresponding residual cut and update the best candidate.
    \EndFor
\EndFor
\State \Return the best candidate together with the side $A$.
\end{algorithmic}
\end{algorithm}

\begin{lemma}[Correctness and running time of forest completion]
\label{lem:forest_completion}
\cref{alg:enumerate_forest_completion} runs in
 $O(n)$
time after the preprocessing for $T$.  If a residual cut crosses at
most one edge of $T \setminus A$, the algorithm evaluates it or its complement.
\end{lemma}

\begin{proof}
There are at most $4$ components, so the first loop contains at most $8$
assignments.  Removing one edge creates at most $5$ pieces, so the inner
loop contains at most $16$ assignments for every $f$.  Hence the total
number of candidates is $O(n)$.

For every enumerated residual side $X$, all surviving tree edges crossing
$X$ belong to $\{f\}$, and every other tree edge crossing $X$ has one endpoint
in $A$.  Thus
$\delta_T(X) \subseteq \delta_T(A)\cup\{f\}$,
and
$|\delta_T(X)|\le5$.
The boundary-edge set and the side containing the tree root determine an
$O(1)$-interval representation by
\cref{lem:capacity_matrix}.  Therefore each residual value
is evaluated in $O(1)$ time.

If the target residual cut crosses no edge of $T \setminus A$, every component of
$T \setminus A$ lies wholly on one side and it occurs in the first loop.  If it crosses
one edge $f$, every component of $(T \setminus A)-f$ lies wholly on one side and it
occurs in the iteration associated with $f$.
\end{proof}

\subsection{Processing One Candidate Against One Tree}
\label{subsec:process_candidate}

For later use, we summarize the complete case analysis in one procedure.

\begin{algorithm}[H]
\caption{$\textsc{ProcessCandidate}(G,T,A)$}
\label{alg:process_candidate}
\begin{algorithmic}[1]
\Require A graph $G$, a sampled tree $T$ with its binary
expansion $\widehat T$, and a nontrivial candidate side $A$.
\Ensure A collection of valid $3$-cut candidates.

\State Compute $d_A:=|\delta_T(A)|$ and the components of $T[A]$ and $T \setminus A$.
\If{$d_A>4$}
    \State \Return the empty collection.
\ElsIf{$d_A\ge3$}
    \State \Return $\textsc{EnumerateForestCompletion}(G,T,A)$.
\ElsIf{$d_A=1$}
    \State Add the direct tag $(A,T \setminus A)$.
\Else
    \Comment{$d_A=2$}
    \If{$T \setminus A$ is connected}
        \State Add the direct tag $(A,T \setminus A)$.
    \Else
        \State Let $K_1,K_2$ be the two components of $T \setminus A$.
        \State Add the virtual tag $(A,(T \setminus A)+u_1u_2)$, where $u_i$ is the outside endpoint of $\delta_T(A)$ in $K_i$.
        \If{$c(E_G(K_1,K_2))>0$}
            \State Construct the constant-size set $Z_{A,T}$ from \cref{lem:deterministic_repair_set}.
            \For{each endpoint pair $e\in Z_{A,T}$}
                \State Add the repair tag $(A,(T \setminus A)+e)$.
            \EndFor
        \EndIf
    \EndIf
\EndIf
\State Run the shared residual-tree optimizer on all created tags.
\State Combine each returned residual cut $X$ with $A$ and output the best resulting $3$-cut.
\end{algorithmic}
\end{algorithm}

\begin{lemma}[Completeness of candidate processing]
\label{lem:process_candidate_complete}
Let $(A,B,C)$ be an optimum with unique medium side $A$, and let $T$ be good
for this optimum.  \cref{alg:process_candidate}, applied to the
true side $A$ and tree $T$, deterministically outputs a $3$-cut of value
$\lambda$.
\end{lemma}

\begin{proof}
If $d_A=1$, or if $d_A=2$ and $T \setminus A$ is connected, the direct tag is a tree
of $G \setminus A$ and the residual optimum crosses at most two of its edges by
\cref{lem:good_pattern_dichotomy}.

Suppose $d_A=2$ and $T \setminus A$ has two components.  By
\cref{lem:virtual_link_coverage}, the virtual tag covers every good
pattern except $(2,3,3)$.  In the remaining pattern,
\cref{lem:deterministic_repair_set} supplies a repair edge, and
then the residual optimum
crosses exactly two edges of the repaired tree.

Finally, if $d_A\ge3$, the residual optimum crosses at most one edge of the
forest $T \setminus A$ and is included by
\cref{lem:forest_completion}.

In every case, the produced residual cut has value at most
$c(E_G(B,C))$.  Therefore the combined partition has value at most
$c_G(A)+c(E_G(B,C))=\lambda$.
Optimality of $\lambda$ forces equality.
\end{proof}

\subsection{The Minimum \texorpdfstring{$3$}{3}-Cut Algorithm}
\label{subsec:3_cut_algorithm}

We now assemble the preceding ingredients.  All candidates are maintained
with their values, and the algorithm returns the lightest one.

\begin{algorithm}[H]
\caption{$\textsc{ThreeCut}(G)$}
\label{alg:three_cut}
\begin{algorithmic}[1]
\Require A connected undirected multigraph with positive polynomially
bounded integral capacities.
\Ensure A minimum $3$-cut with high probability.

\State Compute the strict-small candidate $\mathcal P_{\mathrm{small}}$.
\State Construct the random skeleton $H$ with $\varepsilon=1/100$.
\State Compute its tree packing with $\eta=1/100$.
\State Draw $s_{\mathrm{trees}}=\Theta(\log n)$ tree samples.
\State Construct $\mathcal F$, its longest laminar prefix $\mathcal L$,
and $\mathcal P_{\mathrm{cross}}$ as in
\cref{subsec:reporting_light_cuts}.
\State Initialize the best candidate with the lighter of
$\mathcal P_{\mathrm{small}}$ and $\mathcal P_{\mathrm{cross}}$.
\State Build the inclusion-forest order and capacity matrix for
$\mathcal L$.
\For{each unordered pair $S,T\in\mathcal L$}
    \State If $S,T$ are disjoint or nested, evaluate the partition of
    \cref{lem:two_cut_value} and update the best candidate.
\EndFor
\State Let $\mathcal A$ contain both sides of every cut in $\mathcal L$,
with duplicates removed.
\For{each $T\in\mathcal T$}
    \State Build its capacity matrix.
    \For{each $A\in\mathcal A$}
        \State Run $\textsc{ProcessCandidate}(G,T,A)$ and update the best
        candidate.
    \EndFor
\EndFor
\State \Return the best candidate.
\end{algorithmic}
\end{algorithm}

\begin{lemma}[Correctness of the $3$-cut algorithm]
\label{lem:3_cut_algorithm_correct}
\cref{alg:three_cut} returns a minimum $3$-cut
with probability $1-n^{-\Omega(1)}$.
\end{lemma}

\begin{proof}
Every candidate considered by the algorithm is an explicit $3$-cut.  Fix an
optimum partition.  If it has a strict-small side,
\cref{lem:small_side_completion} shows that
$\mathcal P_{\mathrm{small}}$ is optimum with high probability. The probability of the packing taking too long and failing is exponentially small.

Otherwise $\mathcal P_{\mathrm{small}}$ is not optimum.  Since it is a
valid $3$-cut, $c_G(\mathcal P_{\mathrm{small}})>\lambda$, so
\cref{cor:lambda_three_vs_lambda_two} and the guarantees of
\cref{subsec:sampling_trees} apply.  Moreover, every optimum side
has boundary at least $\lambda/2$.  Choose a light optimum side $A$,
which exists by averaging and satisfies $c_G(A)\le2\lambda/3$.  By
\cref{lem:sampled_light_list_complete}, with high probability either
$\mathcal P_{\mathrm{cross}}$ is optimum or the rooted orientation of
$(A,V\setminus A)$ belongs to $\mathcal L$.

If the optimum has two light sides,
\cref{lem:two_light_side_complete} shows that either the crossing
candidate or the pair scan recovers it.  Otherwise $A$ is the unique light
side, so it is the medium side considered in the completion analysis.  The
algorithm processes the actual side $A$, because it processes both
orientations of every cut in $\mathcal L$.  By
\cref{lem:constant_probability_good_tree},
$\mathcal T$ contains a tree good for the fixed optimum
with high probability, and
\cref{lem:process_candidate_complete} then makes
$\textsc{ProcessCandidate}$ output the optimum.

The required tree-sampling events concern only a constant number of optimum
sides and each fails with probability $n^{-\Omega(1)}$.  The skeleton
guarantee holds simultaneously for all cuts with the same probability, and
the global-minimum-cut calls in the strict-small step are amplified to that
success probability.  A union bound proves the claim.
\end{proof}

\begin{lemma}[Running time]
\label{lem:running_time}
\cref{alg:three_cut} runs in
$O((m+n^2)\log^2 n)$ time and uses $O(m+n^2)$ working space.
\end{lemma}

\begin{proof}
The strict-small candidate costs $O(m\log^2 n)$.  With
$\varepsilon=\eta=1/100$, constructing the random skeleton and its tree
packing costs $O(m\log^2 n)$ by
\cref{lem:fast_skeleton_packing}.  For each of $O(\log n)$ trees,
\cref{lem:respecting_enumeration} with $p=2$ costs
$O(m+n^2)$.  Materializing and deduplicating the retained
$O(n\log n)$ cuts costs $O(n^2\log n)$.  Each call to the
linear-incidence laminarity test costs $O(n^2)$; the binary search for the
first nonlaminar prefix therefore costs $O(n^2\log n)$.  Constructing the
single crossing certificate, the inclusion-forest capacity table, and all
pair candidates costs $O(m+n^2)$.

The laminar family has $O(n)$ members, so $|\mathcal A|=O(n)$.  There
are $O(\log n)$ trees and hence $O(n\log n)$
candidate--tree pairs and tags.  Direct tree scans and all forest-completion
calls cost $O(n^2\log n)$.  For one tree,
\cref{lem:shared_residual_tree_optimizer} handles all its tags in
$O(n^2\log n)$ time; over all trees this is
$O(n^2\log^2 n)$.  Constructing their capacity matrices costs
$O((m+n^2)\log n)$, and the deterministic repair nets are lower order.
These bounds sum to $O((m+n^2)\log^2 n)$.

The skeleton is stored in $O(m)$ space, and only one
$O(n^2)$ capacity table is needed at a time.  Processing sampled trees
sequentially and storing cuts by explicit incidence vectors gives
$O(m+n^2)$ working space.
\end{proof}

\begin{theorem}[Weighted Minimum $3$-Cut]
\label{thm:weighted_three_cut}
Let $G=(V,E,c)$ be a connected undirected capacitated multigraph with
positive polynomially bounded integral capacities.  There is a randomized
algorithm that returns a
minimum $3$-cut with probability $1-n^{-\Omega(1)}$ in
$O((m+n^2)\log^2 n)$ time and $O(m+n^2)$ working space.
\end{theorem}

\begin{proof}
Combine \cref{lem:3_cut_algorithm_correct} and
\cref{lem:running_time}.
\end{proof}

\section{Reducing Minimum \texorpdfstring{$k$}{k}-Cut to Minimum \texorpdfstring{$3$}{3}-Cut}
\label{sec:reduction_kcut}

\subsection{Algorithmizing the Light-Cut Bounds}
\label{subsec:constructive_overflow_certificate}

We make the light-cut counting arguments of~\cite{GHLL22} algorithmic.
Our light and medium ranges use non-strict upper bounds.  Consequently, an
overflowing family may yield an explicit $k$-cut of value at most $U$,
rather than only a bound on the size of the family.

Fix a connected capacitated multigraph $G=(V,E,c)$, an integer $k\ge 3$, and a threshold $U>0$.

\begin{definition}[Venn diagram and atoms]
\label{def:venn_diagram}
For cuts $S_1,\ldots,S_t$, let
$\operatorname{At}(S_1,\ldots,S_t)$
denote the nonempty atoms of their Venn diagram (the nonempty sets $\cap_i X_i$ where $X_i \in \{S_i, V \setminus S_i\}$).  The atom partition has value at most
$\sum_{i=1}^t c_G(S_i)$,
because every edge crossing the atom partition crosses at least one of the cuts $S_i$.
\end{definition}

All cuts are oriented away from a fixed root $r$.

\subsubsection{Small Cuts}

\begin{lemma}[Small-Cut Certificate]
\label{lem:small_cut_certificate}
There is a deterministic algorithm which, given a family $\mathcal S$ of distinct $U$-small cuts, either outputs a $k$-cut of value at most $U$, or certifies that $|\mathcal S|\le 2^k$. The running time is
$2^{O(k)}\widetilde O(|\mathcal S|\, n + m)$.
\end{lemma}

\begin{proof}
Maintain the atom partition generated by a selected sequence of cuts.  Initially the partition is $\{V\}$.  While the current partition has fewer than $k$ atoms, scan $\mathcal S$ and find a cut that splits some current atom.  If such a cut exists, add it to the sequence and refine the atom partition.

If the atom partition ever has at least $k$ atoms, suppose the selected cuts are
$S_1,\ldots,S_t, \qquad t\le k-1$.
Any $k$-part coarsening of the atom partition has value at most
$\sum_{i=1}^t c_G(S_i) \le t\cdot \frac{U}{k-1} \le U$.
Thus we output such a $k$-cut.

Otherwise the process stops with $t<k$ atoms, and no cut in $\mathcal S$ splits any atom.  Hence every cut in $\mathcal S$ is a union of atoms, up to the chosen canonicalization.  Thus
$|\mathcal S|\le 2^t\le 2^k$.
The running time follows because at most $k-1$ cuts are selected, each scan is linear in the explicit representations, and the final certificate is verified by scanning the edge list.
\end{proof}

\subsubsection{Sunflower-Core Search}

We use the following two standard combinatorial tools.  First, let
$\operatorname{Sun}(d,r)$ be a sunflower number for rank-$d$ set systems (a family of sets of size at most $d$):
every such family with more than
$\operatorname{Sun}(d,r)$ members contains an $r$-sunflower, which is a set of members such that all pairwise intersections are the same, this pairwise intersection is called the core. The members with the core removed are called the petals.
The bound of Bell, Chueluecha, and Warnke~\cite{BCW21} gives
$\operatorname{Sun}(d,r)\le (C r\log d)^d$
for an absolute constant $C$. The older Erd\H{o}s-Rado bound would work just as well for our applications.

Second, we use deterministic perfect hash families.  For every $N$ and
$q$, there is an explicit family
$\mathcal H_{N,q}\subseteq \{h:[N]\to[q]\}$
of size
$|\mathcal H_{N,q}|=e^q q^{O(\log q)}\log N$
such that for every $Y\subseteq[N]$ with $|Y|\le q$, some
$h\in\mathcal H_{N,q}$ is injective on $Y$ \cite{NSS95}.

We also use the following extension of the sunflower lemma.

\begin{lemma}[Sunflowers with distinct nonempty cores \cite{GHLL22}]
\label{lem:multiple_nonempty_core_sunflowers}
Let $\mathcal X$ be a rank-$k$ set system over an $N$-element universe.
Let $r\ge 2$ and $s\ge 1$.  If
$|\mathcal X|>\operatorname{Sun}(k,r)\,s\,N$,
then $\mathcal X$ contains $s$ sunflowers with $r$ petals and with
distinct nonempty cores.
\end{lemma}

\begin{lemma}[Nonempty-core sunflower search]
\label{lem:algorithmic_sunflower_search}
Let $\mathcal X$ be a family of distinct nonempty subsets of an $N$-element universe, with
$|X|\le k \qquad\text{for all }X\in\mathcal X$.
Let
$r:=k+1, \qquad s:=2^k+1$.
There is a deterministic algorithm which either returns $s$ distinct nonempty cores
$C_1,\ldots,C_s$
and, for each $C_j$, an $r$-sunflower in $\mathcal X$ with core $C_j$, or certifies
$|\mathcal X|\le \operatorname{Sun}(k,k+1)(2^k+1)N$.
The running time is
$k^{O(k^2)}\widetilde O(|\mathcal X|+N)$.
\end{lemma}

\begin{proof}
For every $X\in\mathcal X$, enumerate all nonempty subsets
$C\subseteq X$.
Insert $X$ into the bucket $B_C$, with associated petal
$P_C(X):=X\setminus C$.
The total number of bucket incidences is
$\sum_{X\in\mathcal X}(2^{|X|}-1) \le (2^k-1)|\mathcal X|$.

For a fixed bucket $B_C$, sets
$X_1,\ldots,X_r\in B_C$
form an $r$-sunflower with core $C$ if and only if the petals
$P_C(X_1),\ldots,P_C(X_r)$
are pairwise disjoint.  Therefore, for each bucket, it suffices to solve a
set-packing problem.

Let
$q:=kr\le k(k+1)$.
Construct $\mathcal H_{N,q}$.  Fix $h\in\mathcal H_{N,q}$.  For every petal $P$, define the color mask
$M_h(P):=\{h(v):v\in P\}\subseteq[q]$.
If petals have pairwise disjoint masks, then the petals are pairwise disjoint. Conversely, if $P_1,\ldots,P_r$ are pairwise disjoint, then
$|P_1\cup\cdots\cup P_r|\le q$,
so some $h\in\mathcal H_{N,q}$ is injective on their union, and for that $h$ the masks are pairwise disjoint.

For each bucket $B_C$ and hash function $h$, run the following dynamic
program.  Its states are $D[a,M]$, where $a\in\{0,\ldots,r\}$ and
$M\subseteq[q]$.  The state $D[a,M]$ means that $a$ petals have been selected with pairwise
disjoint masks and union of masks $M$.  Initialize
$D[0,\emptyset]=\textsc{true}$.
Process the petals one by one.  For a petal $P$, with mask $M_h(P)$, apply
the transition
\begin{equation*}
D[a,M]=\textsc{true}
    \quad\text{and}\quad
    M\cap M_h(P)=\emptyset
    \quad\Longrightarrow\quad
    D[a+1,M\cup M_h(P)]=\textsc{true},
\end{equation*}
in decreasing order of $a$.  Store predecessors.  If any state $D[r,M]$
becomes true, recover the corresponding $r$ petals.

For one hash function, the total number of processed bucket incidences is at
most $(2^k-1)|\mathcal X|$, and each incidence scans at most
$(r+1)2^q$
states.  Since
$q\le k(k+1)$,
and
$|\mathcal H_{N,q}|=e^q q^{O(\log q)}\log N$,
the total running time is
$e^q q^{O(\log q)}2^q2^k\,\widetilde O(|\mathcal X|+N) = k^{O(k^2)}\widetilde O(|\mathcal X|+N)$.

The algorithm records every nonempty core whose bucket contains an
$r$-packing of petals, and stops once $s=2^k+1$ distinct cores have been
recorded.  If fewer than $s$ are found, then by
\cref{lem:multiple_nonempty_core_sunflowers},
$|\mathcal X| \le \operatorname{Sun}(k,k+1)(2^k+1)N$.
\end{proof}

\subsubsection{Medium Cuts with Small Shores}

\begin{lemma}[Sunflower certificate for small shores]
\label{lem:sunflower_certificate}
Let $\mathcal M$ be a family of $U$-medium cuts such that every cut in
$\mathcal M$ has shore size at most $k$.  There is a deterministic
algorithm which either outputs a $k$-cut of value at most $U$, or certifies
$|\mathcal M| \le D_k n$,
where
$D_k:=\operatorname{Sun}(k,k+1)(2^k+1)$.
The running time is
$k^{O(k^2)}\widetilde O(|\mathcal M|\,n+m)$.
In particular, using $\operatorname{Sun}(k,k+1)\le (Ck\log k)^k$,
$D_k\le (C'k\log k)^k$.
\end{lemma}

\begin{proof}
For every $S\in\mathcal M$, orient $S$ to its shore $X(S)$, so
$|X(S)|\le k$.
Let
$\mathcal X:=\{X(S):S\in\mathcal M\}$,
deduplicated as a set family.

Run \cref{lem:algorithmic_sunflower_search} on $\mathcal X$.  If it
certifies
$|\mathcal X|\le D_k n$,
then the same bound holds for $\mathcal M$.

Otherwise it returns $2^k+1$ distinct nonempty cores
$C_1,\ldots,C_{2^k+1}$,
each supporting a $(k+1)$-sunflower in $\mathcal X$.  Compute
$c_G(C_j)$ for all returned cores.

If every returned core is $U$-small, namely
$c_G(C_j)\le \frac{U}{k-1} \qquad\text{for all }j$,
then \cref{lem:small_cut_certificate}, applied to
$\{C_1,\ldots,C_{2^k+1}\}$,
must output a $k$-cut of value at most $U$, because the alternative
certificate would imply that this family has size at most $2^k$.

Thus some returned core $C$ satisfies
$c_G(C)> \frac{U}{k-1}$.  Fix such a sunflower with $k+1$ petals
and write $X_i=C\cup P_i$ for $i\in[k+1]$.
The petals $P_i$ are pairwise disjoint.  Since the sets $X_i$ are distinct,
at most one petal is empty.  Hence at least $k$ petals are nonempty.  Keep
$k$ nonempty petals and relabel them
$P_1,\ldots,P_k$.

For each $i$, define
$L_i:=E_G(P_i,C)$.
The sets $L_1,\ldots,L_k$ are pairwise disjoint subsets of $\delta_G(C)$.
Sort the petals so that
$c(L_1)\le c(L_2)\le\cdots\le c(L_k)$.
Then
$\sum_{i=1}^{k-1} c(L_i) \le \frac{k-1}{k} c_G(C)$.

For each $i$, $C\cup P_i$ is $U$-light, so
$c_G(C\cup P_i)\le \frac{2U}{k}$.
By inclusion-exclusion,
$c_G(P_i)= c_G(C\cup P_i)- c_G(C)+2c(L_i)
\le2U/k- c_G(C)+2c(L_i)$.
Consider the $k$-cut generated by the disjoint nonempty sets
$P_1,\ldots,P_{k-1}$.
Its value is at most
\begin{equation*}
\begin{aligned}
    \sum_{i=1}^{k-1} c_G(P_i)
    &\le
    \frac{2(k-1)U}{k}
    -
    (k-1) c_G(C)
    +
    2\sum_{i=1}^{k-1}c(L_i)                                  \\
    &\le
    \frac{2(k-1)U}{k}
    -
    (k-1) c_G(C)
    +
    2\frac{k-1}{k} c_G(C)                               \\
    &=
    \frac{2(k-1)U}{k}
    -
    \frac{(k-1)(k-2)}{k} c_G(C)                          \\
    &<
    \frac{2(k-1)U}{k}
    -
    \frac{(k-1)(k-2)}{k}\cdot \frac{U}{k-1}                    \\
    &=
    U.
\end{aligned}
\end{equation*}
Hence the produced $k$-cut has value at most $U$.

The running time consists of the algorithmic sunflower-search time,
$\widetilde O(|\mathcal M|n)$ for orienting and deduplicating the shores, and $\widetilde O(m)$ for checking the final certificate.
\end{proof}

\subsubsection{Venn Certificates}

\begin{lemma}[Venn certificates]
\label{lem:venn_certificate}
Let $S_1,\ldots,S_t$ be $U$-medium cuts.

\begin{enumerate}
    \item If $k$ is even,
    $t=\frac{k}{2}, \qquad |\operatorname{At}(S_1,\ldots,S_t)|\ge k$,
    then the atom partition contains a $k$-cut of value at most $U$.

    \item If $k$ is odd, let
    $r:=\frac{k+1}{2}$.
    If
    $|\operatorname{At}(S_1,\ldots,S_{r-1})|\ge k$,
    then the atom partition of $S_1,\ldots,S_{r-1}$ contains a $k$-cut of
    value at most $U$.

    Otherwise, if
    $|\operatorname{At}(S_1,\ldots,S_{r-1})|=k-1$
    and
    $|\operatorname{At}(S_1,\ldots,S_r)|\ge k+1$,
    then one can construct a $k$-cut of value at most $U$.
\end{enumerate}
Given the cuts by their incidence vectors, the corresponding certificate can
be constructed deterministically in $O(kn+m)$ time and $O(kn)$
additional space.
\end{lemma}

\begin{proof}
If $k$ is even and $t=k/2$, the atom partition generated by
$S_1,\ldots,S_t$ has value at most
$\sum_{i=1}^{k/2} c_G(S_i) \le \frac{k}{2}\cdot\frac{2U}{k} = U$.
Any $k$-part coarsening is a valid certificate.

Now assume $k$ is odd and $r=(k+1)/2$.  If the first $r-1$ cuts generate
at least $k$ atoms, their atom partition has value at most
$\sum_{i=1}^{r-1} c_G(S_i) \le (r-1)\frac{2U}{k} = \frac{k-1}{k}U \le U$.
Any $k$-part coarsening is valid.

It remains to handle the case where the first $r-1$ cuts generate exactly
$k-1$ atoms and $S_r$ raises the number of atoms to at least $k+1$.  Let
$\mathcal A:=\operatorname{At}(S_1,\ldots,S_{r-1})$.
The cut $S_r$ splits at least two atoms $A_1,\ldots,A_j$ of $\mathcal A$,
where $j\ge 2$.  The edge sets
$E_G(S_r\cap A_i,A_i\setminus S_r), \qquad i=1,\ldots,j$,
are disjoint subsets of $\delta_G(S_r)$.  Hence some $A_i$ satisfies
$c(E_G(S_r\cap A_i,A_i\setminus S_r)) \le \frac{c_G(S_r)}{j} \le \frac{U}{k}$.
Refine the $(k-1)$-atom partition only by splitting this atom $A_i$.  The
old atom partition has value at most
$(r-1)\frac{2U}{k} = \frac{k-1}{k}U$,
and the additional split costs at most $U/k$.  The resulting $k$-cut has
value at most $U$.

For the algorithm, compute for each vertex its membership
signature with respect to $S_1,\ldots,S_t$ and group equal signatures in a
binary trie.  Since $t\le(k+1)/2$, this constructs and labels all atoms in
$O(kn)$ time and $O(kn)$ space.  In the first two cases, arbitrarily merge
atoms until exactly $k$ parts remain.  In the final case, scan the edges
crossing $S_r$ and accumulate, for every atom $A_i$ split by $S_r$, the
value $c(E_G(S_r\cap A_i,A_i\setminus S_r))$.  This takes $O(m+n)$ time
and identifies an atom of additional splitting cost at most $U/k$.
Constructing the resulting partition takes another $O(n)$ time, proving
the stated bounds.
\end{proof}

\subsubsection{Overflow Certificate}

\begin{algorithm}[H]
\caption{$\textsc{OverflowCertificate}(G,k,U,\mathcal F)$}
\label{alg:overflow_certificate}
\begin{algorithmic}[1]
\Require A connected capacitated multigraph $G=(V,E,c)$, an integer $k\ge3$, a threshold $U>0$, and a family $\mathcal F$ of distinct cuts satisfying $c_G(S)\le 2U/k$ for every $S\in\mathcal F$.
\Ensure Either a $k$-cut of value at most $U$, or a certificate that $|\mathcal F|\le C_k n$.

\State $\mathcal S\gets\{S\in\mathcal F: c_G(S)\le U/(k-1)\}$.
\State Run \cref{lem:small_cut_certificate} on $\mathcal S$.
\If{a $k$-cut $\mathcal P$ is returned}
    \State \Return $\mathcal P$.
\EndIf

\State $\mathcal M\gets\mathcal F\setminus\mathcal S$.

\If{every cut in $\mathcal M$ has shore size at most $k$}
    \State Run \cref{lem:sunflower_certificate} on $\mathcal M$.
    \If{a $k$-cut $\mathcal P$ is returned}
        \State \Return $\mathcal P$.
    \Else
        \State \Return the resulting size certificate.
    \EndIf
\EndIf

\State Choose $S_1\in\mathcal M$ whose two sides both have size larger than $k$.
\State Initialize $\mathcal Q\gets(S_1)$.

\While{there exists $T\in\mathcal M$ such that, for $\mathcal Q=(S_1,\ldots,S_t)$, $|\operatorname{At}(S_1,\ldots,S_t,T)|\ge 2(t+1)$}
    \State Append $T$ to $\mathcal Q$.
    \If{the current prefix triggers \cref{lem:venn_certificate}}
        \State \Return the resulting $k$-cut.
    \EndIf
\EndWhile

\State Let $\mathcal Q=(S_1,\ldots,S_\ell)$.
\State Let $A_1,\ldots,A_i$ be the atoms of $\operatorname{At}(\mathcal Q)$ contained in $S_1$.
\State Let $B_1,\ldots,B_j$ be the atoms of $\operatorname{At}(\mathcal Q)$ contained in $V\setminus S_1$.

\State Form $H_1$ by contracting each $A_a$ to one vertex.
\State Form $H_2$ by contracting each $B_b$ to one vertex.

\State Let $\mathcal F_1$ be the cuts in $\mathcal M$ that are constant on every $A_a$, projected to $H_1$.
\State Let $\mathcal F_2$ be the cuts in $\mathcal M$ that are constant on every $B_b$, projected to $H_2$.

\State Recursively run $\textsc{OverflowCertificate}(H_1,k,U,\mathcal F_1)$.
\If{a $k$-cut is returned}
    \State Lift it to $G$ and return it.
\EndIf

\State Recursively run $\textsc{OverflowCertificate}(H_2,k,U,\mathcal F_2)$.
\If{a $k$-cut is returned}
    \State Lift it to $G$ and return it.
\EndIf

\State Combine the two recursive size certificates and return the resulting size certificate.
\end{algorithmic}
\end{algorithm}

\begin{lemma}[Constructive overflow certificate]
\label{lem:constructive_overflow_certificate}
Let
$D_k:=\operatorname{Sun}(k,k+1)(2^k+1)$.
There is a constant
$C_k=O(kD_k+2^k)$
such that \cref{alg:overflow_certificate} has the following
guarantee.  Given a connected capacitated multigraph $G=(V,E,c)$, a threshold
$U>0$, and a family $\mathcal F$ of distinct cuts satisfying
$c_G(S)\le \frac{2U}{k} \qquad\text{for every }S\in\mathcal F$,
the algorithm either outputs a $k$-cut of value at most $U$, or certifies
$|\mathcal F|\le C_k n$.
With the bound $\operatorname{Sun}(k,k+1)\le (Ck\log k)^k$, we have
$C_k\le (C'k\log k)^k$.

The algorithm runs in
\[
 k^{O(k^2)}
 \bigl(|\mathcal F|\,n^2\log n+mn\bigr)
\]
time.
\end{lemma}

\begin{proof}
Let
$B_k:=2^k$.
We prove the statement that, for a sufficiently large
$\Gamma_k=O(kD_k+B_k)$,
every recursive instance with $n>k$ either outputs a $k$-cut of value at
most $U$, or certifies
$|\mathcal F|\le \Gamma_k(n-k)-B_k$.
This implies the stated bound after increasing constants.

The small-cut step is correct by \cref{lem:small_cut_certificate}.
If it does not output a $k$-cut, then its size certificate gives
$|\mathcal S|\le B_k$.

If every medium cut has shore size at most $k$, then
\cref{lem:sunflower_certificate} either outputs a $k$-cut of value
at most $U$, or certifies
$|\mathcal M|\le D_k n$.
Since $n >k$,
$n \le (k+1)(n -k)$.
Choosing
$\Gamma_k\ge (k+1)D_k+2B_k$
gives
$|\mathcal F| = |\mathcal S|+|\mathcal M| \le B_k+D_k n \le \Gamma_k(n-k)-B_k$.

It remains to consider the recursive case.  The algorithm has selected a
maximal sequence
$S_1,\ldots,S_\ell$
of medium cuts such that
$|\operatorname{At}(S_1,\ldots,S_t)|\ge 2t \qquad\text{for every }t\le \ell$,
and no prefix triggers \cref{lem:venn_certificate}.  Hence
$\ell<\frac{k}{2}$.
Let $A_1,\ldots,A_i$ be the atoms inside $S_1$, and let
$B_1,\ldots,B_j$ be the atoms outside $S_1$.  By maximality,
$i+j<2(\ell+1)\le k+1$,
and therefore
$i+j\le k$.

Both sides of $S_1$ have size larger than $k$, so both recursive graphs
$H_1$ and $H_2$ have fewer vertices than $G$ and more than $k$ vertices.
Moreover,
$|V(H_1)|+|V(H_2)| = n +i+j \le n +k$.

Every cut $T\in\mathcal M$ survives in at least one branch.  Indeed, if $T$
split some atom $A_a$ and also some atom $B_b$, then adding $T$ to
$\mathcal Q$ would split at least two current atoms, contradicting maximality
of $\mathcal Q$.  Hence
$|\mathcal M|\le |\mathcal F_1|+|\mathcal F_2|$.

By induction, unless a recursive call outputs a $k$-cut,
$|\mathcal F_1|\le \Gamma_k(|V(H_1)|-k)-B_k$,
and
$|\mathcal F_2|\le \Gamma_k(|V(H_2)|-k)-B_k$.
Therefore
$|\mathcal M|\le\Gamma_k(|V(H_1)|-k)-B_k
+\Gamma_k(|V(H_2)|-k)-B_k
=\Gamma_k(|V(H_1)|+|V(H_2)|-2k)-2B_k
\le\Gamma_k(n-k)-2B_k$.
Adding $|\mathcal S|\le B_k$ gives
$|\mathcal F| = |\mathcal S|+|\mathcal M| \le \Gamma_k(n-k)-B_k$.
This completes the induction.

It remains to analyze the running time.  Consider a recursive instance on a
graph $H$ with $N$ vertices and $M$ edges and a cut family
$\mathcal F_H$.  The small-cut and sunflower routines, the construction of
the maximal Venn sequence, the Venn-certificate tests from
\cref{lem:venn_certificate}, and the tests determining whether each cut
survives in either child can together be implemented in
\[
 k^{O(k^2)}
 \bigl(|\mathcal F_H|\,N\log N+M\bigr)
\]
time.  Here the atom labels are computed from the incidence vectors of the
cuts in the Venn sequence, and each remaining cut is tested against these
labels by scanning its incidence vector.

To sum this bound over the recursion tree, assign to an instance on $H$
the potential $|V(H)|-k$.  If its two recursive children are $H_1,H_2$,
then the construction above gives
\[
 (|V(H_1)|-k)+(|V(H_2)|-k)
 \le |V(H)|-k.
\]
Every recursive child has more than $k$ vertices, so every leaf has
positive integral potential.  The recursion tree therefore has $O(n)$
nodes.  Consequently, each input cut occurs in at most $O(n)$ recursive
families and each original edge occurs in at most $O(n)$ recursive
graphs.  Since every recursive graph has at most $n$ vertices,
\[
 \sum_H |\mathcal F_H|\,|V(H)|
 \le O(|\mathcal F|n^2),
 \qquad
 \sum_H |E(H)|\le O(mn).
\]
Summing the local bounds and using $\log |V(H)|\le\log n$ proves the
claimed running time.
\end{proof}

\begin{corollary}[Certifying overflow rule]
\label{cor:certifying_overflow_rule}
For a list of more than $C_kn$ distinct cuts $S$ satisfying
$c_G(S)\le 2U/k$,
$\textnormal{\textsc{OverflowCertificate}}$, applied to the first $C_kn+1$
such cuts, outputs an explicit $k$-cut of value at most $U$.
The additional time is
$k^{O(k^2)}\bigl(C_k n^3\log n+mn\bigr)$.
In particular, since $C_k\le (C'k\log k)^k$, this is
$k^{O(k^2)}(n^3\log n+mn)$.
\end{corollary}

\begin{proof}
The size-certificate outcome of
\cref{lem:constructive_overflow_certificate} is impossible for a
family of size greater than $C_kn$, so the algorithm returns a $k$-cut
of value at most $U$.  Applying the running-time bound of that lemma with
$|\mathcal F|=C_kn+1$ gives
$k^{O(k^2)}(C_kn^3\log n+mn)$.  The final bound follows from
$C_k\le(C'k\log k)^k$.
\end{proof}

\subsection{Sampling a Respecting Tree}
\label{sec:random_respecting_tree}

We use the following interface to the packing of \cite{Qua22}, we reprove it for convenience.

\begin{lemma}[Approximate dual tree packing {\cite{CQX19,Qua22}}]
\label{lem:dual_tree_packing}
Let \(F=(V,E,c)\) be a connected weighted graph, put
\(n:=|V|\), \(m:=|E|\), and let \(2\le k\le n\).  For every
\(0<\eta<1/2\), one can deterministically compute in
\[
    O\!\left(m\eta^{-2}\log^3(2n)\right)
\]
time an implicit finitely supported family of nonnegative tree weights
\((y_T)\) and nonnegative edge slacks \((z_e)\).  Writing
\[
    \tau:=\sum_Ty_T,
\]
they satisfy
\begin{align}
    \sum_{T\ni e}y_T
    &\le c(e)+z_e
        \qquad (e\in E),
        \label{eq:dual_packing_load}\\
    (k-1)\tau
    &\ge(1-\varepsilon)\frac{\lambda_k(F)}2+z(E).
        \label{eq:dual_packing_approx}
\end{align}
The support has
\[
    O\!\left(m\eta^{-2}\log^3(2n)\right)
\]
records.

The records can be sampled with probabilities proportional to their
weights.  Given \(r\) sampled record indices, their spanning-tree
extensions can be materialized in
\[
    O\!\left(m\eta^{-2}\log^3(2n)+rn\right)
\]
additional time.  The entire support can likewise be streamed with
\(O(n)\) output work per materialized tree.
\end{lemma}

\begin{proof}
Quanrud's algorithm approximately solves the positive forest-packing
dual
\[
    \max
    \sum_J(|J|+k-n)y_J
    \quad\text{subject to}\quad
    \sum_{J\ni e}y_J\le c(e)
    \quad(e\in E),
\]
where \(J\) ranges over forests of \(F\).  After changing its internal
accuracy by a constant factor, it returns a feasible packing satisfying
\[
    \sum_J(|J|+k-n)y_J
    \ge(1-\eta)\operatorname{OPT}_{\mathrm{LP}}
    \ge(1-\eta)\frac{\lambda_k(F)}2.
\]

For every support forest \(J\), choose a spanning-tree extension
\(T_J\supseteq J\), transfer the weight \(y_J\) to \(T_J\), and define
\[
    z_e:=\sum_{J:\,e\in T_J\setminus J}y_J.
\]
Then
\[
    \sum_{T\ni e}y_T
    =
    \sum_{J:\,e\in J}y_J+
    \sum_{J:\,e\in T_J\setminus J}y_J
    \le c(e)+z_e.
\]
Moreover,
\[
\begin{aligned}
    (k-1)\tau-z(E)
    &=
    \sum_J
    \bigl(k-1-|T_J\setminus J|\bigr)y_J
    =
    \sum_J
    \bigl(k-1-(n-1-|J|)\bigr)y_J\\
    &=
    \sum_J(|J|+k-n)y_J
    \ge(1-\eta)\frac{\lambda_k(F)}2.
\end{aligned}
\]

In Quanrud's implementation, every inserted forest is a prefix of the
current maintained minimum spanning tree.  We choose that tree as
\(T_J\).  Record the insertion index, prefix length, and inserted weight.
Prefix sums support weighted sampling of record indices.  Because the
packing algorithm is deterministic, sorting the sampled indices and
replaying its update sequence once reconstructs the corresponding MST
states; writing one tree explicitly costs \(O(n)\).  The same replay
streams the complete support.
\end{proof}

\begin{lemma}[A light side is $(k-1)$-respected with significant probability]
\label{lem:random_light_side_respecting}
Let $k\ge3$, let $F$ be connected, and let $(y,z)$ be the packing of
\cref{lem:dual_tree_packing} with $\eta\le1/30$.  If
$c_F(S)\le2\lambda_k(F)/k$ and $T$ is drawn with probability
$y_T/\tau$, then
\begin{equation*}
\Pr[|\delta_T(S)|\le k-1]\ge
 \begin{cases}
  \dfrac{1-9\eta}{6(1-\eta)},&k=3,\\[2mm]
  \dfrac{1-4\eta}{3(1-\eta)},&k=4,\\[2mm]
  \dfrac1{k-1},&k\ge5.
 \end{cases}
\end{equation*}
In particular, the probability is $\Omega(1/k)$.
\end{lemma}

\begin{proof}
Write $D:=|\delta_T(S)|$, $\lambda:=\lambda_k(F)$, and $Z:=z(E)$.  The
packing constraints give
\[
 \mathbb E[D]\le
 (k-1)\frac{2\lambda/k+Z}{(1-\eta)\lambda/2+Z}.
\]

Suppose first that $k=3$.  Then
$\mathbb E[D]\le 2(2\lambda/3+Z)/((1-\eta)\lambda/2+Z)
\le8/(3(1-\eta))$; the last inequality follows because the coefficients
of $\lambda$ agree after cross-multiplication, while
$2\le8/(3(1-\eta))$ for the coefficients of $Z$.  Let
$p:=\Pr[D\le2]$.  Since every spanning tree crosses a nontrivial cut at
least once, while $D\ge3$ whenever $D>2$, we have
$\mathbb E[D]\ge p+3(1-p)=3-2p$.  Consequently,
\[
 p\ge\frac{3-\mathbb E[D]}2
 \ge\frac12\left(3-\frac{8}{3(1-\eta)}\right)
 =\frac{1-9\eta}{6(1-\eta)}.
\]

Now suppose that $k=4$.  Similarly,
$\mathbb E[D]\le3(\lambda/2+Z)/((1-\eta)\lambda/2+Z)
\le3/(1-\eta)$.  With $p:=\Pr[D\le3]$, integrality gives
$\mathbb E[D]\ge p+4(1-p)=4-3p$, and hence
\[
 p\ge\frac{4-\mathbb E[D]}3
 \ge\frac13\left(4-\frac{3}{1-\eta}\right)
 =\frac{1-4\eta}{3(1-\eta)}.
\]

Finally, let $k\ge5$.  Since $\eta\le1/30$, we have
$2/k\le2/5<(1-\eta)/2$, and therefore
$2\lambda/k+Z\le(1-\eta)\lambda/2+Z$.  Thus
$\mathbb E[D]\le k-1$.  If $p:=\Pr[D\le k-1]$, then
$\mathbb E[D]\ge p+k(1-p)=k-(k-1)p$, so
$k-(k-1)p\le k-1$, or $p\ge1/(k-1)$.
\end{proof}
\subsection{A Randomized Reduction to Minimum \texorpdfstring{$3$}{3}-Cut}
\label{sec:randomized_reduction_kcut}

Throughout this section, all capacities are positive.  For every
nonempty graph $H$, set $\lambda_1(H):=0$, and set
$\lambda_j(H):=+\infty$ when $|V(H)|<j$.

We dispose of disconnected subproblems before constructing a packing.  If
$H$ has connected components $H_1,\ldots,H_t$, then
$\lambda_j(H)=0$ when $t\ge j$.  If $t<j$, then $\lambda_j(H)$ is
the minimum of $\sum_{a=1}^t\lambda_{j_a}(H_a)$ over integers $j_a\ge1$
with $\sum_a j_a=j$.  A minimizing composition also gives the partition.
Memoizing every pair $(H_a,j_a)$, and using that every index is at most
$k$, makes the component dynamic program contribute only a $k^{O(k)}$
factor.  We may therefore assume below that every graph on which a packing
is constructed is connected.

Let $N$ be the number of vertices in the original input, and let $k$
be its cut parameter.  Fix $\eta:=1/30$.

In every connected
invocation with parameter $j\ge4$, compute the packing
described above once and draw
$s_k:=\Theta(k^2\log N+k^3\log(k+1))$ independent record indices, for a sufficiently large absolute constant, with
probabilities proportional to their weights.  Materialize their
spanning-tree extensions. Since
$j\le k$, \cref{lem:random_light_side_respecting} implies that a
fixed light side is not $(j-1)$-respected by any sampled tree with
probability at most
$\exp(-\Omega(s_k/j))\le k^{-\Omega(k^2)}N^{-\Omega(k)}$.
The constant-probability tree-sampling events in the $3$-cut algorithm,
its random-skeleton construction, and its randomized minimum-cut calls are
amplified to the same failure bound; this changes only the
$k^{O(k^2)}$ factor in its running time.  For $j=2$, we use the
$O(m\log^2 n)$-time minimum-cut algorithm of~\cite{GMW20a}, amplified to
the same success probability.  Thus the only random choices outside the
$2$-cut and $3$-cut base cases are samples from the tree-packing distributions.

For an invocation on an $h$-vertex graph, set $M_4:=2h$, and for
$j\ge5$ set $M_j:=C_jh$, where $C_j$ is the constant from
\cref{lem:constructive_overflow_certificate}.  Recall that
$C_j\le(C'j\log j)^j$.

\begin{algorithm}[H]
\caption{$\textsc{MinimumKCut}(H,j;k,N)$}
\label{alg:minimum_kcut}
\begin{algorithmic}[1]
\Require A positive capacitated graph $H$ with $h$ vertices and $e$ edges, an integer $j\ge1$, and
the original parameters $k,N$
\Ensure An explicit minimum $j$-cut, or $+\infty$ if $|V(H)|<j$
\If{$|V(H)|<j$}
    \State \Return $+\infty$.
\EndIf
\If{$j=1$}
    \State \Return the one-part partition of $V(H)$.
\EndIf
\If{$H$ is disconnected}
    \State Solve the component dynamic program described above and \Return
    its partition.
\EndIf
\If{$j=2$}
    \State \Return a global minimum cut of $H$.
\EndIf
\If{$j=3$}
    \State \Return the algorithm of
    \cref{thm:weighted_three_cut}, with the amplification described
    above.
\EndIf
\State Construct an arbitrary $j$-partition and use it as the current best
candidate.
\State Compute the packing in
\cref{lem:dual_tree_packing} and sample $s_k$ record indices from its
normalized weights.
\State Materialize their spanning-tree extensions.
\State Fix $r\in V(H)$ and root every sampled tree at $r$.
\State Initialize an empty raw list $\mathcal R$.
\For{each sampled tree $T$}
    \State Apply \cref{lem:respecting_enumeration} with
    $p=j-1$ and retain the $M_j+1$ lightest rooted cuts, or all such cuts
    if fewer exist.
    \State Append the retained compact representations to $\mathcal R$.
\EndFor
\State Materialize and deduplicate $\mathcal R$, and let $\mathcal L$
be its $M_j+1$ lightest distinct cuts, or all distinct cuts if fewer exist.
\If{$|\mathcal L|=M_j+1$}
    \State Let $B:=\max_{A\in\mathcal L} c_H(A)$ and
    $U^*:=jB/2$.
    \If{$j=4$}
        \State Find a crossing pair in $\mathcal L$ and add the
        $4$-cut from \cref{lem:crossing_certificate_k34} to the
        candidate set.
    \Else
        \State Add the $j$-cut returned by
        $\textsc{OverflowCertificate}(H,j,U^*,\mathcal L)$ to the
        candidate set.
    \EndIf
\EndIf
\For{each $A\in\mathcal L$}
    \State Recursively compute a minimum $(j-1)$-cut of $H[A]$; if it
    is finite, combine it with the part $V(H)\setminus A$ and add the
    resulting $j$-cut to the candidate set.
    \State Recursively compute a minimum $(j-1)$-cut of
    $H[V(H)\setminus A]$; if it is finite, combine it with the part $A$
    and add the resulting $j$-cut to the candidate set.
\EndFor
\State \Return a minimum-value candidate.
\end{algorithmic}
\end{algorithm}

The crossing pair in the $j=4$ branch can be found without a specialized
data structure: materialize the $2h+1$ sides and test all pairs in $O(h^3)$ time.  Such a pair exists because a laminar family of nonempty
subsets of $V(H)\setminus\{r\}$ has fewer than $2h$ members.  Every cut
in $\mathcal L$ has value at most $B=2U^*/j$, so the crossing
certificate is applicable when $j=4$, while
\cref{cor:certifying_overflow_rule} is applicable when $j\ge5$.

\begin{lemma}[Reduction correctness]
\label{lem:randomized_kcut_correctness}
Let $H$ be connected and let $j\ge4$.  Suppose that all recursive 
calls made by \cref{alg:minimum_kcut} are correct.  Fix
a minimum $j$-cut of $H$, and let $S$ be one of its sides satisfying
$c_H(S)\le2\lambda_j(H)/j$.  If some sampled tree
$(j-1)$-respects the rooted orientation of $S$, then the invocation
returns a minimum $j$-cut.
\end{lemma}

\begin{proof}
Let $S_r$ be the side of the cut $(S,V(H)\setminus S)$ that excludes
the tree root.  Consider a sampled tree that $(j-1)$-respects $S_r$.
If $S_r\in\mathcal L$, one of the two recursive calls for $S_r$ cuts
the side containing the other $j-1$ optimum parts.  Combining that
recursive partition with the remaining optimum side produces a candidate
of value $\lambda_j(H)$.

Suppose instead that $S_r\notin\mathcal L$.  If it was discarded from the
list retained for its sampled tree, that local list contains $M_j+1$
distinct cuts of value at most $c_H(S)$.  If it survived locally
but was discarded globally, the same statement holds for the global
truncation.  In either case, $|\mathcal L|=M_j+1$ and its heaviest value
$B$ satisfies $B\le c_H(S)\le2\lambda_j(H)/j$.  Consequently
$U^*=jB/2\le\lambda_j(H)$.

When $j=4$, the crossing branch returns a $4$-cut of value at most
$U^*$.  When $j\ge5$,
\cref{cor:certifying_overflow_rule} returns a $j$-cut of value at
most $U^*$.  Since no $j$-cut has value below $\lambda_j(H)$, in both
cases this certificate has value $\lambda_j(H)$.  Thus the
candidate set again contains an optimum, and the algorithm returns one.
\end{proof}

\begin{lemma}[Work in one invocation]
\label{lem:randomized_kcut_invocation_time}
Consider a connected invocation with parameter $j\ge4$, $h$ vertices,
and $e$ edges.  Excluding its recursive calls, its running time is
\[
 k^{O(k^2)}\bigl(e\log^3N+s_k(h^2+h^{j-1})
                 +eh+h^3\log N\bigr).
\]
\end{lemma}

\begin{proof}
Computing the implicit packing and materializing the $s_k$
sampled spanning-tree extensions costs
\[
    O(e\log^3N+s_kh),
\]
since $\eta$ is fixed and $h\le N$.  For each of the $s_k$ materialized
trees, constructing the capacity matrix and applying
\cref{lem:respecting_enumeration} with $p=j-1$
costs $O(e+h^2+j^{O(1)}h^{j-1})$.  Materializing and deduplicating
$s_k(M_j+1)$ cuts costs $j^{O(j)}s_kh^2$.  Since
$s_ke\le k^{O(1)}e\log^3N$ and $s_kh\le s_kh^2$, these terms are
covered by the first two terms in the statement.

For $j=4$, the direct
crossing test costs $O(h^3)$.  For $j\ge5$,
\cref{cor:certifying_overflow_rule} costs
$j^{O(j^2)}(C_jh^3\log h+eh)$.  The bounds on $C_j$ and $j\le k$
give the stated expression.
\end{proof}

\begin{theorem}[Randomized reduction to Minimum $3$-Cut]
\label{thm:randomized_reduction_kcut}
For every fixed $k\ge3$, \cref{alg:minimum_kcut}
returns a minimum $k$-cut with probability $1-n^{-\Omega(1)}$.  On an
$n$-vertex, $m$-edge graph, it runs in
$k^{O(k^2)}n^{k-3}(m+n^2)\log^2 n$ time.  The randomness is confined to
the randomized $2$- and $3$-cut base cases and to sampling from the
tree-packing distributions used by the reduction.
\end{theorem}

\begin{proof}
Every minimum $j$-cut has a side $S$ with
$c_H(S)\le2\lambda_j(H)/j$, because the sum of the boundaries of
its $j$ sides is $2\lambda_j(H)$.  Conditional on the complete history
before a connected invocation, fix such a side of a fixed minimum cut.
The samples in the invocation are fresh, so
\cref{lem:random_light_side_respecting} and the definition of $s_k$
show that the hypothesis of
\cref{lem:randomized_kcut_correctness} fails with probability
at most $k^{-\Omega(k^2)}N^{-\Omega(k)}$.

Every invocation retains $O(C_jh)$ cuts and makes
$O(C_jh)$ recursive calls, all with parameter $j-1$.
Consequently, for a top-level
parameter $k$, the number of invocations with parameter $j$ is at most
$k^{O(k^2)}n^{k-j}$, and the total number of invocations is at most
$k^{O(k^2)}n^{k-3}$.  Increasing the absolute constant in $s_k$, and
the sampling constants in the base $3$-cut algorithm, makes a union bound
over all invocations $n^{-\Omega(1)}$.  Induction on $j$, using
\cref{lem:randomized_kcut_correctness}, proves correctness on
the complementary event.  Calls arising from the component dynamic program
do not increase these bounds: its connected components have disjoint vertex
and edge sets, and every recursive index is smaller than its parent index.

It remains to sum the running times.  At parameter $j\ge4$, 
tree enumeration contributes at most
$k^{O(k^2)}n^{k-j}\cdot n^{j-1}\log n
=k^{O(k^2)}n^{k-1}\log n$.

Since there are at most
$k^{O(k^2)}n^{k-j}$ parameter-$j$ invocations, the
$e\log^3N$ packing terms over all nonbase invocations sum to
at most
$k^{O(k^2)}n^{k-4}m\log^3 n$; the maximum occurs at $j=4$.
The $s_kh$ work for materializing the sampled trees is 
dominated by the $s_kh^2$ capacity-matrix work.  The capacity-matrix
terms are at most
$k^{O(k^2)}n^{k-2}\log n$.

Overflow calls occur only for $j\ge5$; their $eh$ terms sum to at most
$k^{O(k^2)}n^{k-4}m$, and their $h^3\log n$ terms sum to at most
$k^{O(k^2)}n^{k-2}\log n$.  Finally, there are at most
$k^{O(k^2)}n^{k-3}$ base $3$-cut calls.  By
\cref{thm:weighted_three_cut}, their total cost is at most
$k^{O(k^2)}n^{k-3}(m+n^2)\log^2 n$.

Since $\log n\le n$, the nonbase packing bound
$n^{k-4}m\log^3 n$ is dominated by
$n^{k-3}m\log^2 n$.  All other terms are likewise bounded by
$k^{O(k^2)}n^{k-3}(m+n^2)\log^2 n$.  The $2$-cut calls made inside
component dynamic programs cost $O(e\log^2 n)$ on disjoint components
and are dominated by the same bound.
\end{proof}

\section{Derandomizing the \texorpdfstring{$3$}{3}-Cut Algorithm}
\label{sec:derandomizing_three_cut}

This section removes the remaining randomness from
\cref{sec:faster_three_cut}.  The obstacle is that a
global minimum cut may have value below $\lambda_3(G)/2$ and therefore we cannot directly use the
guarantee of the skeleton of~\cite{HLRW24}.  We therefore use the unbalanced part of that
construction to locate an edge that can be contracted safely.  In the resulting quotient the
actual global minimum cut is at least $\lambda_3(G)/2$, after which the
ordinary skeleton lower guarantee is sufficient.

We first aggregate all parallel edges with the same endpoints.  Let
$\bar m\le\binom n2$ be the number of resulting weighted edges.  Aggregation
takes $O(m)$ time and preserves every cut; all routines below run on
the aggregated graph.

\subsection{Skeleton Construction}
\label{sec:constant_range_skeleton}

We need the following version of the intermediate and final
guarantees in~\cite{HLRW24} with constants adjusted.  The first graph in the statement is used only
to locate a contraction; the second is the skeleton on which we pack trees. We defer the proof to \cref{app:skeleton_construction}.

\begin{restatable}[Constant-range unbalanced sparsifier and skeleton]
    {theorem}{constantRangeSkeleton}
\label{thm:constant_range_skeleton}
Fix a constant $\alpha\ge1$ and let $0<\xi<1/10$.  Let $F$ be a connected
$n$-vertex, $m$-edge weighted graph with
$\lambda:=\lambda_2(F)>0$.  One can deterministically construct unweighted
multigraphs $J$ and $K$ on $V(F)$, together with scales $W_0,W>0$, in
$m(\log n/\xi)^{O_\alpha(1)}$ time.  The scales satisfy $W_0,W\le\lambda$
and
\begin{equation*}
 \max\left\{\frac{\lambda}{W_0},\frac{\lambda}{W}\right\}
 \le\left(\frac{\log n}{\xi}\right)^{O_\alpha(1)}.
\end{equation*}
For every cut $S$ satisfying $c_F(S)\le\alpha\lambda$, we have
$|W_0 c_J(S)-c_F(S)|\le\xi\lambda$.  Moreover, $K$ has at
most $n(\log n/\xi)^{O_\alpha(1)}$ edges, every nontrivial cut $S$
satisfies $W c_K(S)\ge(1-\xi)\lambda$, and every cut $S$ satisfying
$c_F(S)\le\alpha\lambda$ also satisfies
$W c_K(S)\le c_F(S)+\xi\lambda$.
\end{restatable}

\subsection{Eliminating the Unique Cut Below \texorpdfstring{$\lambda_3/2$}{lambda3/2}}
\label{sec:eliminating_small_cut}

Write $\mu:=\lambda_2(G)$ and $\lambda:=\lambda_3(G)$.  Use the
deterministic minimum-cut algorithm of~\cite{HLRW24} for the three calls in
the strict-small preprocessing, and let $(R,V\setminus R)$ be the global
minimum cut used there.  The numerical and structural consequences of
failure of this preprocessing were proved in
\cref{lem:strict_small_failure_structure} and
\cref{cor:lambda_three_vs_lambda_two}.  It remains to show that
the intermediate graph $J$ exposes an edge that can be contracted without
destroying a fixed optimum.

\begin{lemma}[The intermediate sparsifier contains a safe contraction]
\label{lem:safe_contraction_exists}
Suppose $c_G(\mathcal P_{\mathrm{small}})>\lambda$ and
$\mu<\lambda/2$, and fix an optimum $3$-cut.  Let
$R\subsetneq A$, $D:=A\setminus R$, and
$x:=c(E_G(R,D))$ be as in
\cref{lem:strict_small_failure_structure}.  If $J$ is the
intermediate graph supplied by
\cref{thm:constant_range_skeleton} with $\alpha=3$ and
$\xi=1/1000$, then $J$ contains an edge with one endpoint in $R$ and
the other in $D$.  Contracting the endpoints of any such edge preserves
the fixed optimum.
\end{lemma}

\begin{proof}
\cref{lem:strict_small_failure_structure} gives
$x>\mu/3$, $c_G(D)<3\mu$, and $c_G(A)<2\mu$; also
$c_G(R)=\mu$.  Thus the additive approximation guarantee for $J$
applies to $R,D,A$.  Since $A=R\mathbin{\dot\cup}D$, the cut identity
$2c(E_J(R,D))=c_J(R)+c_J(D)-c_J(A)$ and the three approximation bounds
give
$2W_0c(E_J(R,D))\ge2x-3\xi\mu>0$.
Hence $J$ contains an edge between $R$ and $D$.  Both endpoints lie
in the optimum part $A$, so contracting them preserves the fixed optimum.
\end{proof}

Construct the graph $J$ and let $\mathcal E$ be the set of distinct
endpoint pairs of its edges crossing $(R,V\setminus R)$.  Since $R$ is a
minimum cut, the upper guarantee gives
$|\mathcal E|\le c_J(R)\le(1+\xi)\lambda_2(G)/W_0=\log^{O(1)}n$.
Let $G_\bot:=G$, and for every $e=uv\in\mathcal E$ let $G_e:=G/uv$ be the
quotient obtained by identifying $u$ and $v$.

\begin{lemma}[One quotient has the required minimum-cut floor]
\label{lem:good_quotient}
Suppose $c_G(\mathcal P_{\mathrm{small}})>\lambda_3(G)$, and fix an optimum
$3$-cut of $G$.  Some
$a\in\{\bot\}\cup\mathcal E$ has both of the following properties: the
optimum survives in $G_a$ with the same value, and
$\lambda_2(G_a)\ge\lambda_3(G)/2$.  For this quotient also
$\lambda_2(G_a)\le2\lambda_3(G)/3$.
\end{lemma}

\begin{proof}
If $\lambda_2(G)\ge\lambda_3(G)/2$, take $a=\bot$.  Otherwise the hypotheses
of \cref{lem:safe_contraction_exists} hold.  Choose the edge
$e=uv$ supplied by that lemma.  Its endpoints lie in the same optimum part,
so the optimum survives the contraction.

Any cut of $G_e$ lifts to a cut of $G$ that does not separate $u$ and $v$.
If its value were below $\lambda_3(G)/2$ then, because $e$ crosses $(R,V\setminus R)$, it and the global minimum cut
$(R,V\setminus R)$ would be two distinct cut bipartitions of value below
$\lambda_3(G)/2$, contradicting
\cref{lem:unique_strict_small_cut}.  Thus
$\lambda_2(G_e)\ge\lambda_3(G)/2$.  Finally, every optimum has a side of
boundary at most $2\lambda_3(G)/3$; this side survives in the chosen
quotient and upper-bounds $\lambda_2(G_a)$.
\end{proof}

\subsection{The Deterministic Tree List}
\label{sec:deterministic_skeleton_tree_list}

The sampled skeleton and the faster packing of
\cref{subsec:sampling_trees} are used only by the randomized
algorithm.  For deterministic applications, the general dual packing
routine of \cref{lem:dual_tree_packing}  is used in
$O(m\log^3 n/\eta^2)$ time.  Here we obtain the smaller list needed for
$3$-cut directly by packing edge-disjoint trees in the deterministic
skeletons.

For every $a\in\{\bot\}\cup\mathcal E$, apply
\cref{thm:constant_range_skeleton} to $G_a$ with $\alpha=5/3$ and
accuracy $\xi=1/1000$, and denote the final skeleton and its scale by
$K_a,W_a$.  Put $\mu_a:=\lambda_2(G_a)$ and
$b_a:=\max\{1,\lceil2W_a/(\xi\mu_a)\rceil\}$.  Replace every edge of
$K_a$ by $b_a$ parallel copies and replace $W_a$ by $W_a/b_a$, reusing
$K_a,W_a$ for the resulting graph and scale.  This preserves every scaled
cut value $W_ac_{K_a}(S)$.  Moreover, $b_a=O(1)$ because the original
scale satisfies $W_a\le\mu_a$, and the new scale satisfies
$W_a\le\xi\mu_a/2$.

Let $q_a:=\lfloor\lambda_2(K_a)/2\rfloor$.  By Nash-Williams'
theorem~\cite{NW61}, $K_a$ contains $q_a$ edge-disjoint spanning trees,
which we find deterministically using Gabow's algorithm~\cite{Gab95} in $O(m \lambda_2(K_a) \log n)$ time.
Applying the skeleton guarantees to a minimum cut of $G_a$ gives
\[
 (1-\xi)\frac{\mu_a}{W_a}
 \le\lambda_2(K_a)
 \le(1+\xi)\frac{\mu_a}{W_a},
\]
so $\lambda_2(K_a)=\Theta(\mu_a/W_a)=\log^{O(1)}n$.  Finally,
\[
 q_aW_a
 \ge\left(\frac{\lambda_2(K_a)}2-1\right)W_a
 \ge\frac{1-\xi}{2}\mu_a-W_a
 \ge\frac{1-2\xi}{2}\mu_a.
\]

Lift every packed tree to a reference tree on $V(G)$.  For $a=\bot$ nothing
is changed.  If $a=uv$, replace the contracted vertex by $u$, attach all its
incident tree edges to $u$, and add the edge $uv$ to the tree.  All lifted
edges are assigned capacity zero when they are not edges of $G$.  This is
permitted by \cref{lem:oracle_two_respecting_optimizer}.  Let
$\mathcal K$ be the union of all lifted lists.  Both the number of quotients
and the number of trees per quotient are $\log^{O(1)}n$, so
$|\mathcal K|=\log^{O(1)}n$.

\begin{lemma}[Deterministic coverage of the light optimum sides]
\label{lem:skeleton_three_cut_coverage}
Suppose $c_G(\mathcal P_{\mathrm{small}})>\lambda_3(G)$.  Fix an optimum
$3$-cut and let $S$ be one of its sides with
$c_G(S)\le2\lambda_3(G)/3$.  Some $T\in\mathcal K$ satisfies
$|\delta_T(S)|\le2$.
\end{lemma}

\begin{proof}
Choose the quotient $G_a$ from \cref{lem:good_quotient}.  If a pair was
contracted, its endpoints lie in one optimum part, so $S$ induces a cut of
the same value in $G_a$.  Since
$\lambda_2(G_a)\ge\lambda_3(G)/2$, we have
$c_{G_a}(S)\le2\lambda_3(G)/3\le(4/3)\lambda_2(G_a)<(5/3)\lambda_2(G_a)$,
so the skeleton upper guarantee applies.  Using also
$\lambda_2(G_a)\le2\lambda_3(G)/3$, it gives
$\smash{W_a c_{K_a}(S)\le c_{G_a}(S)+\xi\lambda_2(G_a)
\le2(1+\xi)\lambda_3(G)/3}$.

The packed trees are edge-disjoint in $K_a$, and hence
$\sum_{T\subseteq K_a}|\delta_T(S)|\le c_{K_a}(S)$.  Moreover, \cref{lem:good_quotient} gives
$q_aW_a\ge(1-2\xi)\lambda_2(G_a)/2\ge(1-2\xi)\lambda_3(G)/4$.  Consequently,
\begin{equation*}
 \frac1{q_a}\sum_{T\subseteq K_a}|\delta_T(S)|
 \le\frac{W_a c_{K_a}(S)}{q_aW_a}
 \le\frac{2(1+\xi)\lambda_3(G)/3}
          {(1-2\xi)\lambda_3(G)/4}
 =\frac{8(1+\xi)}{3(1-2\xi)}<3,
\end{equation*}
where the final inequality holds for $\xi=1/1000$.  Since every crossing
number is integral, some packed tree crosses $S$ at most twice.  Lifting this
tree to $G$ preserves the crossing number because the contracted pair lies
in one optimum part.
\end{proof}

\begin{lemma}[The deterministic list contains a good tree]
\label{lem:deterministic_good_skeleton_tree}
Suppose $c_G(\mathcal P_{\mathrm{small}})>\lambda_3(G)$, and let $(A,B,C)$ be
an optimum $3$-cut in which $A$ is the unique medium side.  Some tree in
$\mathcal K$ is good for $(A,B,C)$ in the sense of
\cref{lem:constant_probability_good_tree}.
\end{lemma}

\begin{proof}
Choose the quotient $G_a$ from \cref{lem:good_quotient}.  Every optimum
side survives in $G_a$ with the same boundary.  As observed before
\cref{lem:substantial_repair_capacity}, each side has boundary
strictly below $5\lambda_3(G)/6$.  Since
$\lambda_2(G_a)\ge\lambda_3(G)/2$, this is strictly below
$(5/3)\lambda_2(G_a)$, so the skeleton upper guarantee applies to all three
sides.

For a packed tree $T$, let
$X_T:=|\delta_T(A)|+2|\delta_T(B)|+2|\delta_T(C)|$.  Edge-disjointness gives
\begin{align*}
 W_a\sum_{T\subseteq K_a}X_T
 &\le W_a\bigl(c_{K_a}(A)
          +2 c_{K_a}(B)+2 c_{K_a}(C)\bigr) \\
 &\le c_G(A)+2 c_G(B)+2 c_G(C)
          +5\xi\lambda_2(G_a).
\end{align*}
Because
$c_G(A)+ c_G(B)+ c_G(C)=2\lambda_3(G)$, the first three
terms equal $4\lambda_3(G)- c_G(A)$.  Moreover,
$c_G(A)\ge\lambda_3(G)/2$ because $A$ is medium, while
$\lambda_2(G_a)\le2\lambda_3(G)/3$ by
\cref{lem:good_quotient}.  Therefore
\begin{equation*}
 W_a\sum_{T\subseteq K_a}X_T
 \le4\lambda_3(G)-\frac{\lambda_3(G)}2
       +5\xi\frac{2\lambda_3(G)}3
 =\left(\frac72+\frac{10\xi}{3}\right)\lambda_3(G).
\end{equation*}
On the other hand,
$q_aW_a\ge(1-2\xi)\lambda_2(G_a)/2
\ge(1-2\xi)\lambda_3(G)/4$.  Dividing the preceding upper bound by this
packing lower bound gives
\begin{equation*}
 \frac1{q_a}\sum_{T\subseteq K_a}X_T
 \le
 \frac{(7/2+10\xi/3)\lambda_3(G)}
      {(1-2\xi)\lambda_3(G)/4}
 =\frac{14+40\xi/3}{1-2\xi}<15,
\end{equation*}
where again $\xi=1/1000$.  Hence some packed tree satisfies $X_T<15$.

Indeed, such a tree must be good.  If
$h:=|T\cap\delta_G(A,B,C)|\ge5$, then
$X_T=4h-|\delta_T(A)|\ge3h\ge15$.  If $h\le4$ but $T$ is not good, its
degree triple is $(1,3,4)$ or $(1,4,3)$, and in either case $X_T=15$.
Thus $X_T<15$ excludes every nongood pattern.  Finally, lifting the tree
preserves its degree triple because the added contraction edge lies inside
one optimum part.
\end{proof}

\subsection{The Deterministic \texorpdfstring{$3$}{3}-Cut Algorithm}
\label{sec:deterministic_three_cut}

We modify \cref{alg:three_cut}.  Use the
deterministic minimum-cut algorithm of~\cite{HLRW24} in the strict-small
preprocessing, omit the random skeleton and its packing, and replace
the sampled family $\mathcal T$ by the deterministically constructed
family $\mathcal K$.

Fix a common root $r$.  For every $T\in\mathcal K$, apply
\cref{lem:respecting_enumeration} with $p=2$, retaining
the $2n+1$ lightest nonempty rooted cuts that $2$-respect $T$, or all
of them if fewer exist.  Materialize the retained cuts, radix-sort their
incidence vectors to remove duplicates, and retain the globally lightest
$2n+1$ distinct cuts.  Order the resulting family $\mathcal F$ by
nondecreasing boundary value, using the same fixed tie-breaking rule as in
\cref{subsec:reporting_light_cuts}.  Use
\cref{lem:laminarity_test} to find its longest laminar prefix
$\mathcal L$ and construct $\mathcal P_{\mathrm{cross}}$ as in
that section.  Finally, use every tree in $\mathcal K$, rather than the
sampled family $\mathcal T$, in the completion stage.

\begin{theorem}[Deterministic Minimum $3$-Cut]
\label{thm:deterministic_three_cut}
A minimum $3$-cut of an $n$-vertex, $m$-edge weighted
multigraph with positive polynomially bounded integral capacities can be
found deterministically in
$O(m)+\widetilde O(n^2)$ time.
\end{theorem}

\begin{proof}
Every candidate considered by the modified algorithm is an explicit
$3$-cut whose value is computed.  Fix an optimum partition.  If it
has a strict-small side, \cref{lem:small_side_completion} shows that
the deterministically computed candidate
$\mathcal P_{\mathrm{small}}$ is optimum.

Suppose otherwise, and choose an optimum side $S$ satisfying
$c_G(S)\le2\lambda_3(G)/3$.  By
\cref{lem:skeleton_three_cut_coverage}, some tree in $\mathcal K$
$2$-respects the rooted orientation $S_r$ of this cut.  If $S_r$ is
discarded by either the local or global truncation, then
$\mathcal F$ contains $2n+1$ distinct cuts of boundary at most
$c_G(S)$.  If $S_r\in\mathcal F\setminus\mathcal L$, then every
cut in the first nonlaminar prefix also has boundary at most
$c_G(S)$.  In either case, the crossing certificate constructed
from that prefix has value at most
$3 c_G(S)/2\le\lambda_3(G)$, and hence is optimum.  Otherwise
$S_r\in\mathcal L$.

It follows that the crossing and two-light-side arguments of
\cref{subsec:two_light_sides} are deterministic.  If they do not
recover the optimum, the optimum has a unique medium side $A$.  The
algorithm processes both sides of every cut in $\mathcal L$, so it
processes $A$.  By \cref{lem:deterministic_good_skeleton_tree}, some
tree in $\mathcal K$ is good for this optimum, and
\cref{lem:process_candidate_complete} then recovers it, including
the deterministically chosen repair edge.  Thus the returned candidate is
a minimum $3$-cut.

Aggregate parallel edges in $O(m)$ time, and let $\bar m\le n^2$ be the
number of remaining edges.  The deterministic strict-small preprocessing,
the graph $J$, and all quotient skeletons take
$\widetilde O(\bar m)$ time in total.  The family $\mathcal K$ has
polylogarithmic size.  For each $T\in\mathcal K$,
\cref{lem:respecting_enumeration} costs
$O(\bar m+n^2)$; materializing, deduplicating, and finding the longest
laminar prefix costs $\widetilde O(n^2)$ in total.

There are $O(n)$ candidate sides and only
$\log^{O(1)}n$ completion trees.  The capacity matrices, residual tags,
deterministic repair sets, shared residual-tree optimizations, and forest
completions therefore cost another $\widetilde O(\bar m+n^2)$.  Since
$\bar m\le n^2$, the total running time is
$O(m)+\widetilde O(n^2)$.
\end{proof}

\section{Derandomizing the \texorpdfstring{$k$}{k}-Cut Algorithm}
\label{sec:derandomizing_reduction_overview}

The deterministic reduction uses a specialized algorithm for $k=4$ because the bound on respecting some tree in the tree packing is worse than for $k \ge 5$. It then enumerates light cuts, using the entire packing support after spectrally sparsifying, and finds the best $k$-cut refining each for $k\ge5$.  We first combine the
constant-range skeleton with safe contractions to obtain a deterministic
$O(m)+\widetilde O(n^3)$-time $4$-cut algorithm.

\subsection{Deterministic \texorpdfstring{$4$}{4}-Cut}
\label{sec:deterministic_four_cut}

The spectral-sparsifier argument used below for $k\ge5$ yields only a
$3$-respecting tree when $k=4$.  Its full packing support may contain
$\widetilde O(n)$ trees, so enumerating the $3$-respecting cuts of every
support tree would take $\widetilde O(n^4)$ time.  We instead reuse the two
graphs from \cref{thm:constant_range_skeleton}.  We use a recursive approach; when the current global
minimum cut is sufficiently large, the final skeleton supplies only
polylogarithmically many trees, one of which $3$-respects a light optimum
side.  Otherwise, the intermediate graph supplies a polylogarithmic-size
set of contractions, one of which preserves a fixed optimum while destroying that global minimum cut similar to \cref{lem:safe_contraction_exists}.  The
strict-small cut bound limits the latter branching to three levels.

We use the following bound on strict-light cuts (alternatively we could use \cref{lem:strict_small_cut_bound}).

\begin{lemma}[Strict-small cut bound {\cite{GHLL22}}]
\label{lem:strict_small_cut_bound}
Let $G$ be a graph with $|V(G)|\ge k$, let
$\lambda:=\lambda_k(G)$, and fix a root $r\in V(G)$.  There are at most
$2^{k-2}-1$ distinct nonempty sets $S\subseteq V(G)\setminus\{r\}$
satisfying $c_G(S)<\lambda/(k-1)$.
\end{lemma}

For a graph $Q$ and a nontrivial cut $A\subsetneq V(Q)$, write
$\overline A:=V(Q)\setminus A$ and define
\begin{equation}
\label{eq:four_cut_refinement_value}
 \Psi_4(Q,A):=c_Q(A)+\min\bigl\{
 \lambda_3(Q[A]),\
 \lambda_2(Q[A])+\lambda_2(Q[\overline A]),\
 \lambda_3(Q[\overline A])
 \bigr\}.
\end{equation}
As in the preliminaries, an infeasible term is interpreted as $+\infty$.

\begin{lemma}[Optimal refinement of a cut]
\label{lem:optimal_four_cut_refinement}
$\Psi_4(Q,A)$ is the minimum value of a $4$-cut of $Q$ that refines
$(A,\overline A)$.  A partition attaining this value can be found
deterministically in $O(|E(Q)|)+\widetilde O(|V(Q)|^2)$ time.
\end{lemma}

\begin{proof}
A $4$-cut refining $(A,\overline A)$ distributes its four parts between
the two sides as $(3,1)$, $(2,2)$, or $(1,3)$.  After paying $c_Q(A)$ for
all edges between the two sides, the minimum additional costs in these
three cases are respectively $\lambda_3(Q[A])$,
$\lambda_2(Q[A])+\lambda_2(Q[\overline A])$, and
$\lambda_3(Q[\overline A])$.  This proves the formula.  The algorithmic
claim follows from \cref{thm:deterministic_three_cut} and the deterministic
minimum-cut algorithm of~\cite{HLRW24}; the induced subgraphs have
edge-disjoint edge sets.
\end{proof}

We next describe the tree list used at a recursive node.  Let $Q$ be a
connected quotient, put $\mu:=\lambda_2(Q)$, and apply
\cref{thm:constant_range_skeleton} with $\alpha=13$ and
$\xi:=1/1000$, obtaining $J,K,W_0,W$.  We ensure $W\le\xi\mu/2$ as in \cref{sec:deterministic_skeleton_tree_list}, again using $K,W$ for the modified graph and
scale. We then identically compute a packing $\mathcal T_Q$ of value $q = \lfloor \lambda_2(K)/2 \rfloor$ satisfying $q W \ge \frac{1}{2} (1 - 2 \xi) \mu$ by Gabow's algorithm.

\begin{lemma}[Coverage at a well-scaled quotient]
\label{lem:four_cut_skeleton_coverage}
Let $\lambda:=\lambda_4(Q)$, and suppose $\lambda\le3\mu$.  If $S$ is a
side of a minimum $4$-cut satisfying $c_Q(S)\le\lambda/2$, then some
$T\in\mathcal T_Q$ satisfies $|\delta_T(S)|\le3$.
\end{lemma}

\begin{proof}
We have $c_Q(S)\le\lambda/2\le3\mu/2$, so the upper guarantee
applies and gives $Wc_K(S)\le c_Q(S)+\xi\mu\le(3/2+\xi)\mu$.
Because the trees in $\mathcal T_Q$ are edge-disjoint in $K$,
$\sum_{T\in\mathcal T_Q}|\delta_T(S)|\le c_K(S)$.  Combining this with
$q W \ge \frac{1}{2} (1 - 2 \xi) \mu$ yields
\begin{equation*}
 \frac1q\sum_{T\in\mathcal T_Q}|\delta_T(S)|
 \le
 \frac{(3/2+\xi)\mu}{(1-2\xi)\mu/2}
 =\frac{3+2\xi}{1-2\xi}<4.
\end{equation*}
The crossing numbers are integral, so one of the trees crosses $S$ at
most three times.
\end{proof}

The intermediate graph $J$ is used when the hypothesis of
\cref{lem:four_cut_skeleton_coverage} fails.

\begin{lemma}[Safe contraction for $4$-cut]
\label{lem:four_cut_safe_contraction}
Let $\mathcal P=(P_1,P_2,P_3,P_4)$ be a minimum $4$-cut of $Q$, of value
$\lambda$, and let $(R,\overline R)$ be a global minimum cut of value
$\mu$.  Suppose that $\mathcal P$ does not refine
$(R,\overline R)$ and that $\mu<\lambda/3$.  Then the intermediate graph
$J$ constructed above has an edge crossing $(R,\overline R)$ whose
endpoints belong to the same part of $\mathcal P$.
\end{lemma}

\begin{proof}
For $i\in[4]$, let $L_i:=P_i\cap R$, $U_i:=P_i\setminus R$, and
$x_i:=c(E_Q(L_i,U_i))$, and put $x:=\sum_i x_i$.  Let $y$ be the total
capacity of edges that cross both $(R,\overline R)$ and $\mathcal P$.
Every edge crossing $R$ is counted either by $x$ or by $y$, and hence
$\mu=x+y$.

Consider two nonempty atoms of the common refinement of $\mathcal P$ and
$(R,\overline R)$ that lie on the same side of $R$ but belong to distinct
parts of $\mathcal P$, and let $w$ be the capacity between them.  The
common refinement has value $\lambda+x$.  Merge the chosen pair of atoms
and then merge arbitrary pairs of atoms until four parts remain.  The
result is a $4$-cut of value at most $\lambda+x-w$.  Optimality of
$\mathcal P$ therefore implies $w\le x$.

There are at most $2\binom42=12$ such atom pairs.  Their interactions are
exactly the edges crossing $\mathcal P$ but not $(R,\overline R)$, whose
total capacity is $\lambda-y$.  Consequently
$\lambda-y\le12x$, or equivalently $\lambda\le\mu+11x$.  Since
$\lambda>3\mu$, it follows that $x>2\mu/11$, and hence some $i$ satisfies
$x_i>\mu/22$.  Moreover, $x\le\mu$ gives $\lambda\le12\mu$.  For this
index $i$, every edge of $\delta_Q(P_i)$ is cut by $\mathcal P$, and hence
\begin{equation*}
 c_Q(P_i)\le\lambda\le12\mu,
 \qquad
 c_Q(L_i),c_Q(U_i)\le c_Q(P_i)+x_i\le13\mu.
\end{equation*}
In particular, the additive guarantee for $J$ applies to all three cuts.
Using $P_i=L_i\mathbin{\dot\cup}U_i$ and the cut identity in $J$, we obtain
\begin{align*}
 2W_0c(E_J(L_i,U_i))
 &=W_0\bigl(c_J(L_i)+c_J(U_i)-c_J(P_i)\bigr)\\
 &\ge c_Q(L_i)+c_Q(U_i)-c_Q(P_i)-3\xi\mu\\
 &=2x_i-3\xi\mu>0,
\end{align*}
where the last inequality uses $\xi=1/1000<1/33$.  Thus $J$ has an edge
between $L_i$ and $U_i$.  This edge crosses $R$, while both endpoints lie
in the optimum part $P_i$.
\end{proof}

For a quotient $Q$, let $\mathcal E_Q$ be the set of distinct endpoint
pairs of the edges of $J$ crossing a fixed global minimum cut $R$.  The
upper guarantee for $J$ gives
\begin{equation}
\label{eq:four_cut_branching_bound}
 |\mathcal E_Q|\le c_J(R)
 \le\frac{(1+\xi)\mu}{W_0}
 =\log^{O(1)}|V(Q)|.
\end{equation}

\begin{algorithm}[H]
\caption{$\textsc{DeterministicFourCut}(G)$}
\label{alg:deterministic_four_cut}
\begin{algorithmic}[1]
\Require A connected positive-capacity graph $G$ with $|V(G)|\ge4$
\Ensure A minimum $4$-cut of $G$
\State Aggregate parallel edges and discard self-loops
\State Let $\mathcal P_{\mathrm{best}}$ be an arbitrary $4$-cut of $G$
\Procedure{Search}{$Q,d$}
    \If{$|V(Q)|<4$}
        \State \Return
    \EndIf
    \State Compute a global minimum cut $R$ of $Q$
    \State Compute a partition attaining $\Psi_4(Q,R)$ and use its lift to
           update $\mathcal P_{\mathrm{best}}$
    \State Construct $J,K$ and the packed tree family $\mathcal T_Q$ as above
    \State Fix a root $r_Q\in V(Q)$ and set $M_Q\gets2|V(Q)|+1$
    \State $\mathcal R_Q\gets\emptyset$
    \For{$T\in\mathcal T_Q$}
        \State Apply \cref{lem:respecting_enumeration} with $p=3$ and
               retain the $M_Q$ lightest nonempty rooted cuts of $Q$ that
               $3$-respect $T$, or all of them if fewer exist, using a
               fixed tie-breaking rule
        \State Append their compact representations to $\mathcal R_Q$
    \EndFor
    \State Materialize and deduplicate $\mathcal R_Q$ using exact
           capacities in $Q$
    \State Let $\mathcal L_Q$ be its $M_Q$ lightest distinct cuts, or all
           distinct cuts if fewer exist
    \For{$A\in\mathcal L_Q$}
        \State Compute a partition attaining $\Psi_4(Q,A)$ and use its lift
               to update $\mathcal P_{\mathrm{best}}$
    \EndFor
    \For{each pair $A,B\in\mathcal L_Q$ that crosses}
        \State Use the lift of the four-atom partition
               generated by $A$ and $B$ to update
               $\mathcal P_{\mathrm{best}}$
    \EndFor
    \If{$d<3$}
        \For{$uv\in\mathcal E_Q$}
            \State \Call{Search}{$Q/uv,d+1$}
        \EndFor
    \EndIf
\EndProcedure
\State \Call{Search}{$G,0$}
\State \Return $\mathcal P_{\mathrm{best}}$
\end{algorithmic}
\end{algorithm}

Every quotient in the algorithm retains its natural map to the original
vertices and is stored with parallel edges aggregated.  A cut or partition
is lifted by replacing every quotient vertex with its preimage.  The
capacity matrix and the respecting-cut enumeration are formed directly on
the quotient vertex set; only retained cuts and candidate partitions are
lifted.

\begin{theorem}[Deterministic Minimum $4$-Cut]
\label{thm:deterministic_four_cut}
A minimum $4$-cut of an $n$-vertex, $m$-edge weighted multigraph with
positive polynomially bounded integral capacities can be found
deterministically in $O(m)+\widetilde O(n^3)$ time.
\end{theorem}

\begin{proof}
If $G$ is disconnected, let its components be $G_1,\ldots,G_t$.  For
$t\ge4$, a zero-value $4$-cut is immediate.  For $t\in\{2,3\}$, minimize
$\sum_a\lambda_{j_a}(G_a)$ over positive integers $j_a$ with
$\sum_a j_a=4$.  Every $j_a$ is at most three, so the necessary values and
partitions are supplied by deterministic minimum-cut and $3$-cut
computations.  This takes $O(m)+\widetilde O(n^2)$ time.  We may therefore
assume that $G$ is connected and apply
\cref{alg:deterministic_four_cut}.

Every candidate considered by the algorithm is an explicit $4$-cut of a
quotient and therefore lifts to a $4$-cut of $G$ with the same value.  It
remains to prove that one considered candidate is optimum.

Fix a minimum $4$-cut $\mathcal P$ of $G$, of value $\lambda$.  We trace a
path in the recursion along which $\mathcal P$ survives.  At a node on this
path, let $Q$ be the current quotient, let $R$ be its global minimum cut,
and put $\mu:=c_Q(R)$.  Since every cut of $Q$ lifts to a cut of $G$, while
$\mathcal P$ induces a $4$-cut of $Q$ with value $\lambda$, we have
$\lambda_4(Q)=\lambda$.

If $\mathcal P$ refines $(R,V(Q)\setminus R)$, then
\cref{lem:optimal_four_cut_refinement} gives
$\Psi_4(Q,R)=\lambda$, so the algorithm considers an optimum candidate.
Suppose next that $\mu\ge\lambda/3$.  Choose a side $S$ of $\mathcal P$
with $c_Q(S)\le\lambda/2$.  By
\cref{lem:four_cut_skeleton_coverage}, some tree in $\mathcal T_Q$
$3$-respects the rooted orientation of this cut.

If that rooted cut belongs to $\mathcal L_Q$, its optimum refinement has
value $\lambda$ and is considered by the algorithm.  Otherwise it was
discarded either from the local list for its covering tree or from the
global list.
In either case, $\mathcal L_Q$ contains $M_Q=2|V(Q)|+1$ distinct rooted
cuts of capacity at most $c_Q(S)$.  A laminar family of nonempty subsets of
$V(Q)\setminus\{r_Q\}$ has fewer than $2|V(Q)|$ members, so two cuts
$A,B\in\mathcal L_Q$ cross.  Their four atoms form a $4$-cut of value at
most $c_Q(A)+c_Q(B)\le2c_Q(S)\le\lambda$.  It is therefore optimal and is
considered by the algorithm.

The remaining case is that $\mathcal P$ does not refine $R$ and
$\mu<\lambda/3$.  By \cref{lem:four_cut_safe_contraction}, some pair
$uv\in\mathcal E_Q$ lies in one part of $\mathcal P$.  The corresponding
child $Q/uv$ therefore preserves $\mathcal P$ and continues the path.

It remains to show that this last case cannot occur at all four depths
$0,1,2,3$.  Let $R_i$ be the lift to $G$ of the minimum cut encountered at
depth $i$ along the path.  Every $R_i$ has value below $\lambda/3$.
Moreover, these lifted cut bipartitions are distinct: the contraction
chosen after processing $R_i$ identifies a pair separated by $R_i$, while
every cut lifted from a later quotient keeps that pair together.  After
orienting all these cuts away from one fixed original vertex, four such
cuts would contradict \cref{lem:strict_small_cut_bound} with $k=4$.
Consequently, by depth three the algorithm considers an optimum candidate,
and $\mathcal P_{\mathrm{best}}$ is optimum when the algorithm terminates.

We now prove the running time.  After the initial $O(m)$ aggregation, every
quotient has at most $n^2$ weighted edges.  By
\eqref{eq:four_cut_branching_bound}, every recursion node has
$\log^{O(1)}n$ children, and the depth is at most three.  Thus the recursion
contains $\log^{O(1)}n$ nodes.

Consider a node on an $h$-vertex, $e$-edge quotient.  Constructing the two
skeleton graphs and their packed trees takes
$e\log^{O(1)}h$ time, and $|\mathcal T_Q|=\log^{O(1)}h$.  For each packed
tree, \cref{lem:respecting_enumeration} with $p=3$ takes
$O(e+h^2+h^3)$ time.  Hence all tree enumeration at the node costs
$\widetilde O(e+h^3)$.  Only $O(h)$ compact cuts are retained
before global truncation, so materializing and deduplicating them via radix sort costs
$O(h^2)$.  Testing every pair in $\mathcal L_Q$ for crossing
costs $O(h^3)$ by direct scans of their incidence vectors.

Finally, the algorithm computes $\Psi_4(Q,A)$ for $O(h)$ cuts.  By
\cref{lem:optimal_four_cut_refinement}, these calls take
$\widetilde O(he+h^3)=\widetilde O(h^3)$ time because $e\le h^2$ after
aggregation.  Thus every recursion node costs $\widetilde O(h^3)$, and the
polylogarithmic number of nodes costs $\widetilde O(n^3)$ in total.  Adding
the initial aggregation proves the theorem.
\end{proof}

\subsection{Deterministic Spectral Sparsification}
\label{sec:deterministic_spectral_sparsification}

The case $k\ge5$ will use the full support of a packing in the
deterministic spectral sparsifier of \cite{BSS12}, compensating for the increased number of trees with a better respecting guarantee.

\begin{definition}[Spectral sparsifier]
\label{def:spectral_sparsifier}
For a weighted graph $G=(V,E,c)$ and $0<\varepsilon<1$, a spectral sparsifier is a reweighted
subgraph $H$ such that
$(1-\varepsilon)x^\top L_Gx\le x^\top L_Hx\le
(1+\varepsilon)x^\top L_Gx$ for every $x\in\mathbb R^V$.
\end{definition}

In particular, every cut and every $j$-cut is preserved within the same
$(1\pm\varepsilon)$ factors.

\begin{definition}[Matrix multiplication exponent]
\label{def:mm_exponents}
The square matrix multiplication exponent $\omega$ is the infimum over all real numbers $\rho$ such that two $n\times n$ matrices can be multiplied using $n^{\rho}$ arithmetic operations.
\end{definition}

\begin{lemma}[Running time of the BSS construction]
\label{lem:bss-running-time}
Let $G$ be an $n$-vertex, $m$-edge weighted multigraph and let
$0<\varepsilon\le1/2$.  The deterministic algorithm of Batson, Spielman, and
Srivastava~\cite{BSS12} can be implemented to construct a
$(1\pm\varepsilon)$-spectral sparsifier with $O(n\varepsilon^{-2})$ edges in
\[
    O\!\left(m+n^{\omega+1}\varepsilon^{-2}\right)
\]
arithmetic operations.
\end{lemma}

\begin{proof}
Aggregate parallel edges and delete self-loops in $O(m)$ time, leaving
at most $n^2$ candidate edges.  It suffices to analyze one connected
component.  Associate with every edge $e$ the vector
$v_e:=\sqrt{w_e}L_G^{\dagger/2}b_e$.  Then
\[
    \sum_e v_e v_e^{\mathsf T}=I_{\operatorname{im}(L_G)}.
\]
The BSS algorithm performs $O(n\varepsilon^{-2})$ iterations, but these vectors
need never be formed explicitly.

Let $L_H$ be the current weighted sum of selected edge Laplacians and
$A:=L_G^{\dagger/2}L_HL_G^{\dagger/2}$.  For an upper barrier $u$, put
$K_u:=uL_G-L_H$.  Direct calculation on $\operatorname{im}(L_G)$ gives
\[
    v_e^{\mathsf T}(uI-A)^{-1}v_e
    =
    w_e b_e^{\mathsf T}K_u^\dagger b_e
\]
and
\[
    v_e^{\mathsf T}(uI-A)^{-2}v_e
    =
    w_e b_e^{\mathsf T}
    K_u^\dagger L_GK_u^\dagger b_e.
\]
The analogous identities hold at the lower barrier with
$K_\ell:=L_H-\ell L_G$.

In one iteration, compute the two pseudoinverses and the associated
products $K^\dagger L_GK^\dagger$ in $O(n^\omega)$ arithmetic
operations.  Once these dense matrices are available, every required
edge quadratic form has the form
$b_{uv}^{\mathsf T}Mb_{uv}=M_{uu}+M_{vv}-2M_{uv}$, and hence all
candidates can be checked in $O(n^2)$ time.  Thus one iteration costs
$O(n^\omega)$, and the total cost is
$O(m+n^{\omega+1}\varepsilon^{-2})$.

For a disconnected graph, apply the construction independently to its
components and use
$\sum_i n_i^{\omega+1}\le n^{\omega+1}$.
\end{proof}

\subsection{Deterministic Light-Cut Enumeration for
\texorpdfstring{$k\ge5$}{k at least 5}}
\label{sec:deterministic_light_cut_enumeration}

The reduction in \cref{sec:randomized_reduction_kcut} uses randomness
only to obtain a short collection of packing trees covering a light optimum
side.  We replace that collection deterministically.

Fix $\varepsilon=\eta=1/100$.  Compute the spectral sparsifier $H$ and then compute the
$k$-cut packing of \cref{lem:dual_tree_packing} in $H$.  Let
$\mathcal D_k$ be its full support.  As before,
$|\mathcal D_k|=\widetilde O(n)$.

\begin{lemma}[Deterministic coverage for $k\ge5$]
\label{lem:spectral_kcut_coverage}
Let $k\ge5$, and suppose
$c_F(S)\le2\lambda_k(F)/k$.  Some tree
$T\in\mathcal D_k$ satisfies
$|\delta_T(S)|\le k-2$.
\end{lemma}

\begin{proof}
Write $\lambda_H:=\lambda_k(H)$ and $Z:=z(E(H))$.  Spectral
approximation and $k\ge5$ give
$c_H(S) \le\frac{2(1+\varepsilon)}{k(1-\varepsilon)}\lambda_H <\frac{1-\eta}{2}\lambda_H$.
The packing guarantee now implies
$(k-1)\tau \ge\frac{1-\eta}{2}\lambda_H+Z >c_H(S)+Z$.
The load constraints give
$\tau\mathbb E[|\delta_T(S)|]\le c_H(S)+Z$, so the expectation is
strictly smaller than $k-1$.  Some support tree therefore crosses $S$
at most $k-2$ times.
\end{proof}

Fix a common root $r$.  For every $T\in\mathcal D_k$, apply
\cref{lem:respecting_enumeration} with $p=k-2$, retaining
the $C_kn+1$ lightest nonempty rooted cuts that $(k-2)$-respect $T$,
or all such cuts if fewer exist.  Materialize and deduplicate the retained
cuts, and let $\mathcal L_k$ consist of the globally lightest $C_kn+1$
distinct cuts, or all distinct cuts if fewer exist.

If $|\mathcal L_k|=C_kn+1$, let
$B:=\max_{A\in\mathcal L_k}c_F(A)$ and $U^*:=kB/2$, and add the
cut returned by
\[
    \textsc{OverflowCertificate}(F,k,U^*,\mathcal L_k)
\]
to the candidate
set.  This call is valid because every $A\in\mathcal L_k$ satisfies
$c_F(A)\le B=2U^*/k$.

\begin{lemma}[Cost of deterministic enumeration]
\label{lem:deterministic_kcut_light_enumeration_time}
After parallel edges are aggregated, the construction of
$\mathcal L_k$ and its overflow candidate, for $k\ge5$, takes
$k^{O(k^2)}n^{k-1}\log^{O(1)}n$ time.  Including aggregation, the bound is
$O(m)+k^{O(k^2)}n^{k-1}\log^{O(1)}n$.
\end{lemma}

\begin{proof}
The aggregated graph has at most $n^2$ edges.  The deterministic
sparsifier and its packing take
$\widetilde O(n^{\omega+1})+m$ time, and
$|\mathcal D_k|=O(n)$.  For each support tree,
\cref{lem:respecting_enumeration} with $p=k-2$ costs
$k^{O(1)}n^{k-2}$ time after its $O(n^2)$-time capacity matrix has been
constructed.  Over the full support this is
$k^{O(1)}n^{k-1}$.

There are at most $C_kn\,|\mathcal D_k|
 =C_kn^2$ retained compact cuts.  Materializing and
deduplicating them costs $k^{O(k)}n^3$.  The overflow
certificate costs $k^{O(k^2)}n^3\log^{O(1)}n$.  Since $k\ge5$, all
these terms, including $\widetilde O(n^{\omega+1})$, are dominated by
$k^{O(k^2)}n^{k-1}\log^{O(1)}n$.
\end{proof}

\subsection{The Derandomized Reduction}
\label{sec:derandomized_reduction}

As in \cref{sec:randomized_reduction_kcut}, every disconnected
subproblem is handled by a dynamic program before
light-cut enumeration.  All calls described below are therefore on
connected graphs.  The component subproblems are solved using the
deterministic base cases and the same derandomized recursion.

The deterministic reduction uses
\cref{thm:deterministic_four_cut} as its $k=4$ base case.
Consider an invocation on $F$ with parameter $j\ge5$.  Construct
the family $\mathcal L_j$ and, when it is full, its overflow candidate as
above.  Initialize the best candidate with an arbitrary $j$-cut and the
overflow candidate, when present.  For every $A\in\mathcal L_j$,
recursively compute a minimum $(j-1)$-cut of $F[A]$ and combine it with
the part $V(F)\setminus A$; also compute a minimum $(j-1)$-cut of
$F[V(F)\setminus A]$ and combine it with the part $A$.  Return the
lightest candidate.

\begin{lemma}[Correctness of the derandomized reduction]
\label{lem:derandomized_reduction_correct}
The reduction above returns a minimum $k$-cut deterministically for every
fixed $k\ge4$.
\end{lemma}

\begin{proof}
The claim follows by induction on $k$.  The base case is
\cref{thm:deterministic_four_cut}.  Suppose $k\ge5$, fix a
minimum $k$-cut, and choose one of its sides $S$ satisfying
$c_F(S)\le2\lambda_k(F)/k$.  By
\cref{lem:spectral_kcut_coverage}, some tree in $\mathcal D_k$
$(k-2)$-respects the rooted orientation $S_r$ of this cut.

If $S_r\in\mathcal L_k$, one of the two recursive calls associated with
$S_r$ cuts the side containing the other $k-1$ optimum parts.  By the
induction hypothesis, that call returns their minimum $(k-1)$-cut, so the
resulting candidate has value $\lambda_k(F)$.

Suppose $S_r\notin\mathcal L_k$.  If it was discarded locally, the list
retained for its support tree contains $C_kn+1$ distinct cuts of boundary
at most $c_F(S)$.  If it survived locally but was discarded
globally, the same statement holds for the global truncation.  Hence
$|\mathcal L_k|=C_kn+1$ and
$B\le c_F(S)\le2\lambda_k(F)/k$.  Therefore
$U^*=kB/2\le\lambda_k(F)$.  The overflow rule returns a $k$-cut of value
at most $U^*$, which must consequently have value exactly
$\lambda_k(F)$.  In either case the candidate set contains an optimum,
and every candidate is a valid $k$-cut, proving the claim.
\end{proof}

\begin{theorem}[Derandomized reduction to Minimum $3$-Cut]
\label{thm:derandomized_reduction}
For every fixed $k\ge4$, Minimum $k$-Cut can be solved deterministically
in
$k^{O(k^2)}(n^{k-4}m+n^{k-1})\log^{O(1)}n$ time.
In particular, this is
$k^{O(k^2)}n^{k-3}(m+n^2)\log^{O(1)}n$.
\end{theorem}

\begin{proof}
Every nonbase invocation with parameter $j$ retains $O(C_jn)$ cuts and
makes $O(C_jn)$ recursive calls, all with parameter $j-1$.  Thus the
number of parameter-$j$ invocations below a top-level parameter $k$ is
at most $k^{O(k^2)}n^{k-j}$.

By \cref{lem:deterministic_kcut_light_enumeration_time}, the
nonrecursive work of a parameter-$j$ invocation is
$O(m_H)+j^{O(j^2)}n_H^{j-1}\log^{O(1)}n$, where $H$ is its input graph.
Replacing local parameters by the top-level bounds and summing over all
parameter-$j$ invocations gives
$k^{O(k^2)}n^{k-1}\log^{O(1)}n$ for enumeration and
$k^{O(k^2)}n^{k-4}m\log^{O(1)}n$ for graph aggregation and edge scans.
There are at most $k^{O(k^2)}n^{k-4}$ parameter-$4$ leaves.  By
\cref{thm:deterministic_four_cut}, their total cost is at most
\[
    k^{O(k^2)}
    \bigl(n^{k-4}m+n^{k-1}\bigr)\log^{O(1)}n.
\]

For $j\ge5$, the total spectral-sparsifier cost is at most
$k^{O(k^2)}n^{k-4+\omega}\log^{O(1)}n$, which is dominated by the
$n^{k-1}$ term because $\omega<3$.  Combining the bounds proves the
theorem.
\end{proof}

\paragraph{Acknowledgments.}
The author thanks Jason Li for suggestions regarding the exposition of this paper.

\appendix

\section{Deferred Proofs for the Skeleton Construction}
\label[appendix]{app:skeleton_construction}

We restate the adjusted version of the skeleton theorem of \cite{HLRW24} for convenience.

\constantRangeSkeleton*

\begin{proof}
Put \(n:=|V(F)|\), \(m_{\mathrm{in}}:=|E(F)|\), and
\(\lambda:=\lambda_2(F)\).  We use the notation of~\cite{HLRW24}:
$d_H^W(v)$ and $\operatorname{vol}_H^W(X)$ denote weighted degree and
volume in a weighted graph $H$, $w_H(X,Y):=c(E_H(X,Y))$,
and, for pairwise disjoint sets $X_1,\ldots,X_t$,
$w_H(X_1,\ldots,X_t)$ denotes the total capacity of edges joining
distinct sets in this family.  All constants hidden below may depend on
$\alpha$.

Aggregate parallel edges and delete self-loops in \(O(m_{\mathrm{in}})\)
time, and let \(\bar m\) be the number of remaining edge pairs.  Thus
\[
    \bar m\le \min\left\{m_{\mathrm{in}},\binom n2\right\}.
\]
We compute \(\lambda\) on this aggregated graph using~\cite{HLRW24} and
set \(\beta:=\lceil\alpha\rceil+2\).  Replace each remaining capacity by
\(\overline c(e):=\min\{c(e),\beta\lambda\}\).
No cut of value at most $\alpha\lambda$ contains an edge whose capacity
is changed.  Conversely, every cut containing such an edge has truncated
value at least $\beta\lambda>\lambda$.  Thus truncation preserves both
the minimum-cut value and every target cut.  We again denote the resulting
graph by $F$.  It satisfies
\begin{equation}
\label{eq:initial_volume_concise}
    \bar m\le\binom n2,
    \qquad
    \operatorname{vol}_F^W(V(F))\le2\beta \bar m\lambda.
\end{equation}
After every contraction, we likewise delete self-loops and aggregate
parallel edges.  Consequently every quotient has simple support, although
an aggregated edge may have capacity greater than $\beta\lambda$.

Let $\theta>0$ be a sufficiently small constant multiple of $\xi$,
to be fixed at the end.  Consider a quotient $G_j$ in the decomposition
hierarchy.  Every nontrivial cut of $G_j$ lifts to a cut of $F$, and
hence
\begin{equation}
\label{eq:quotient_connectivity_concise}
    \lambda_2(G_j)\ge\lambda,
    \qquad
    d_{G_j}^W(v)\ge\lambda
\end{equation}
for every vertex of a nontrivial quotient.

Set
\[
    \widetilde\delta:=\frac{\beta\lambda}{1.1}.
\]
Add an auxiliary self-loop at each vertex whose weighted degree is below
$\widetilde\delta$, of the weight needed to raise its degree to
$\widetilde\delta$.  If $G_j^+$ denotes the padded graph, then
\[
    \operatorname{vol}_{G_j^+}^W(V(G_j))
    \le
    \operatorname{vol}_{G_j}^W(V(G_j))
      +|V(G_j)|\widetilde\delta
    =
    O_\alpha\!\left(
        \operatorname{vol}_{G_j}^W(V(G_j))
    \right)
\]
by~\eqref{eq:quotient_connectivity_concise}.  The loops alter no cut.
Moreover, deleting the padding can only decrease the volume of either
shore, so every cluster that is $s_0$-strong in $G_j^+$ is also
$s_0$-strong in $G_j$.

Apply Lemma~5.1 of~\cite{HLRW24} to $G_j^+$.  Its proof permits
$s_0=10^{11}\widetilde\delta\tau^2$ for every sufficiently large
$\tau$.  Taking
\begin{equation}
\label{eq:s0_concise}
    \tau=(\log n/\theta)^{C_\alpha},
    \qquad
    s_0=\lambda(\log n/\theta)^{O_\alpha(1)}
\end{equation}
therefore produces clusters that are $s_0$-strong for every cut of
value at most $1.1\widetilde\delta=\beta\lambda$.
Their total intercluster boundary is at most
\[
    O\!\left(
        \sqrt{\frac{\widetilde\delta\log n}{s_0}}
    \right)
    \operatorname{vol}_{G_j^+}^W(V(G_j))
    =
    O_\alpha\!\left(\frac{\sqrt{\log n}}{\tau}\right)
    \operatorname{vol}_{G_j^+}^W(V(G_j)),
\]
and the running time is $|E(G_j)|(\log n/\theta)^{O_\alpha(1)}$.
We next establish the two constant-range modifications needed to refine
these clusters.  In these two local arguments, write $G:=G_j$.

\smallskip
\noindent\emph{Small-cluster claim.}
For every cluster $A$ and every sufficiently small
$\varepsilon>0$, the routine of Lemma~6.1 of~\cite{HLRW24} can be
modified to produce a partition $A=A_1\mathbin{\dot\cup}\cdots\mathbin{\dot\cup}A_k$
with the following properties.  For every $S\subseteq A$ satisfying
$c_{G[A]}(S)\le\beta\lambda$, there is a partition $\mathcal P$ of
$\{A_1,\ldots,A_k\}$ such that, for every $P\in\mathcal P$, the set
\[
    S\cap A_P,
    \qquad
    A_P:=\bigcup_{B\in P}B,
\]
is non-$(1-\varepsilon)$-boundary-sparse in $A_P$.  Moreover,
\begin{equation}
\label{eq:small_boundary_concise}
    \sum_{i=1}^k c_G(A_i)
    \le
    \varepsilon^{-O_\alpha(1)}
    (\log |A|)^{O_\alpha(1)}c_G(A),
\end{equation}
and the routine runs in $|A|^{O_\alpha(1)}$ time.

We verify the claim by adjusting the two stages in the proof of
Lemma~6.1.  Put $\rho:=1/4$.  During the first stage, process an active
set $A'$ only when
\[
    c_G(A')>(2\beta+2)\lambda.
\]
First split along an induced cut of value at most $\rho\lambda$, if one
exists.  Otherwise the induced minimum cut is greater than
$\rho\lambda$, so Theorem~6.2 of~\cite{HLRW24}, used with the constant
approximation parameter $\beta/\rho$, lists every induced cut of value
at most $\beta\lambda$.  Split along any listed
$(1-\varepsilon)$-boundary-sparse cut, or declare $A'$ finished if
none exists.

For the current partition $\mathcal D$, define
\[
    \Phi(\mathcal D)
    :=
    \sum_{A'\in\mathcal D}
    \max\{0,c_G(A')-(2\beta+1)\lambda\}.
\]
Suppose $A'$ is split into $U$ and $A'\setminus U$, and write
\[
    x:=c_{G[A']}(U),
    \qquad
    y:=w_G(U,V(G)\setminus A').
\]
If both child boundaries are above the breakpoint, the change in
$\Phi$ is
\[
    2x-(2\beta+1)\lambda\le-\lambda.
\]
If both are below it, the change is at most
\[
    -\bigl(c_G(A')-(2\beta+1)\lambda\bigr)<-\lambda.
\]
In the mixed case, orient the split so that $U$ is the child below the
breakpoint.  The change is then $x-y$.  If $x\le\rho\lambda$, the
global minimum-cut lower bound gives $x+y=c_G(U)\ge\lambda$, and hence
\[
    x-y\le2x-\lambda\le-\lambda/2.
\]
Otherwise $U$ is boundary-sparse.  Since the first case is unavailable,
$x>\rho\lambda$, and boundary sparsity gives
\[
    x-y
    <
    x-\frac{x}{1-\varepsilon}
    =
    -\frac{\varepsilon x}{1-\varepsilon}
    <
    -\frac{\varepsilon\rho}{1-\varepsilon}\lambda.
\]
Thus every split lowers $\Phi$ by
$\Omega(\varepsilon\lambda)$, so there are
\[
    O\!\left(\frac{c_G(A)}{\varepsilon\lambda}\right)
\]
first-stage splits.  Every split has induced value at most
$\beta\lambda$, and therefore
\[
    \sum_{B\in\mathcal D_{\mathrm{final}}}c_G(B)
    =
    c_G(A)+2\sum_{\text{splits}}c_{G[A']}(U)
    =
    O_\alpha(\varepsilon^{-1}c_G(A)).
\]

For every finished piece of boundary at most
$(2\beta+2)\lambda$, continue splitting along induced cuts of value at
most $\rho\lambda$ until none remains.  On the descendants of a proper
piece, define
\[
    \Psi(\mathcal D)
    :=
    \sum_{B\in\mathcal D}
    \bigl(c_G(B)-2\rho\lambda\bigr).
\]
Every such $B$ is a nontrivial cut of $G$, so
$c_G(B)\ge\lambda$ and $\Psi\ge0$.  A split of induced value
$x\le\rho\lambda$ changes $\Psi$ by
\[
    2x-2\rho\lambda\le0.
\]
Furthermore,
\[
    c_G(B)
    \le
    \frac{c_G(B)-2\rho\lambda}{1-2\rho}.
\]
Thus this cleanup increases the total boundary by at most the constant
factor $(1-2\rho)^{-1}$, and every output piece has induced minimum cut
greater than $\rho\lambda$.  If the piece is all of $V(G)$, no first
cleanup split exists because
$\lambda_2(G)\ge\lambda>\rho\lambda$.  High-boundary pieces finished
in the first stage already have this induced connectivity and contain no
relevant boundary-sparse cut.

It remains to run the recursion of Lemma~6.4 on each cleaned
low-boundary piece.  Enumerate all its induced cuts of value at most
$\beta\lambda$ and maintain them exactly as in that lemma.  Put
$\sigma:=1/3$.  In case~(a), split along a maintained cut of induced
value at most $\sigma\lambda$.  If its value is $x$, the global
minimum-cut lower bound shows that each child boundary is at most the
parent boundary minus
\[
    a\lambda,
    \qquad
    a:=1-2\sigma>0.
\]
In case~(b), the chosen cut is boundary-sparse and has induced value
$x>\sigma\lambda$, so each child boundary is at most the parent
boundary minus
\[
    \delta_\varepsilon\lambda,
    \qquad
    \delta_\varepsilon
    :=
    \frac{\varepsilon\sigma}{1-\varepsilon}.
\]

After the pruning step of a recursive instance $A'$, define
\[
    \chi(A')
    :=
    \max_C c_{G[A']}(C\cap A'),
    \qquad
    q:=1-\frac{\sigma}{\beta},
\]
where the maximum ranges over the surviving maintained cuts, and set
$\chi(A'):=0$ if none survives.  As in Observation~6.8
of~\cite{HLRW24}, case~(b) selects a side $U$ with
$|U|\le|A'|/2$, and every maintained cut crossing $U$ also crosses
$A'\setminus U$.

Suppose first that $\chi(U)\ge q\chi(A')$, and let $C$ attain
$\chi(U)$.  Then
\[
\begin{aligned}
    c_{G[A'\setminus U]}(C\cap(A'\setminus U))
    &\le
    c_{G[A']}(C\cap A')-c_{G[U]}(C\cap U)\\
    &\le
    \chi(A')-\chi(U)
    \le
    (1-q)\chi(A')
    \le
    \sigma\lambda.
\end{aligned}
\]
Thus the complementary recursive instance immediately takes case~(a).
Otherwise $\chi(U)<q\chi(A')$.

For $c>\sigma\lambda$, define
\[
    h(c):=\min\{t:q^tc\le\sigma\lambda\}.
\]
Since $c\le\beta\lambda$, the maximum possible rank is
$O_\alpha(1)$.  Let $f_h(b,d)$ be the maximum possible number of
leaves of an instance whose boundary is at most $b$, whose size is at
most $d$, and whose rank is at most $h$.  The preceding observations
give
\[
    f_h(b,d)
    \le
    \max\left\{
        1,\;
        2f_h(b-a\lambda,d)+f_h(b,d/2),\;
        f_{h-1}(b,d)
          +f_h(b-\delta_\varepsilon\lambda,d)
    \right\}.
\]
For $h=0$, only case~(a) occurs, and
\[
    f_0(b,d)\le2^{O(b/\lambda)}.
\]
Repeating the double induction in Cases~1--3 of the proof of
Lemma~6.4, now over the $O_\alpha(1)$ possible ranks, gives
\[
    f_h(b,d)
    \le
    \left(
        1+\frac{C_\alpha b}{\varepsilon\lambda}
    \right)^{C_\alpha}
    (2+2\log_2d)^{C_\alpha(1+b/\lambda)}.
\]
Every input to this recursion has
$b\le(2\beta+2)\lambda$, and neither kind of split increases either
child boundary.  Hence each cleaned low-boundary piece contributes at
most
\[
    \varepsilon^{-O_\alpha(1)}
    (\log|A|)^{O_\alpha(1)}
\]
times its own boundary to the final sum.  Combining this with the
$O_\alpha(\varepsilon^{-1}c_G(A))$ first-stage boundary bound, and
retaining the finished high-boundary pieces unchanged, proves
\eqref{eq:small_boundary_concise}.  Theorem~6.2 uses
$|A|^{O_\alpha(1)}$ time for the constant approximation parameter
$\beta/\rho$, and all recursive pieces form a laminar family of size
$O(|A|)$.  Thus the entire modified routine takes
$|A|^{O_\alpha(1)}$ time, proving the small-cluster claim.

\smallskip
\noindent\emph{Large-cluster claim.}
The routine of Lemma~6.10 of~\cite{HLRW24} can be modified to partition
a cluster $A$ into $A_0,A_1,\ldots,A_r$ such that
\begin{equation}
\label{eq:large_claim_concise}
    \operatorname{vol}_G^W(A_i)
    \le
    \frac{C_\alpha s_0^2}{\varepsilon\lambda}
    \quad(i\ge1),
    \qquad
    w_G(A_0,A_1,\ldots,A_r)
    \le
    C_\alpha\varepsilon^{-1}c_G(A).
\end{equation}
Moreover, every $U\subseteq A$ satisfying
\[
    c_{G[A]}(U)\le(\beta-1)\lambda,
    \qquad
    \operatorname{vol}_G^W(U)\le s_0
\]
admits a set $U^*\subseteq U$, disjoint from $A_0$, such that
\begin{equation}
\label{eq:large_uncrossing_concise}
    w_G(U\setminus U^*,V(G)\setminus A)
      +c_{G[A]}(U^*)
    \le
    (1+\varepsilon)c_{G[A]}(U).
\end{equation}
The routine runs in
\[
    \widetilde O\!\left(
        \left(
            \frac{s_0}{\varepsilon\lambda}
        \right)^{O_\alpha(1)}
        |E(G[A])|
    \right)
\]
time.

We describe the necessary changes to Section~6.2 of~\cite{HLRW24}.
Retain the thresholds
$0.1\lambda$, $0.2\lambda$, and $0.4\lambda$ used there.  In the
analogue of Lemma~6.12, set
\[
    \varpi:=\frac{\lambda^3}{50s_0^2}
\]
and replace every capacity $c(e)$ by
$\lfloor c(e)/\varpi\rfloor$.

If \(U\) is a set under consideration, then
\[
    |U|
    \le \frac{\operatorname{vol}_G^W(U)}{\lambda}
    \le \frac{s_0}{\lambda}.
\]
Since the support of \(G\) is simple, every cut of \(G[U]\) crosses at
most \(|U|^2/4\) support edges.  Rounding therefore decreases the
capacity of any such cut by at most
\[
    \frac{|U|^2}{4}\varpi
    \le
    \frac{s_0^2}{4\lambda^2}\varpi
    =
    \frac{\lambda}{200}.
\]
Consequently, an induced subgraph whose minimum
cut is at least $0.1\lambda$ has rounded minimum cut at least
$\lambda/(25\varpi)$, as required by Lemma~6.12.

Lemmas~6.13 and~6.14 are unchanged.  In the analogue of Lemma~6.15, a
set $U$ under consideration satisfies
\[
    |\delta_H(U)|
    \le
    \frac{\beta\lambda}{\varpi}.
\]
Each rounded edge belongs to at most $100\log \bar m$ packed forests, while
the packing contains at least
\[
    \frac{\lambda}{25\varpi}\log \bar m
\]
forests.  Hence some forest crosses $U$ at most $2500\beta$ times.
We may therefore replace $2525$ throughout Lemmas~6.15 and~6.17 by
\[
    M_\alpha:=\lceil2500\beta\rceil+1.
\]
The resulting version of Lemma~6.11 produces sets of volume at most
$3s_0$, has the same polynomial overlap bound, and covers every
internally $0.1\lambda$-connected subset of a set $U$ satisfying
\[
    c_{G[A]}(U)\le\beta\lambda,
    \qquad
    \operatorname{vol}_G^W(U)\le s_0
\]
by at most $M_\alpha$ such sets.

We next adjust Lemma~6.16.  Partition $U$ iteratively along internal
cuts of value below $0.1\lambda$, obtaining internally
$0.1\lambda$-connected sets $U_1,\ldots,U_k$.  Their total interpart
capacity is at most $0.1(k-1)\lambda$.  If more than
$k/2+6\beta$ of them had boundary greater than $0.4\lambda$ in
$G[A]$, then
\[
\begin{aligned}
    w_G(U_1,\ldots,U_k,A\setminus U)
    &=
    \frac12\left(
        \sum_{i=1}^k c_{G[A]}(U_i)+c_{G[A]}(U)
    \right)\\
    &>
    0.2\left(\frac{k}{2}+6\beta\right)\lambda\\
    &=
    0.1k\lambda+1.2\beta\lambda.
\end{aligned}
\]
On the other hand,
\[
    w_G(U_1,\ldots,U_k,A\setminus U)
    \le
    0.1(k-1)\lambda+\beta\lambda,
\]
a contradiction.  Repeating the same argument after intersecting all
sets with an active set $A'\subseteq A$ shows that, among
\[
    I:=\{i:U_i\cap A'\ne\emptyset\},
\]
at least $|I|/2-6\beta$ indices satisfy
\[
    c_{G[A']}(U_i\cap A')\le0.4\lambda.
\]
Thus the analogue of Lemma~6.19 decreases $|I|$ by a constant factor
whenever
\[
    |I|>K_\alpha,
    \qquad
    K_\alpha:=\lceil24\beta\rceil.
\]

We now choose the remaining parameters in the order needed by the
algorithm.  The modified Lemma~6.17 enumerates connected families of at
most $M_\alpha K_\alpha$ sets from Lemma~6.11.  Hence every resulting
source set has volume at most
\[
    B_\alpha s_0,
    \qquad
    B_\alpha:=3M_\alpha K_\alpha.
\]
Choose sufficiently large constants $R_\alpha\ge B_\alpha$ and
$D_\alpha$, and set
\begin{equation}
\label{eq:large_gamma}
    \gamma
    :=
    \frac{\varepsilon\lambda}{D_\alpha s_0}.
\end{equation}
Apply Lemma~6.17 with parameters $(\gamma,K_\alpha)$.  It produces a
source family $\mathcal S$ in which every vertex belongs to at most
\[
    P_\alpha
    :=
    \left(
        \frac{s_0}{\varepsilon\lambda}
    \right)^{O_\alpha(1)}
    \log^{O_\alpha(1)}n
\]
source sets, in time
\[
    \widetilde O\!\left(
        P_\alpha|E(G[A])|
    \right).
\]
The relevant coverage guarantee is the following.  Suppose
$U'\subseteq U\subseteq A$,
\[
    c_{G[A]}(U)\le\beta\lambda,
    \qquad
    \operatorname{vol}_G^W(U)\le s_0,
\]
the minimum cut of $G[U']$ is at least $\gamma\lambda$, and $U'$
is covered by at most $K_\alpha$ internally
$0.1\lambda$-connected subsets of $U$.  Then some
$S\in\mathcal S$ contains $U'$.

Every source has at most $B_\alpha s_0/\lambda$ vertices, and every
vertex belongs to at most $P_\alpha$ sources.  The intersection graph
of $\mathcal S$ therefore has maximum degree at most
\[
    \frac{B_\alpha s_0}{\lambda}P_\alpha.
\]
A greedy coloring partitions $\mathcal S$ into
\[
    \left(
        \frac{s_0}{\varepsilon\lambda}
    \right)^{O_\alpha(1)}
    \log^{O_\alpha(1)}n
\]
families, each consisting of pairwise vertex-disjoint source sets.

For completeness, we record the required modification of Lemma~6.18.
Let $\zeta>0$ be a sufficiently small absolute constant multiple of
$\varepsilon$.  Given pairwise disjoint source sets
$S_1,\ldots,S_q$ in the current active cluster $A'$, construct the
auxiliary isolating-cut graph as in~\cite{HLRW24}, except that
\[
    w(t_i,v)
    :=
    (1-2\zeta)w_G(v,V(G)\setminus A')
    \qquad(v\in S_i)
\]
and
\[
    w(t_*,v)
    :=
    \frac{\zeta\lambda}{R_\alpha s_0}\,d_G^W(v)
    \qquad(v\in A').
\]
Run the approximate-isolating-cut algorithm with accuracy
\[
    \zeta'
    \le
    \min\left\{
        \zeta,\,
        a_\alpha\frac{\zeta\lambda}{s_0}
    \right\},
\]
where $a_\alpha>0$ is sufficiently small.  The auxiliary graph has
conductance
\[
    \Omega_\alpha\!\left(
        \frac{\zeta\lambda}{s_0}
    \right).
\]
The proof of Lemma~6.18 therefore returns pairwise disjoint sets
$C_i\subseteq A'$ satisfying
\begin{equation}
\label{eq:modified_isolating_volume_boundary}
\begin{split}
    \operatorname{vol}_G^W(C_i)
    &\le
    \frac{C_\alpha s_0^2}{\varepsilon\lambda},\\
    c_{G[A']}(C_i)
    &\le
    (1-\zeta)w_G(C_i,V(G)\setminus A').
\end{split}
\end{equation}
Indeed, the singleton cut around $t_i$ has capacity at most
$\operatorname{vol}_G^W(S_i)\le B_\alpha s_0$, whereas every vertex
of $C_i$ contributes
\[
    \frac{\zeta\lambda}{R_\alpha s_0}d_G^W(v)
\]
to the regularizing terminal.  This proves the first inequality after
increasing $C_\alpha$.  The second follows from the flow guarantee
exactly as in the original proof:
\[
    c_{G[A']}(C_i)
    \le
    \frac{1-2\zeta}{1-\zeta'}
    w_G(C_i\cap S_i,V(G)\setminus A')
    \le
    (1-\zeta)w_G(C_i,V(G)\setminus A').
\]

Furthermore, whenever
\[
    U^\dagger\subseteq S_i,
    \qquad
    \operatorname{vol}_G^W(U^\dagger)\le s_0,
    \qquad
    c_{G[A']}(U^\dagger)\le\beta\lambda,
\]
there is $U^*\subseteq U^\dagger\cap C_i$ such that
\begin{equation}
\label{eq:modified_isolating_uncrossing}
    w_G(U^\dagger\setminus U^*,V(G)\setminus A')
      +c_{G[A']}(U^*)
    \le
    \left(1+\frac{\varepsilon}{2}\right)
    c_{G[A']}(U^\dagger).
\end{equation}
Moreover,
\begin{equation}
\label{eq:modified_isolating_change}
    U^*\ne U^\dagger
    \quad\Longrightarrow\quad
    c_{G[A']}(U^\dagger)\ge0.2\lambda.
\end{equation}
To see this, repeat Equations~(7)--(10) in the proof of Lemma~6.18.
Before the final case distinction, the two additional errors are
\[
    \zeta' O_\alpha(s_0)
    \qquad\text{and}\qquad
    O\!\left(\frac{\zeta\lambda}{R_\alpha}\right).
\]
The choices of $a_\alpha$ and $R_\alpha$ make their sum at most
$\zeta\lambda/2$.  If
$c_{G[A']}(U^\dagger)\ge0.2\lambda$, this error is absorbed
multiplicatively; choosing the constant in
$\zeta=\Theta(\varepsilon)$ sufficiently small gives
\eqref{eq:modified_isolating_uncrossing}.  If the boundary is below
$0.2\lambda$, the global minimum-cut lower bound forces
$U^*=U^\dagger$, exactly as in Equation~(10) of~\cite{HLRW24}.  This
proves~\eqref{eq:modified_isolating_change}.

Run these isolating-cut instances cyclically through the disjoint source
families, always deleting the output sets $C_i$ from the active
cluster.  The generalized Lemma~6.19 and the preceding calculation imply
that after $O(\log n)$ cycles, at most $K_\alpha$ of the pieces
$U_i$ still meet the active cluster.  Perform one additional cycle and
let $A_0$ be the remaining active set.  The deleted sets form
$A_1,\ldots,A_r$.

The volume conclusion in~\eqref{eq:large_claim_concise} follows from
\eqref{eq:modified_isolating_volume_boundary}.  To prove its boundary
conclusion, consider one batch of deleted sets and let $A'$ be the
active cluster immediately before the batch.  Put
\[
    x:=\sum_i w_G(C_i,V(G)\setminus A'),
    \qquad
    y:=\sum_i c_{G[A']}(C_i).
\]
The boundary of the active cluster decreases by at least $x-y$, while
the new interpart capacity created by the batch is at most $y$.  Since
$y\le(1-\zeta)x$,
\[
    y\le\frac1\zeta(x-y).
\]
Summing over all batches telescopes against the decrease of
$c_G(A')$, and hence
\[
    w_G(A_0,A_1,\ldots,A_r)
    \le
    \zeta^{-1}c_G(A)
    =
    O(\varepsilon^{-1})c_G(A).
\]

It remains to prove the uncrossing guarantee.  Fix $U\subseteq A$ as
in the claim, and let $A^\dagger$ be the active cluster after the first
$O(\log n)$ cycles.  Refine $U\cap A^\dagger$ iteratively along
internal cuts of value below $\gamma\lambda$, and let $\mathcal Q$
be the resulting partition.  Since
\[
    |U|
    \le
    \frac{\operatorname{vol}_G^W(U)}{\lambda}
    \le
    \frac{s_0}{\lambda},
\]
the refinement makes fewer than $s_0/\lambda$ cuts, and therefore
\begin{equation}
\label{eq:large_refinement_weight}
    w_G(\mathcal Q)
    \le
    \frac{s_0}{\lambda}\gamma\lambda
    \le
    \frac{\varepsilon\lambda}{D_\alpha}.
\end{equation}
Every $Q\in\mathcal Q$ has minimum cut at least
$\gamma\lambda$, is covered by at most $K_\alpha$ of the surviving
internally $0.1\lambda$-connected sets, and satisfies
\[
\begin{aligned}
    c_{G[A^\dagger]}(Q)
    \le
    c_{G[A^\dagger]}(U\cap A^\dagger)
      +w_G(\mathcal Q)
    \le
    (\beta-1)\lambda
      +\frac{\varepsilon\lambda}{D_\alpha}
    <
    \beta\lambda.
\end{aligned}
\]
Thus the modified Lemmas~6.17 and~6.18 apply to every part during the
additional cycle.

For each $Q_i\in\mathcal Q$, let $A_i^*\supseteq A_0$ be the active
cluster when its source family is processed, and let
$U_i^*\subseteq Q_i\cap A_i^*$ be the set obtained from
\eqref{eq:modified_isolating_uncrossing}.  Put
\[
    D_i:=(Q_i\cap A_i^*)\setminus U_i^*,
    \qquad
    D:=\bigcup_iD_i,
    \qquad
    \widehat U:=U\setminus D.
\]
Every vertex of $\widehat U$ lies in a deleted set, so
$\widehat U\cap A_0=\emptyset$.  Repeating the edge accounting in
Equations~(1)--(6) of the proof of Lemma~6.10 gives
\begin{equation}
\label{eq:large_calculation_concise}
    c_{G[A]}(\widehat U)-c_{G[A]}(U)
    \le
    \frac{\varepsilon}{2}c_{G[A]}(U)
      -w_G(D,V(G)\setminus A)
      +C_0w_G(\mathcal Q),
\end{equation}
where $C_0$ is an absolute constant.

Choose $D_\alpha\ge20C_0$.  If
$c_{G[A]}(U)\ge\lambda/10$, then
\eqref{eq:large_refinement_weight} gives
\[
    C_0w_G(\mathcal Q)
    \le
    \frac{\varepsilon}{2}c_{G[A]}(U),
\]
and \eqref{eq:large_calculation_concise} proves
\eqref{eq:large_uncrossing_concise} with $U^*:=\widehat U$.

Suppose instead that $c_{G[A]}(U)<\lambda/10$.  If some $D_i$ were
nonempty, \eqref{eq:modified_isolating_change} would give
\[
    c_{G[A_i^*]}(Q_i\cap A_i^*)\ge0.2\lambda.
\]
The comparison following Equation~(6) in~\cite{HLRW24}, together with
\eqref{eq:large_refinement_weight}, would give
\[
    c_{G[A_i^*]}(Q_i\cap A_i^*)
    \le
    c_{G[A]}(U)+2w_G(\mathcal Q)
    <
    0.2\lambda,
\]
a contradiction after increasing $D_\alpha$ if necessary.  Hence
$D=\emptyset$, so $\widehat U=U$, and the desired inequality is
immediate.  This proves~\eqref{eq:large_uncrossing_concise}.

Lemma~6.17 is run once, and the modified Lemma~6.18 is run for
\[
    \left(
        \frac{s_0}{\varepsilon\lambda}
    \right)^{O_\alpha(1)}
    \log^{O_\alpha(1)}n
\]
source-family/cycle pairs.  Each invocation costs the same polynomial
factor times $|E(G[A])|$.  This proves the asserted running time and
completes the large-cluster claim.

We now return to the notation $G_j$.  Combining the strong partition
with the two claims gives the following constant-range version of
Lemma~3.2 of~\cite{HLRW24}.  For every
$S\subseteq V(G_j)$ satisfying
\[
    c_{G_j}(S)\le(\beta-1)\lambda,
\]
there is a cut $S'$ crossing no output cluster such that
\begin{equation}
\label{eq:structure_concise}
\begin{split}
    c_{G_j}(S')
    &\le
    (1+c_\alpha\varepsilon)c_{G_j}(S),\\
    \operatorname{vol}_{G_j}^W(S\triangle S')
    &\le
    \frac{c_\alpha s_0^2}{\varepsilon\lambda}.
\end{split}
\end{equation}

Indeed, in every crossed strong cluster, choose the shore of smaller
padded volume; its unpadded volume is at most $s_0$.  Apply the
large-cluster claim to that shore.  If the resulting set differs from the
original shore, then applying the global minimum-cut lower bound to their
difference and using~\eqref{eq:large_uncrossing_concise} gives
\[
    \lambda
    \le
    (2+\varepsilon)c_{G_j[A]}(S\cap A).
\]
The induced boundaries for distinct clusters are disjoint subsets of
$\delta_{G_j}(S)$.  Hence only $O_\alpha(1)$ large clusters change,
and their total changed volume is $O_\alpha(s_0)$.

After this step, the cut crosses only the small pieces produced by the
large-cluster routine.  Each has volume at most
\[
    \frac{C_\alpha s_0^2}{\varepsilon\lambda}.
\]
For sufficiently small $\varepsilon$,
\[
    (1+\varepsilon)(\beta-1)\lambda<\beta\lambda,
\]
so the small-cluster claim applies to every induced cut encountered
here.  The uncrossing argument of Claims~3.6--3.7 of~\cite{HLRW24}
changes only $O_\alpha(1)$ groups: every changed group contributes
$\Omega(\lambda)$ disjoint induced boundary.  Thus the additional
changed volume is
\[
    O_\alpha\!\left(
        \frac{s_0^2}{\varepsilon\lambda}
    \right).
\]
The large-cluster uncrossing increases the cut by at most an
$O(\varepsilon)$ fraction, and the small-cluster uncrossing contributes
a factor $1/(1-\varepsilon)$.  Increasing $c_\alpha$ if necessary
gives both conclusions in~\eqref{eq:structure_concise}.

The preceding local argument is uniform for every sufficiently small
$\varepsilon>0$.  Let $\varepsilon_\alpha>0$ be smaller than all
constant upper bounds on $\varepsilon$ used above, and in particular
assume
\[
    \varepsilon_\alpha\le\frac1{4\beta}.
\]
Set
\[
    L_0:=\left\lceil\log_2(2\beta \bar m)\right\rceil+2
    =O(\log n)
\]
and fix
\[
    \varepsilon
    :=
    \min\left\{
        \varepsilon_\alpha,\,
        \frac{\theta}{4c_\alpha\alpha L_0}
    \right\}.
\]
We now construct the hierarchy with this value of $\varepsilon$.

Choose the exponent in~\eqref{eq:s0_concise} sufficiently large.  The
calculation of Lemma~3.1 of~\cite{HLRW24}, using the intercluster-boundary
bounds proved above, makes the unpadded volume of each quotient at most
half that of its predecessor.  By~\eqref{eq:initial_volume_concise} and
the lower bound $2\lambda$ on the volume of a nontrivial quotient, the
hierarchy has at most $L_0$ levels.  Moreover,
\begin{equation}
\label{eq:hierarchy_volume_concise}
    \sum_{j=0}^{L-1}
    \operatorname{vol}_{G_j}^W(V(G_j))
    =
    O_\alpha(\bar m\lambda).
\end{equation}

The rounded graphs used in Lemma~6.12 have, over all calls,
\[
    O\!\left(
        \sum_j
        \frac{\operatorname{vol}_{G_j}^W(V(G_j))}{\varpi}
    \right)
    =
    \bar m\left(\frac{s_0}{\lambda}\right)^{O_\alpha(1)}
\]
explicit edges, and the forest-packing overhead is
$(s_0/\lambda)^{O_\alpha(1)}$.  All remaining routines are charged to
the simple supports of the quotients.  The small pieces have at most
\[
    O_\alpha\!\left(
        \frac{s_0^2}{\varepsilon\lambda^2}
    \right)
    =
    (\log n/\theta)^{O_\alpha(1)}
\]
vertices by~\eqref{eq:quotient_connectivity_concise} and
\eqref{eq:large_claim_concise}.  Thus the entire hierarchy is constructed
in
\[
    \bar m(\log n/\theta)^{O_\alpha(1)}
\]
time.

Fix a target cut $S$.  Starting with $S_0=S$, apply
\eqref{eq:structure_concise} before each contraction.  For every
$j\le L$,
\[
\begin{aligned}
    c_{G_j}(S_j)
    \le
    \alpha\lambda(1+c_\alpha\varepsilon)^j
    \le
    \alpha\lambda e^{c_\alpha\varepsilon L_0}
    \le
    \alpha\lambda e^{\theta/(4\alpha)}
    \le
    (\alpha+\theta)\lambda
    <
    (\beta-1)\lambda.
\end{aligned}
\]
Here the penultimate inequality holds after decreasing the fixed upper
bound on $\theta$.  Thus the structure lemma remains applicable at
every nontrivial level.

The proof of Lemma~7.3 of~\cite{HLRW24} then represents $S$ as
\[
    S=\mathop{\triangle}_{v\in D}\overline v
\]
for hierarchy nodes $D$, where $\overline v\subseteq V(F)$ is the
preimage of $v$, and their total unpadded degree is at most
\begin{equation}
\label{eq:Delta_concise}
    \Delta_\alpha
    :=
    \frac{C_\alpha L_0s_0^2}{\varepsilon\lambda}
    =
    \lambda(\log n/\theta)^{O_\alpha(1)}.
\end{equation}
The same hierarchy and the same bound work simultaneously for every
target cut.

Set
\[
    Z:=\frac{\theta\lambda}{\Delta_\alpha},
    \qquad
    \nu:=\frac{\lambda}{Z},
    \qquad
    \eta:=\theta Z^2.
\]
The sampling construction underlying Lemma~7.1 of~\cite{HLRW24}
allows the common edge quantum to be fixed to any sufficiently small
constant multiple of its stated upper bound.  Fixing such a constant,
we take
\begin{equation}
\label{eq:W0_concise}
\begin{aligned}
    W_0
    =
    \Theta\!\left(
        \frac{\eta\lambda^2}
             {\nu\log(2L_0 \bar m)}
    \right)
    =
    \Theta\!\left(
        \frac{\theta Z^3\lambda}
             {\log(2L_0 \bar m)}
    \right)
    =
    \lambda
    \left(\frac{\theta}{\log n}\right)^{O_\alpha(1)}.
\end{aligned}
\end{equation}
for a sufficiently small absolute constant.
It produces an unweighted graph $J$, supported on the original edge
pairs, with
\[
    |E(J)|=\bar m(\log n/\theta)^{O_\alpha(1)}.
\]
By Lemma~7.2 of~\cite{HLRW24} and
\eqref{eq:Delta_concise}, every target cut satisfies
\begin{equation}
\label{eq:J_concise}
\begin{aligned}
    \left|W_0c_J(S)-c_F(S)\right|
    \le
    \left(
        \frac{\Delta_\alpha}{\lambda}
    \right)^2
    \eta\lambda
    =
    \theta^3\lambda.
\end{aligned}
\end{equation}

Let $X$ be the extension graph of Section~7.2 of~\cite{HLRW24}, formed
from the unpadded hierarchy degrees.  For a terminal cut
$S\subseteq V(F)$, write
\[
    d(S)
    :=
    \min\left\{
        \sum_{v\in D}d^W(v):
        S=\mathop{\triangle}_{v\in D}\overline v
    \right\}.
\]
Lemma~7.4 of~\cite{HLRW24} identifies $d(S)$ with the minimum value of
$c_X(S^\star)$ over all extensions
$S^\star\subseteq V(X)$ satisfying
$S^\star\cap V(F)=S$.

Form $G'$ by assigning weight $W_0$ to every edge of $J$ and
multiplying every edge weight of $X$ by $Z$.  If $S$ is a target
cut, then $d(S)\le\Delta_\alpha$, so it has an extension $S^\star$
satisfying
\begin{equation}
\label{eq:target_extension_concise}
\begin{aligned}
    c_{G'}(S^\star)
    =
    W_0c_J(S)+Zc_X(S^\star)
    \le
    c_F(S)+\theta^3\lambda+Z\Delta_\alpha
    =
    c_F(S)+(\theta+\theta^3)\lambda.
\end{aligned}
\end{equation}

Conversely, let $S$ be any nontrivial terminal cut and let
$S^\star$ be any extension.  If $d(S)\ge\lambda/Z$, Lemma~7.4 gives
\[
    Zc_X(S^\star)\ge\lambda.
\]
If $d(S)<\lambda/Z=\nu$, apply Lemma~7.2 to a representation attaining
$d(S)$.  Since $c_F(S)\ge\lambda$,
\begin{equation}
\label{eq:extension_lower_concise}
\begin{aligned}
    W_0c_J(S)
    \ge
    c_F(S)
      -\left(
          \frac{d(S)}{\lambda}
       \right)^2\eta\lambda
    \ge
    c_F(S)-Z^{-2}\eta\lambda
    =
    c_F(S)-\theta\lambda
    \ge
    (1-\theta)\lambda.
\end{aligned}
\end{equation}
Thus every extension of every nontrivial terminal cut has value at least
$(1-\theta)\lambda$ in $G'$.

Every $J$-edge has weight $W_0$ in $G'$.  Every edge of $X$
corresponding to a node $v\in V(G_j)$ has unscaled weight
$d_{G_j}^W(v)\ge\lambda$, and hence has weight at least
$Z\lambda$ in $G'$.  Set
\[
    W:=\theta\min\{W_0,Z\lambda\}
\]
and replace every edge $e$ of $G'$ by
\[
    \left\lfloor\frac{w_{G'}(e)}W\right\rfloor
\]
parallel unit edges.  Denote the resulting unweighted extension graph by
$H$.  Since $w_{G'}(e)\ge W/\theta$,
\[
    (1-\theta)w_{G'}(e)
    \le
    W\left\lfloor\frac{w_{G'}(e)}W\right\rfloor
    \le
    w_{G'}(e).
\]
Consequently, for every $Y\subseteq V(G')$,
\begin{equation}
\label{eq:rounding_cut_comparison}
    (1-\theta)c_{G'}(Y)
    \le
    Wc_H(Y)
    \le
    c_{G'}(Y).
\end{equation}
It follows that the Steiner minimum cut of $H$, with terminal set
$V(F)$, is at least
\[
    \frac{(1-\theta)^2\lambda}{W}.
\]
On the other hand, \eqref{eq:target_extension_concise} gives an extension
of every target cut satisfying
\[
    Wc_H(S^\star)
    \le
    c_F(S)+(\theta+\theta^3)\lambda.
\]

The graph $H$ has
\[
    |V(H)|=O(nL_0),
    \qquad
    |E(H)|=\bar m(\log n/\theta)^{O_\alpha(1)}.
\]
Indeed, its $J$-part contributes at most
\[
    |E(J)|\frac{W_0}{W}
\]
copies.  The total unscaled weight of $X$ is
\[
    \sum_{j=0}^{L-1}
    \operatorname{vol}_{G_j}^W(V(G_j))
    =
    O_\alpha(\bar m\lambda)
\]
by~\eqref{eq:hierarchy_volume_concise}, so its scaled copy count is at
most
\[
    O_\alpha\!\left(\frac{\bar m Z\lambda}{W}\right).
\]
Both $W_0/W$ and $Z\lambda/W$ are
$(\log n/\theta)^{O_\alpha(1)}$ by
\eqref{eq:Delta_concise} and~\eqref{eq:W0_concise}.

Apply Theorem~7.7 of~\cite{HLRW24}, with terminals $V(F)$, to split
off all Steiner vertices of $H$.  The preceding upper and lower bounds
show that the Steiner minimum cut of $H$ is
$\Theta(\lambda/W)$.  Since $\bar m\ge n-1$, the running time is
\[
    \widetilde O\!\left(
        |E(H)|+|V(H)|(\lambda/W)^2
    \right)
    =
    \bar m(\log n/\theta)^{O_\alpha(1)}.
\]
Delete the isolated Steiner vertices afterward and call the resulting
unweighted multigraph on $V(F)$ $K_0$.

Splitting off cannot increase a cut and preserves the Steiner minimum
cut.  For a target cut $S$, let $S^\star$ be the extension from
\eqref{eq:target_extension_concise}.  Since every Steiner vertex is
isolated after splitting,
\[
\begin{aligned}
    Wc_{K_0}(S)
    =
    Wc_{\widehat K}(S^\star)
    \le
    Wc_H(S^\star)
    \le
    c_F(S)+(\theta+\theta^3)\lambda,
\end{aligned}
\]
where $\widehat K$ denotes the graph immediately after splitting off.
Moreover, every nontrivial cut of $K_0$ separates terminals and is
therefore at least the preserved Steiner minimum cut.  After increasing
a constant $C_\alpha$ if necessary, we obtain
\begin{equation}
\label{eq:K0_bounds}
    Wc_{K_0}(S)\ge(1-C_\alpha\theta)\lambda
\end{equation}
for every nontrivial cut, while every target cut satisfies
\[
    Wc_{K_0}(S)
    \le
    c_F(S)+C_\alpha\theta\lambda.
\]

Finally, set
\[
    r
    :=
    \left\lceil
        \frac{(\alpha+2C_\alpha\theta)\lambda}{W}
    \right\rceil+1
\]
and let $K\subseteq K_0$ be an $r$-edge-connectivity certificate
from Nagamochi--Ibaraki~\cite{NI92}.  It satisfies
\[
    c_K(S)\ge\min\{r,c_{K_0}(S)\}
\]
for every cut and preserves every $K_0$-cut of size below $r$
exactly.  For every target cut,
\[
    c_{K_0}(S)
    \le
    \frac{(\alpha+C_\alpha\theta)\lambda}{W}
    <
    r,
\]
so its value is preserved.  Furthermore,
\[
    rW\ge\lambda,
\]
and therefore \eqref{eq:K0_bounds} and the certificate inequality imply
\[
    Wc_K(S)\ge(1-C_\alpha\theta)\lambda
\]
for every nontrivial cut.  Finally,
\[
    |E(K)|
    =
    O(nr)
    =
    n(\log n/\theta)^{O_\alpha(1)}.
\]

Equations~\eqref{eq:Delta_concise} and~\eqref{eq:W0_concise}, together
with the definition of $W$, give
\[
    W_0,W\le\lambda,
    \qquad
    \max\left\{
        \frac{\lambda}{W_0},
        \frac{\lambda}{W}
    \right\}
    \le
    \left(\frac{\log n}{\theta}\right)^{O_\alpha(1)}.
\]
Choose
\[
    \theta=c_\alpha'\xi
\]
for a sufficiently small positive constant $c_\alpha'$.  Then
\[
    \theta^3\le\xi,
    \qquad
    C_\alpha\theta\le\xi.
\]
Since the initial truncation preserved every target cut and
$\lambda_2(F)$, all guarantees hold for the original input graph.
Replacing $\theta$ by its constant multiple of $\xi$ in the running
time and scale bounds proves the theorem.
\end{proof}

\section{Consequences for \texorpdfstring{$3$}{3}-Cut and \texorpdfstring{$4$}{4}-Cut on Simple Graphs}
\label[appendix]{app:simple_graph_consequences}

We apply the border-and-island framework of~\cite{LV26} with the weighted algorithm above on the contracted multigraph.  We use the notation of that paper and write $p:=2.37134$ and $a:=3.250386$, the rounded upper bounds from \cite{LV26}.

\begin{theorem}[Minimum $3$-Cut and $4$-Cut in simple graphs]
\label{thm:simple_graph_consequences}
Minimum $3$-Cut in a simple unweighted $n$-vertex graph can be solved
with high probability in $O(n^2)$ time.  Minimum $4$-Cut can be solved with high
probability in $\widetilde O(n^{14/5})$ time.
\end{theorem}

\begin{proof}
Let $s=n^\theta$, suppressing polylogarithmic factors for $k\ge4$.  As in~\cite{LV26}, the small-$s$ branch has exponent $f(\theta):=1+(6k-6)\theta$, while the sparsified instance has $N=\widetilde O(n/s)$ vertices and $M=O(sn)$ edges.  The new weighted bound gives
\begin{equation*}
 \widetilde O\!\left(N^{k-3}(M+N^2)\right)
 =\widetilde O\!\left(\frac{n^{k-2}}{s^{k-4}}+\frac{n^{k-1}}{s^{k-1}}\right),
\end{equation*}
so the no-island exponent is $b_0(\theta):=\max\{k-2-(k-4)\theta,k-1-(k-1)\theta\}$.  The remaining exponents are
\begin{equation*}
 b_1=k-1-(k-2)\theta,\quad b_2=k-2-(k-4)\theta,\quad
 b_i=k-i+\widehat\Phi(i)-(k-i)\theta\ (3\le i\le k-2),\quad
 b_{k-1}=\widehat\Phi(k-1).
\end{equation*}
For every threshold below, $\theta<1/2$, and hence $b_0,b_2<b_1$.

For $k=4$, $f=1+18\theta$, $b_1=3-2\theta$, $b_0=\max\{2,3-3\theta\}$, $b_2=2$, and $b_3 = \widehat \Phi(3) < p$.  Thus the new threshold is determined by $f=b_1$, namely $\theta_4=1/10$, and
\begin{equation*}
 E_4=1+18\theta_4=3-2\theta_4=\frac{14}{5}>
 \max\{3-3\theta_4,2,p\}.
\end{equation*}

For $k=3$, retain the polylogarithmic factors and choose
$s=n^{1/24}$.  The FPT branch runs in
$\widetilde O(s^{12}n)=\widetilde O(n^{3/2})=O(n^2)$.
The two-island branch takes $O(n^2)$ time by checking all pairs, while the
no-island branch takes
$\widetilde O((n/s)^2)=O(n^2)$; the one-island branch is no slower up to polylogarithmic factors, as in \cite{LV26}, and thus runs in $O(n^2)$.
Thus Minimum $3$-Cut on simple graphs is solvable in $O(n^2)$ time.
\end{proof}

\clearpage
\section*{AI Methodology}

The author used OpenAI's GPT-5.6 Sol (Max effort) to assist with
drafting and revising portions of the exposition, including the
introduction and explanatory passages accompanying algorithms and
proofs throughout the technical sections. The author reviewed and
revised the AI-assisted material, checked the mathematical claims
and citations, and assumes responsibility for all content.

\bibliographystyle{alpha}
\bibliography{references}

\end{document}